\documentclass[hidelinks,a4paper,UKenglish]{lipics-v2018}

\newcommand{\short}[1]{}
\newcommand{\full}[1]{#1}

\usepackage{microtype}

\title{Up-To Techniques for Behavioural Metrics
  via~Fibrations}

\titlerunning{Up-To Techniques for Behavioural Metrics via~Fibrations}

\author{Filippo Bonchi}{Universit{\'a} di Pisa}{filippo.bonchi@unipi.it}{}{}
\author{Barbara K\"onig}{Universit\"at Duisburg-Essen}{barbara\_koenig@uni-due.de}{}{}
\author{Daniela Petri\c san}{CNRS, IRIF, Universit\'e Paris Diderot}{petrisan@irif.fr}{}{}

\authorrunning{F. Bonchi, B.  K\"onig, D. Petri\c san}

\Copyright{Filippo Bonchi, Barbara  K\"onig, Daniela Petri\c san}

\subjclass{
  \ccsdesc[500]{Theory of computation~Concurrency},
  \ccsdesc[500]{Theory of computation~Formal languages and automata theory},
  \ccsdesc[500]{Theory of computation~Logic and verification}
}

\keywords{behavioural metrics, bisimilarity, up-to techniques,
  coalgebras, fibrations}

\funding{The first author acknowledges financial support from project
  ANR-16-CE25-0011 REPAS, the second author from DFG project BEMEGA,
  and the third author from the European Research Council (ERC) under
  the European Union's Horizon 2020 research and innovation programme
  (grant agreement No.670624).
}

\acknowledgements{The authors are grateful to Shin-ya Katsumata,
  Henning Kerstan, Damien Pous and Paolo Baldan for precious
  suggestions and inspiring discussions\full{ and to Thomas Colcombet
    for the \kl{knowledge} package used for creating the hyperlinks in
    this paper}.}

\EventEditors{Sven Schewe and Lijun Zhang}
\EventNoEds{2}
\EventLongTitle{29th International Conference on Concurrency Theory
(CONCUR 2018)}
\EventShortTitle{CONCUR 2018}
\EventAcronym{CONCUR}
\EventYear{2018}
\EventDate{September 4--7, 2018}
\EventLocation{Beijing, China}
\EventLogo{}
\SeriesVolume{118}
\ArticleNo{17} 
\nolinenumbers 
\full{\hideLIPIcs}  

\usepackage[svgnames]{xcolor} 

\usepackage{thmtools}
\usepackage{thm-restate}
\usepackage{multicol}
\usepackage{rotating}
\usepackage{adjustbox}
\usepackage{wrapfig}

\usepackage[hyperref,notion,electronic]{knowledge}
\usepackage{amsmath,amssymb}
\usepackage{mathbbol}
\usepackage{enumitem}
\usepackage{stmaryrd}
\usepackage{stackrel}
\usepackage{aliascnt}
\usepackage{xypic}
\usepackage{tikz}
\usepackage{tikz-cd}
\usetikzlibrary{matrix}
\usetikzlibrary{positioning,arrows,decorations.markings}
\usetikzlibrary{shapes,arrows} 
\usetikzlibrary{positioning}
\usetikzlibrary{chains}
\usetikzlibrary{automata}
\usetikzlibrary{patterns}
\usetikzlibrary{decorations.pathreplacing}
\usetikzlibrary{calc,decorations.markings,decorations.pathreplacing,fit,backgrounds,shapes.symbols,shapes.geometric}
\tikzset{
  gnode/.style={draw,shape=circle,inner sep=0,minimum height=.2cm,
    minimum width=.2cm},
  hyperedge/.style={shape=rectangle,draw,inner sep=0,minimum width=.6cm,
    minimum height=.4cm}
}

\tikzset{every fit/.style={shape=rectangle,inner sep=5pt}}

\tikzset{
  mono/.style={>->},  
  ontop/.style={preaction={draw,-,line width=3pt,white}},
  arlab/.style={circle,inner sep=1pt,font=\scriptsize}
}

\newcommand\pseudopar[1]{{\sffamily\bfseries\raggedright{#1}}}
\usepackage[colorinlistoftodos,
textsize=footnotesize,color=orange!70]{todonotes}

\renewcommand{\todo}[1]{}

\full{
  \knowledgedirective{math}{notion}
\knowledgedirective{mathsymb}{math}

\knowledge{\VV}{mathsymb}
\knowledge{sound}{notion}
\knowledge{coalgebra for a functor}[$F$-coalgebra|coalgebra|coalgebras]{notion}

\knowledge{Algebras for the functor}[Algebra for a functor|algebra for
a functor|$F$-algebra|algebra for the functor|$F$-algebras|$T$-algebra]{notion}
\knowledge {behaviour functor}{notion}

\knowledge{carrier}{notion}
\knowledge{shortest-distinguishing-word-distance}{notion}

\knowledge {total category}{notion}
\knowledge {base category}{notion}
\knowledge {bialgebra}[bialgebras]{notion}
\knowledge {fibration}[fibrations]{notion}
\knowledge{fibre above $X$}[fibre|fibres]{notion}
\knowledge{fibrePred}{notion}
\knowledge{fibreRel}{notion}
\knowledge {Cartesian lifting}[Cartesian liftings]{notion}
\knowledge {split fibrations}{notion}
\knowledge {bifibration}[bifibrations]{notion}
\knowledge{direct images along $f$}[direct image]{notion}
\knowledge {lifting}{notion}
\knowledge {restriction of lifting}{notion}
\knowledge {reindexing functor}[reindexing]{notion}
\knowledge {fibred lifting}[fibred liftings]{notion}
\knowledge {isoFibre }{notion}
\knowledge{Beck-Chevalley}{notion}

\knowledge {quantale}[quantales]{notion}
\knowledge {$\VV $-valued predicate}[$\VV $-valued predicates]{notion}
\knowledge {$\VV $-valued relation}[$\VV $-valued relations]{notion}
\knowledge {morphism between $\VV $-valued predicates}{notion}
\knowledge {morphism between $\VV $-valued relations}[morphism between
  $\VV$-relations]{notion}
\knowledge {unital}{notion}
\knowledge {unit}[identity element]{notion}
\knowledge {commutative}{notion}
\knowledge {reflexive}[reflexivity|Reflexivity]{notion}
\knowledge {transitive}[transitivity]{notion}
\knowledge {symmetric}[symmetry]{notion}

\knowledge{evaluation maps}[evaluation map]{notion}
\knowledge{monotone evaluation maps}{notion}

\knowledge {Wasserstein lifting}[Wasserstein liftings]{notion}
\knowledge {isoFibre}{notion}
\knowledge {distributive law}{notion}
\knowledge{Hausdorff metric}[Hausdorff distance]{notion}

\knowledge{$b$-compatible}[compatible|compatibility]{notion}

\knowledge{\revOrd}{mathsymb}
\knowledge{\sdwDist}{mathsymb}
\knowledge{\bard}{mathsymb}
\knowledge{\liftF}{mathsymb}
\knowledge{\barF}{mathsymb}
\knowledge{\liftT}{mathsymb}
\knowledge{\Set}{mathsymb}
\knowledge {\cartLiftfR }{mathsymb}
\knowledge {\reindexf }{mathsymb}
\knowledge {\Vpred }{mathsymb}
\knowledge {\Vrel }{mathsymb}
\knowledge {\Vcat }{mathsymb}
\knowledge {\Vcatsym }{mathsymb}
\knowledge {\isoPR }{mathsymb}
\knowledge {\chbaseFunct }{mathsymb}

\knowledge{\lambdaChBase }{mathsymb}
\knowledge{constant predicate}{notion}
\knowledge{diagonal function}{notion}
\knowledge{sym}{mathsymb}
\knowledge{\uparr}{mathsymb}
\knowledge {\liftFcan }{mathsymb}
\knowledge {\liftFcanTwo }{mathsymb}
\knowledge {\true }{mathsymb}
\knowledge {\truer }{mathsymb}
\knowledge {\evcan }{mathsymb}
\knowledge {\objTruth }{mathsymb}
\knowledge{\distrLawNFA}{mathsymb}

\knowledge {\evNFAF }{mathsymb}
\knowledge {\evNFAT }{mathsymb}

\knowledge {\barTF }{mathsymb}
\knowledge {\barFT }{mathsymb}

\knowledge {\liftDistr }{mathsymb}
\knowledge {\barDistr }{mathsymb}
\knowledge {\barT }{mathsymb}
\knowledge {lambdaChBaseF}{mathsymb}
\knowledge {\orderFV }{mathsymb}
\knowledge {\distrLaw }{mathsymb}

\knowledge{knowledge}{url={https://ctan.org/pkg/knowledge?lang=en}}

\knowledge {diagonal relation}{notion}
\knowledge {composition}{notion}
\knowledge {\barFcan }{mathsymb}

\knowledge {behavioural closure}{notion}
\knowledge {up-to reflexivity}{notion}
\knowledge {up-to transitivity}{notion}
\knowledge {up-to symmetry}{notion}
\knowledge {up-to metric closure}{notion}
\knowledge {canonical evaluation map}{notion}
\knowledge {up-to contextual closure}{notion}
\knowledge {up-to convex closure}{notion}
\knowledge {\VVlll }{mathsymb}
\knowledge {constructively completely distributive}[cccd-lattice]{notion}     
\knowledge {constructively-convergent}{notion}
\knowledge {canonical Wasserstein lifting}{notion}
}

\short{\input{macros}}

\theoremstyle{plain}

\newtheorem{prop}{Proposition}
\newtheorem{lem}[prop]{Lemma}

\newtheorem{defi}[prop]{Definition}
\newtheorem{ex}[prop]{Example}

\newrobustcmd\VV{{\kl[\VV]{\mathcal{V}}}}
\newrobustcmd{\Vpred}{\kl[\Vpred]{\mathcal{V}\textrm{-}\mathsf{Pred}}}
\newrobustcmd{\Pred}{\mathsf{Pred}}
\newrobustcmd{\Vrel}{\kl[\Vrel]{\mathcal{V}\textrm{-}\mathsf{Rel}}}
\newrobustcmd{\Rel}{\mathsf{Rel}}

\newcommand{\Vcat}{\kl[\Vcat]{\mathcal{V}\textrm{-}\mathsf{Cat}}}
\newcommand{\Vcatsym}{\kl[\Vcatsym]{\mathcal{V}\textrm{-}\mathsf{Cat}_\mathsf{sym}}}

\newcommand{\B}{\mathcal B}

\newcommand{\PP}{\mathcal P}
\newcommand{\Ccal}{\mathcal C}

\renewcommand{\P}{\ensuremath{\mathcal E}}

\newcommand{\Set}{\kl[\Set]{\mathsf{Set}}}

\newrobustcmd\cref{\kl[up-to reflexivity]{\mathit{ref}}}
\newrobustcmd\csym{\kl[up-to transitivity]{\mathit{sym}}}
\newrobustcmd\ctrn{\kl[up-to symmetry]{\mathit{trn}}}
\newrobustcmd\cbhv{\kl[behavioural closure]{\mathit{bhv}}}
\newrobustcmd\cmtr{\kl[up-to metric closure]{\mathit{mtr}}}
\newrobustcmd\cctx{\kl[up-to contextual closure]{\mathit{ctx}}}
\newrobustcmd\ccvx{\kl[up-to convex closure]{\mathit{cvx}}}

\newcommand{\tensor}{\otimes}
\newcommand{\righthom}[2]{[{#1},{#2}]}
\newcommand{\lefthom}[2]{\llbracket{#1},{#2}\rrbracket}

\newcommand{\ev}{\mathit{ev}}

\newcommand{\To}{\Rightarrow}

\makeatletter
\newcommand{\dotminus}{\mathbin{\text{\@dotminus}}}
\newcommand{\op}{\mathit{op}}
\newcommand{\pullbackcorner}[1][dr]{\save*!/#1-1.5pc/#1:(-1.2,1.2)@^{|-}\restore}

\newcommand{\@dotminus}{%
  \ooalign{\hidewidth\raise1ex\hbox{.}\hidewidth\cr$\m@th-$\cr}%
}
\makeatother

\newcommand{\pair}[2]{\langle{#1},{#2}\rangle}

\newcommand{\mynewtheorem}[2]{
  \newaliascnt{#1}{dummy}
  \newtheorem{#1}[#1]{#2}
  \aliascntresetthe{#1}
  \expandafter\def\csname #1autorefname\endcsname{#2}
}

 \newcommand\restr[2]{{
  \left.\kern-\nulldelimiterspace 
  #1 
  \vphantom{\big|} 
  \right|_{#2} 
  }}

\newcommand{\id}{\mathit{id}}

\newrobustcmd\RelQ{\Rel_Q}
\newrobustcmd\bard{\kl[\bard]{\bar{d}}}
\newrobustcmd\revOrd{\mathrel{\kl[\revOrd]{\preceq}}}
\newrobustcmd\Pow{\ensuremath{\mathcal{P}}}
\newrobustcmd\Distr{\ensuremath{\mathcal{D}}}

\newrobustcmd\sdwDist{\kl[\sdwDist]{d_{\mathtt{sdw}}}}
\newrobustcmd\SigmaBB{\mathbb{\Sigma}}

\newrobustcmd\liftF{\kl[\liftF]{\widehat{F}}}
\newrobustcmd\liftT{\kl[\liftF]{\widehat{T}}}
\newrobustcmd\liftG{\kl[\liftF]{\widehat{G}}}
\newrobustcmd\liftFT{\kl[\liftF]{\widehat{FT}}}
\newrobustcmd\liftTF{\kl[\liftF]{\widehat{TF}}}
\newrobustcmd\liftFcT{\kl[\liftF]{\widehat{F\circ T}}}
\newrobustcmd\liftTcF{\kl[\liftF]{\widehat{T\circ F}}}

\newrobustcmd\liftFrest[1]{\kl[restriction of lifting]{\widehat{F}_{#1}}}
\newrobustcmd\liftGrest[1]{\kl[restriction of lifting]{\widehat{G}_{#1}}}

\newrobustcmd\cartLiftfR{\kl[\cartLiftfR]{\widetilde{f_R}}}
\newrobustcmd\cartLiftf{\kl[\cartLiftfR]{\widetilde{f}}}
\newrobustcmd\cartLift[2]{\kl[\cartLiftfR]{\widetilde{#1}_{#2}}}
\newrobustcmd\reindexf{\kl[\reindexf]{f^*}}
\newrobustcmd\reindexg{\kl[\reindexf]{g^*}}
\newrobustcmd\reindexfg{\kl[\reindexf]{(fg)^*}}
\newrobustcmd\reindex[1]{\kl[\reindexf]{({#1})^*}}
\newrobustcmd\reindexnobrk[1]{\kl[\reindexf]{{#1}^*}}
\newrobustcmd\fibre[1]{\kl[fibre]{\P_{#1}}}
\newrobustcmd\fibrePrim[1]{\kl[fibre]{\P'_{#1}}}
\newrobustcmd\fibrePred[1]{\kl[fibrePred]{\mathcal{V}\textrm{-}\mathsf{Pred}_{#1}}}
\newrobustcmd\fibreRel[1]{\kl[fibreRel]{\ensuremath{\mathcal{V}\textrm{-}\mathsf{Rel}_{#1}}}}

\newrobustcmd\directImage[1]{\kl[direct image]{\Sigma_{#1}}}

\newrobustcmd\chbaseFunct{\kl[\chbaseFunct]{\Delta}}
\newrobustcmd\isoPR{{\kl[\isoPR]{\iota}}}
\newrobustcmd\isoFibre[1]{{\kl[isoFibre]{\iota_{#1}}}}

\newrobustcmd\lambdaChBase{\kl[\lambdaChBase]{\lambda}}
\newrobustcmd\lambdaChBaseF[1]{\kl[lambdaChBaseF]{\lambda^{F}_{#1}}}
\newrobustcmd\lambdaChBaseG[1]{\kl[lambdaChBaseF]{\lambda^{G}_{#1}}}
\newrobustcmd\lambdaChBaseT[1]{\kl[lambdaChBaseF]{\lambda^{T}_{#1}}}
\newrobustcmd\lambdaChBaseFT[1]{\kl[lambdaChBaseF]{\lambda^{FT}_{#1}}}
\newrobustcmd\lambdaChBaseTF[1]{\kl[lambdaChBaseF]{\lambda^{TF}_{#1}}}

\newrobustcmd\barF{\kl[\barF]{\overline{F}}}
\newrobustcmd\barG{\kl[\barF]{\overline{G}}}
\newrobustcmd\barT{\kl[\barF]{\overline{T}}}
\newrobustcmd\barTF{\kl[\barTF]{\overline{T\circ F}}}
\newrobustcmd\barFT{\kl[\barFT]{\overline{F\circ T}}}
\newrobustcmd\barFrest[1]{\kl[restriction of lifting]{\overline{F}_{#1}}}
\newrobustcmd\barGrest[1]{\kl[restriction of lifting]{\overline{G}_{#1}}}
\newrobustcmd\barFcan{\kl[\barFcan]{\overline{F}_{\mathsf{can}}}}
\newrobustcmd\barTcan{\kl[\barFcan]{\overline{T}_{\mathsf{can}}}}

\newrobustcmd\sym[1]{\kl[sym]{\mathsf{sym}_{#1}}}
\newrobustcmd\constPred[1]{\kl[constant predicate]{\kappa_{#1}}}
\newrobustcmd\diago[1]{{\kl[diagonal function]{\delta_{#1}}}}

\newrobustcmd\true{\kl[\true]{\mathsf{true}}}
\newrobustcmd\truer{\kl[\truer]{\mathsf{true}_{r}}}
\newrobustcmd\trues{\kl[\truer]{\mathsf{true}_{s}}}
\newrobustcmd\truers{\kl[\truer]{\mathsf{true}_{r\otimes s}}}
\newrobustcmd\trueone{\kl[\truer]{\mathsf{true}_{1}}}
\newrobustcmd\liftFcan{\kl[\liftFcan]{\widehat{F}_{\mathsf{can}}}}
\newrobustcmd\liftTcan{\kl[\liftFcan]{\widehat{T}_{\mathsf{can}}}}
\newrobustcmd\liftFcanTwo{\kl[\liftFcanTwo]{\widehat{F}_{\mathsf{can}}}}
\newrobustcmd\uparr{\kl[\uparr]{\uparrow\!r}}
\newrobustcmd\upars{\kl[\uparr]{\uparrow\!s}}
\newrobustcmd\uparu{\kl[\uparr]{\uparrow\!u}}
\newrobustcmd\uparrs{\kl[\uparr]{\uparrow\!(r\otimes s)}}
\newrobustcmd\upar[1]{\kl[\uparr]{\uparrow\!{#1}}}
\newrobustcmd{\evcan}{\kl[\evcan]{\mathit{ev}_{\mathsf{can}}}}
\newrobustcmd{\evTcan}{\kl[\evcan]{\mathit{ev}^{T}_{\mathsf{can}}}}
\newrobustcmd{\evPcan}{\kl[\evcan]{\mathit{ev}^{\Pow}_{\mathsf{can}}}}
\newrobustcmd{\evFcan}{\kl[\evcan]{\mathit{ev}^{F}_{\mathsf{can}}}}
\newrobustcmd\bitensor[2]{{#1}\langle\!\otimes\!\rangle{#2}}
\newrobustcmd\orderFV{\kl[\orderFV]{\ll}}
\newrobustcmd\restrictionle{\le_{\uparr}}

\newrobustcmd\onebb{\mathbb{1}}
\newrobustcmd\liftDistr{\kl[\liftDistr]{\widehat{\zeta}}}
\newrobustcmd\barDistr{\kl[\barDistr]{\overline{\zeta}}}
\newrobustcmd\rotle{\rotatebox[origin=c]{45}{$\ge$}}
\newrobustcmd\rotlen{\rotatebox[origin=c]{90}{$\ge$}}
\newrobustcmd\objTruth{\kl[\objTruth]{\Omega}}
\newrobustcmd\distrLaw{\kl[\distrLaw]{\zeta}}
\newrobustcmd\distrLawNFA{\kl[\distrLawNFA]{\zeta}}
\newrobustcmd\evNFAF{\kl[\evNFAF]{\mathit{ev}_{F}}}
\newrobustcmd\evNFAT{\kl[\evNFAT]{\mathit{ev}_{T}}}

\newrobustcmd\leR{{\le}_{\mathbb{R}}}
\newrobustcmd\geR{\ge_{\mathbb{R}}}

\newrobustcmd\diagRel[1]{\kl[diagonal relation]{\mathit{diag}_{#1}}}
\newrobustcmd\comp{\kl[composition]{\boldsymbol{\ \cdot\ }}}

\newrobustcmd\VVlll{\mathrel{\kl[\VVlll]{\lll}}}

\begin{document}
\maketitle

\begin{abstract}
  Up-to techniques are a well-known method for enhancing coinductive
  proofs of behavioural equivalences. We introduce up-to techniques
  for behavioural metrics between systems modelled as coalgebras and
  we provide abstract results to prove their soundness in a
  compositional way.

  In order to obtain a general framework, we need a systematic way to
  lift functors: we show that the Wasserstein lifting of a functor,
  introduced in a previous work, corresponds to a change of base in a
  fibrational sense.  This observation enables us to reuse existing
  results about soundness of up-to techniques in a fibrational
  setting.
  We focus on the fibrations of predicates and relations valued in a
  quantale, for which pseudo-metric spaces are an example. To
  illustrate our approach we provide an example on distances between
  regular languages.
\end{abstract}

\section{Introduction}
\label{sec:introduction}

Checking whether two systems have an equivalent (or similar) behaviour
is a crucial problem in computer science. In concurrency theory, one
standard methodology for establishing behavioural equivalence of two
 systems is constructing a bisimulation relation between
them. When the systems display a quantitative behaviour, the notion of
behavioural equivalence is replaced with the more robust notion of
behavioural metric
\cite{bw:behavioural-pseudometric,afs:linear-branching-metrics-quant,DBLP:journals/tcs/DesharnaisGJP04}.

Due to the sheer complexity of state-based systems, computing their
behavioural equivalences and metrics can be very costly, therefore
optimization techniques---the so called \emph{up-to techniques}---have
been developed to render these computations more efficient. These
techniques found applications in various domains such as checking
algorithms~\cite{bp:checking-nfa-equiv,DBLP:conf/tacas/Bonchi0K17},
abstract interpretation~\cite{AbstractInt} and proof
assistants~\cite{danielsson2017up}.  In the qualitative setting and in
particular in concurrency, the theory of up-to techniques for
bisimulations and various other coinductive predicates has been
thoroughly
studied~\cite{ms:up-to,ps:enhancements-coind-proofs,DBLP:conf/popl/HurNDV13}. On
the other hand, in the quantitative setting, so far, only
\cite{cgv:up-to-bisim-metrics} has studied up-to techniques for
behavioural metrics. However, the notion of up-to techniques therein
and the accompanying theory of soundness are specific for
probabilistic automata and are not instances of the standard lattice
theoretic framework, which we will briefly recall next.

Suppose we want to verify whether two states in a system
behave in the same way, (e.g. whether two states of an NFA accept the
same language). The starting observation is that the relation of
interest (e.g. behavioural equivalence or language equivalence) can be
expressed as the greatest fixed point $\nu b$ of a monotone function
$b\colon\RelQ\to\RelQ$ on the complete lattice $\RelQ$ of relations on
the state space $Q$ of the system.
Hence, in order to prove that two
states $x$ and $y$ are behaviourally equivalent, i.e., 
$(x,y)\in\nu b$, it suffices to find a witness relation $r$ which on
one hand is a post-fixpoint of $b$, that is, $r\subseteq b(r)$ and on
the other hand contains the pair $(x,y)$. This is simply the
coinduction proof principle. 
However, exhibiting such a witness relation $r$ can be sometimes
computationally expensive. In many situations this computation can be
significantly optimized, if instead of computing a post-fixpoint of
$b$ one exhibits a relaxed invariant, that is a relation $r$ such that
$r\subseteq b(f(r))$ for a suitable function $f$. \AP The function $f$
is called a \intro{sound} up-to technique when the proof principle
\[
  \frac{(x,y)\in r \quad r\subseteq b(f(r))}{(x,y)\in \nu b}
\]
is valid.  Establishing the soundness of up-to techniques on a
case-by-case basis can be a tedious and sometimes delicate problem,
see e.g.~\cite{m:cac}. 
For this reason, several
works~\cite{San98MFCS,DBLP:conf/aplas/Pous07,ps:enhancements-coind-proofs,DBLP:conf/popl/HurNDV13,parrow2016largest,DBLP:conf/lics/Pous16}
have established a lattice-theoretic framework for proving soundness
results in a modular fashion. The key notion is compatibility: for
arbitrary monotone maps $b$ and $f$ on a complete lattice $(C,\leq)$,
the up-to technique $f$ is \intro{$b$-compatible} iff
$f\circ b\leq b\circ f$. Compatible techniques are sound and, most
importantly, can be combined in several useful ways.

In this paper we develop a generic theory of up-to techniques for
behavioural metrics applicable to different kinds of systems and
metrics, which reuses established methodology.
To achieve this we exploit the theory developed
in~\cite{bppr:up-to-fibration} by modelling systems as
\emph{coalgebras} \cite{rutten2000universal,jacobs2012introduction}
and behavioural metrics as coinductive predicates in a
\emph{fibration} \cite{hj:structural-induction-coinduction}. In order
to provide general soundness results, we need a principled way to lift
functors from $\intro\Set$ to metric spaces, a problem that has been studied in
\cite{h:closed-objects} and
\cite{bbkk:behavioral-metrics-functor-lifting}.
Our key observation is that these liftings arise from a change-of-base
situation between $\Vrel$ and $\Vpred$, namely the fibrations of
relations, respectively predicates, valued over a quantale $\VV$ (see
Section~\ref{sec:lifting-vrel-vcat} and~\ref{sec:wass-lifting}).

In Section~\ref{sec:up-to-techniques} we provide sufficient conditions
ensuring the compatibility of basic quantitative up-to techniques, as
well as proper ways to compose them. Interestingly enough, the
conditions ensuring compatibility of the quantitative analogue of
up-to reflexivity and up-to transitivity are subsumed by those used
in~\cite{h:closed-objects} to extend monads to a bicategory of
many-valued relations and generalize those
in~\cite{bbkk:behavioral-metrics-functor-lifting} (see the discussion
after Theorem~\ref{thm:restricting-Wasserstein}).

When the state space of a system is equipped with an algebraic
structure, e.g. in process algebras, one can usually exploit this
structure by reasoning up-to context. Assuming that the system forms a
\emph{bialgebra}
\cite{tp:math-operational-sem,DBLP:journals/tcs/Klin11}, i.e., that the
algebraic structure distributes over the coalgebraic behaviour as in
GSOS specifications, we give sufficient conditions ensuring the
compatibility of the quantitative version of contextual closure
(Theorem \ref{thm:Lifting-Distr-Law-Wass}).

In the qualitative setting, the sufficient conditions for
compatibility are automatically met when taking as lifting the
\emph{canonical} relational one (see~\cite{bppr:up-to-fibration}).  We
show that the situation is similar in the quantitative setting for a
certain notion of \emph{quantitative canonical} lifting. In
particular, up-to context is compatible for the canonical lifting
under very mild assumptions (Theorem \ref{thm:summary}). As an
immediate corollary we have that, in a bialgebra, syntactic contexts
are \emph{non-expansive} with respect to the behavioural metric induced by the
canonical lifting. This property and weaker variants of it (such as
non-extensiveness or uniform continuity), considered to be the
quantitative analogue of behavioural equivalence being a congruence,
have recently received considerable attention (see
e.g.~\cite{DBLP:journals/tcs/DesharnaisGJP04,bacci2013computing,tini2017compositional}).

To fix the intuition, Section~\ref{sec:motivating-example} provides a
motivating example, formally treated in Section~\ref{sec:examples}. We
conclude with a comparison to related work and a discussion of open
problems in Section~\ref{sec:conclusion}.

\full{All proofs and additional material are provided in the
  appendix.}\short{All proofs and additional material are provided in
  the full version of this paper
  \cite{bkp:up-to-behavioural-metrics-fibrations-arxiv}.}

\section{Motivating example: distances between regular languages}
\label{sec:motivating-example}

\AP Computing various distances (such as the edit-distance or Cantor
metric) between strings, and more generally between regular languages
or string distributions, has found various practical applications in
various areas such as speech and handwriting recognition or
computational biology. In this section we
focus on a simple distance between regular languages, which we
will call \intro{shortest-distinguishing-word-distance} and is defined
as $\intro\sdwDist(L,K)=c^{|w|}$ -- where $w$ is the shortest word
which belongs to exactly one of the languages $L,K$ and $c$ is a
constant such that $0<c<1$. 

As a running example, which will be formally explained in
Section~\ref{sec:examples}, we consider the non-deterministic finite
automaton in Figure~\ref{fig:ex-automaton} and the languages accepted
by the states $x_0$, respectively $y_0$. We can similarly define a
distance on the states of an automaton as the aforementioned distance
between the languages accepted by the two states. The inequality
\begin{equation}
  \label{eq:ineq}
  \sdwDist(x_0,y_0)\le c^n \qquad \mbox{(even $\sdwDist(x_0,y_0) =
    c^n$)} 
\end{equation}
holds in this example since no word of length smaller than
$n$ is accepted by either state. Note that computing this distance is
$\textsf{PSPACE}$-hard since the language equivalence problem for
non-deterministic automata can be reduced to it.

\begin{figure}
  \begin{center}
    \tikzset{every state/.style={minimum size=3em}}
    \tikzset{elliptic state/.style={draw,ellipse}}
    \scalebox{0.85}{\begin{tikzpicture}[x=2cm,y=-1cm,double distance=4pt]
      \node[elliptic state] (x0) at (0,1) {$x_0$} ;
      \node[elliptic state] (x1) at (1,1) {$x_1$} ;
      \node[elliptic state] (x2) at (2,1) {$x_2$} ;
      \node[elliptic state] (x3) at (3,1) {$x_{n-1}$} ;
      \node[elliptic state,accepting] (x4) at (4,1) {$x_n$} ;
      \node[elliptic state] (y0) at (0,2) {$y_0$} ;
      \node[elliptic state] (y1) at (1,2) {$y_1$} ;
      \node[elliptic state] (y2) at (2,2) {$y_2$} ;
      \node[elliptic state] (y3) at (3,2) {$y_{n-1}$} ;
      \node[elliptic state,accepting] (y4) at (4,2) {$y_n$} ;
      \begin{scope}[->]
        \path[shorten <=1pt] 
        (x0) edge[loop left] node[arlab] {$a,b$} (x0) 
        edge node[arlab,above] {$a$} (x1) ;
        \path[shorten <=1pt] 
        (x1) edge node[arlab,above] {$a,b$} (x2) ;
        \path[shorten <=1pt,dotted] 
        (x2) edge (x3) ;
        \path[shorten <=1pt] 
        (x3) edge node[arlab,above] {$a,b$} (x4) ;
        \path[shorten <=1pt] 
        (y0) edge[loop left] node[arlab] {$a,b$} (y0) 
        edge node[arlab,above] {$b$} (y1) ;
        \path[shorten <=1pt] 
        (y1) edge node[arlab,above] {$a,b$} (y2) ;
        \path[shorten <=1pt, dotted] 
        (y2) edge (y3) ;
        \path[shorten <=1pt] 
        (y3) edge node[arlab,above] {$a,b$} (y4) ;
      \end{scope}
    \end{tikzpicture}}
  \end{center}	

  \caption{Example automaton}
  \label{fig:ex-automaton}
\end{figure}

\AP One way to show this is to determinize the automaton in
Figure~\ref{fig:ex-automaton} and to use the fact that for
deterministic automata the \kl{shortest-distinguishing-word-distance}
can be expressed as the greatest fixpoint for a monotone
function. Indeed, for a finite deterministic automaton
$(Q,(\delta_a\colon Q\to Q)_{a\in A}, F\subseteq Q)$ over a finite
alphabet $A$, we have that $\sdwDist\colon Q\times Q\to[0,1]$ is the
greatest fixpoint of a function $b$ defined on the complete lattice
$[0,1]^{Q\times Q}$ of functions ordered with the \emph{reversed}
pointwise order $\intro*\revOrd$ and given by
\begin{eqnarray}
  \label{eq:def-of-b-NFA}
  b(d)(q_1,q_2)\!\!\!\!\! & = \!\!\!\!\!& \left\{
    \begin{array}{ll}
      1,\ \mbox{if only one of $q_1,q_2$ is in $F$}\\
      \max\limits_{a\in A} c\cdot \{ d(\delta_a(q_1),\delta_a(q_2)) \},\ 
      \mbox{otherwise}
    \end{array}
    \right.
\end{eqnarray}
Notice that we use the reversed order on $[0,1]$, for technical
reasons (see Section~\ref{sec:lifting-vrel-vcat}).

\AP In order to prove~\eqref{eq:ineq} we can define a
witness distance $\intro\bard$ on the states of the determinized
automaton such that $\bard(\{x_0\},\{y_0\})\le c^n$ and which is a
post-fixpoint for $b$, i.e., $\bard\revOrd b(\bard)$. Notice that this
would entail $\bard\revOrd \sdwDist$ and hence
$\sdwDist(\{x_0\},\{y_0\})\le \bard(\{x_0\},\{y_0\})\le c^n$.

This approach is problematic since the determinization of the
automaton is of exponential size, 
so we have to define $\bard$ for exponentially many pairs of sets of
states.  In order to mitigate the state space explosion we will use an
up-to technique, which, just as up-to congruence
in~\cite{bp:checking-nfa-equiv}, exploits the join-semilattice
structure of the state set $\Pow Q$ of the determinization of an NFA
with state set $Q$. The crucial observation is the fact that given the
states $Q_1,Q_2,Q'_1,Q'_2\in\Pow Q$ in the determinization of an NFA,
the following inference rule holds
\[
  \frac{\sdwDist(Q_1,Q_2)\le r\qquad \sdwDist(Q'_1,Q'_2)\le
    r}{\sdwDist(Q_1\cup Q'_1,Q_2\cup Q'_2)\le r}
\]
Based on this, we can define a monotone function $f$ on
$[0,1]^{\Pow Q\times\Pow Q}$ that closes a function $d$ according to
such proof rules, producing $f(d)$ such that $d\revOrd f(d)$ (the
formal definition of $f$ is given in Section~\ref{sec:examples}).
The general theory developed in this paper allows us to show in
Section~\ref{sec:examples} that $f$ is a sound up-to technique, i.e.,
it is sufficient to prove $\bard\revOrd b(f(\bard)) $ in order to
establish $\bard \revOrd \sdwDist$.

Using this technique it suffices to consider a quadratic number of
pairs of sets of states in the example. In particular we define a
function $\bard\colon \Pow Q \times \Pow Q \to [0,1]$ as follows:
\[\bard(\{x_i\},\{y_j\}) = c^{n-\max\{i,j\}} \] and
$\bard(X_1,X_2) = 1$ for all other values. Note that this function is not a
metric but rather, what we will call in Section \ref{sec:lifting-vrel-vcat}, a relation valued in $[0,1]$.

It holds that $\bard(\{x_0\},\{y_0\}) = c^n$. It remains to show that
$ \bard \revOrd b(f(\bard))$.
For this, it suffices to prove that
\[
  b(f(\bard))(\{x_i\},\{y_j\}) \le \bard(\{x_i\},\{y_j\})\,.
\]
For instance, when $i=j=0$ we compute the sets of $a$-successors,
which are $\{x_0,x_1\}$, $\{y_0\}$. We have that
$\bard(\{x_0\},\{y_0\}) = c^n \le c^{n-1}$,
$\bard(\{x_0\},\{y_1\}) = c^{n-1}$ and using the up-to proof rule
introduced above we obtain that
$f(\bard)(\{x_0,x_1\},\{y_0\}) \le c^{n-1}$. The same holds for the
sets of $b$-successors and since $x_0$ and $y_0$ are both non-final we
infer
$b(f(\bard))(\{x_0\},\{y_0\}) \le c\cdot c^{n-1} = c^n =
\bard(\{x_0\},\{y_0\})$. The remaining cases (when $i\neq 0\neq j$)
are analogous.

Our aim is to introduce such proof techniques for behavioural metrics,
to make this kind of reasoning precise, not only for this specific
example, but for coalgebras in general. Furthermore, we will not limit
ourselves to metrics and distances, but we will consider more general
relations valued in arbitrary quantales, of which the interval $[0,1]$
is an example.

\section{Preliminaries}
\label{sec:preliminaries}

We recall here formal definitions for notions
such as \kl{coalgebras}, \kl{bialgebras} or \kl{fibrations}.

\begin{defi}
  \AP A \intro{coalgebra for a functor} $F\colon\Ccal\to\Ccal$, or an
  \kl{$F$-coalgebra} is a morphism $\gamma\colon X\to FX$ for some
  object $X$ of $\Ccal$, referred to as the \intro{carrier} of the
  coalgebra $\gamma$. A morphism between two coalgebras
  $\gamma\colon X\to FX$ and $\xi\colon Y\to FY$ is a morphism
  $f\colon X\to Y$ such that $\xi\circ f=Ff\circ\gamma$.
  \intro{Algebras for the functor} $F$, or \kl{$F$-algebras}, are
  defined dually as morphisms of the form $\alpha\colon FX\to X$.
\end{defi}

\begin{defi}
  \label{def:bialgebra}
  \AP Consider two functors $F,T$ and a natural transformation
  $\intro\distrLaw\colon TF\Rightarrow FT$. A \intro{bialgebra} for
  $\distrLaw$ is a tuple $(X,\alpha,\gamma)$ such that
  $\alpha\colon TX\to X$ is a \kl{$T$-algebra}, $\gamma\colon X\to FX$
  is
\noindent  \begin{minipage}{0.35\linewidth}
   $
    \xymatrix@R=1.5em@C=1.5em{
      TX \ar[r]^\alpha \ar[d]^{T\gamma} & X \ar[r]^\gamma & FX \\
      TFX \ar[rr]^{\distrLaw_X} & & FTX \ar[u]^{F\alpha} }
    $
  \end{minipage}
  \begin{minipage}{0.65\linewidth}
    an \kl{$F$-coalgebra} so that the diagram on the left commutes.
    We call $\distrLaw$ the \intro{distributive law} of the bialgebra
    $(X,\alpha,\gamma)$, even when $T$ is not a monad.
  \end{minipage}
  \end{defi}

  \begin{ex}
  \label{ex:NFA-as-bialg}

  The determinization of an NFA can be seen as a \kl{bialgebra} with
  $X=\Pow Q$, the algebra $\mu_Q\colon \Pow\Pow Q\to \Pow Q$ given by
  the multiplication of the powerset monad, a coalgebra for the
  functor $F(X)=2\times X^A$, and a \kl{distributive law}
  $\intro\distrLawNFA\colon \Pow F\to F\Pow$ defined for
  $M \subseteq 2\times X^A$ by
  $\distrLawNFA_X(M) = (\bigvee_{(b,f)\in M} b, [a\mapsto \{f(a)\mid
  (b,f)\in M\}])$. See
  \cite{DBLP:journals/corr/abs-1302-1046,DBLP:journals/jcss/Jacobs0S15}
  for more details.
\end{ex}

We now introduce the notions of fibration and bifibration.

\begin{defi}
  \label{def:fibration}
  \AP A functor $p\colon\P\to\B$ is called a \intro{fibration} when
  for every morphism $f\colon X\to Y$ in $\B$ and every $R$ in $\P$
  with $p(R)=Y$ there exists a map
  $\intro\cartLiftfR\colon \intro\reindexf(R)\to R$ such that
  $p(\cartLiftfR)=f$,
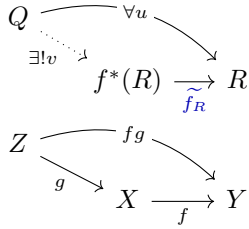
  \begin{minipage}{0.35\linewidth}
  $
    \begin{tikzcd}[column sep=1.5em, row sep=0.7em]
      Q\ar[dr,dotted,swap,"{\exists ! v}"]\ar[rrd,bend left,"{\forall u}" description] & &  \\
      & f^*(R)\ar[r,swap,"{{\cartLiftfR}}"] & R \\
      Z\ar[rd,swap,"{g}"]\ar[rrd,bend left,"{fg}" description] & &\\
      &X\ar[r,swap,"{f}"] & Y\\
    \end{tikzcd}
  $
\end{minipage}~~
\begin{minipage}[t]{0.65\linewidth}
  \vspace{-4em}
  satisfying the following universal property:\\
  For all maps
  $g:Z\to X$ in $\B$ and $u\colon Q\to R$ in $\P$ sitting above $fg$
  (i.e., $p(u)=fg$) there is a unique map $v\colon Q\to \reindexf(R)$
  such that $u=\cartLiftfR v$ and $p(v)=g$.
  
  \AP For $X$ in $\B$ we denote by $\fibre{X}$ the \intro{fibre above
$X$}, i.e., the subcategory of $\P$ with objects mapped by $p$ to $X$
and arrows sitting above the identity on $X$.
\end{minipage}
\end{defi}

\AP A map $\cartLiftf$ as above is called a \intro{Cartesian lifting}
of $f$ and is unique up to isomorphism.  If we make a choice of
Cartesian liftings, the association $R\mapsto \reindexf(R)$ gives rise
to the so-called \intro{reindexing functor}
$\reindexf\colon \fibre{Y}\to\fibre{X}$. In what follows we will only
consider \intro{split fibrations}, that is, the \kl{Cartesian
  liftings} are chosen such that we have $\reindexfg=\reindexg \reindexf$.

\AP A functor $p\colon \P\to\B$ is called a \intro{bifibration} if both
$p\colon \P\to\B$ and $p^\op\colon \P^\op\to\B^\op$ are
\kl{fibrations}.
Interestingly, a \kl{fibration} 
is a \kl{bifibration} if and
only if each \kl{reindexing functor}
$\reindexf\colon \fibre{Y}\to\fibre{X}$ has a left adjoint
$\directImage{f}\dashv \reindexf$,
see~\cite[Lemma~9.1.2]{Jacobs:fib}. We will call the functors
$\directImage{f}$ \intro{direct images along $f$}.

Two important examples of \kl{bifibrations} are those of relations
over sets, $p\colon \Rel\to\Set$, and of predicates over sets,
$p\colon \Pred\to\Set$, which played a crucial role
in~\cite{bppr:up-to-fibration}.  We do not recall their exact
definitions here, as they arise as instances of the more general
bifibrations of quantale-valued relations and predicates described in
detail in the next section. 
\noindent
\begin{minipage}{0.18\linewidth}
{  \begin{tikzcd}
    \P\ar[d,swap,"p"]\ar[r,"\liftF"] & \P'\ar[d,"p'"]\\
    \B\ar[r,"F"] & \B'
  \end{tikzcd}}
\end{minipage}~~
\begin{minipage}{0.80\linewidth}
  \AP Given \kl{fibrations} $p\colon \P\to\B$ and $p'\colon \P'\to\B'$
  and a functor on the base categories $F\colon \B\to\B'$, we call
  $\intro\liftF\colon \P\to\P'$ a \intro{lifting} of $F$ when
  $p'\liftF=Fp$.  \AP Notice that a lifting $\liftF$ restricts to a
  functor between the \kl{fibres} \phantomintro{restriction of
    lifting} $\liftFrest{X}\colon \fibre{X}\to\fibrePrim{FX}$.  We omit the
subscript $X$ when it is clear from the context.
\end{minipage}

  Consider an arbitrary lifting $\liftF$ of $F$ and a morphism
  $f\colon X\to Y$ in $\B$. For any $R\in\fibre{Y}$  the
  maps
  $\cartLift{Ff}{\liftF{R}}\colon\reindex{Ff}(\liftF R)\to\liftF R$
  and $\liftF(\cartLiftfR)\colon\liftF(\reindexf R)\to\liftF R$ sit
  above $Ff$. Using the universal property in
  Definition~\ref{def:fibration}, we obtain a canonical morphism
\begin{equation}
  \label{eq:canonical-ineq}
\liftF\circ \reindexf(R)\to\reindex{Ff}\circ\liftF (R)\,.  
\end{equation}
A lifting $\liftF$ is called a \intro{fibred lifting} when the natural
transformation in~\eqref{eq:canonical-ineq} is an isomorphism.

\section{Moving towards a quantitative setting}
\label{sec:lifting-vrel-vcat}

We start by introducing two \kl{fibrations} which are the foundations
for our quantitative reasoning: predicates and relations valued in a
\kl{quantale}.

\begin{defi}
  \label{def:quantale}
  \AP A \intro{quantale} $\intro\VV$ is a complete lattice equipped with an
  associative operation $\tensor:\VV\times\VV\to\VV$ which is
  distributive on both sides over arbitrary joins $\bigvee$.
\end{defi}

This implies that for every $y\in\VV$ the functor $-\tensor y$ has a
right adjoint $[y,-]$. Similarly, for every $x\in \VV$, the functor
$x\tensor -$ has a right adjoint, denoted by $\lefthom{x}{-}$.  Thus,
for every $x,y,z\in\VV$, we have:
$x\tensor y\le z \iff x\le\righthom{y}{z} \iff y\le\lefthom{x}{z}$.

If $\tensor$ has an \kl{identity element} or \intro{unit} $1$ for
$\tensor$ the quantale is called \intro{unital}.  If
$x\tensor y=y\tensor x$ for every $x,y\in\VV$ the quantale is called
\intro{commutative} and we have $\righthom{x}{-}=\lefthom{x}{-}$.
Hereafter, we only work with \kl{unital}, \kl{commutative}
\kl{quantales}.
\begin{ex}
  \label{ex:quantale}
  The Boolean algebra $2$ with $\tensor=\wedge$ is a \kl{unital} and
  \kl{commutative} \kl{quantale}: the unit is $1$ and
  $\righthom{y}{z}=y\to z$.
  The complete lattice $[0,\infty]$ ordered by the reversed
  order\footnote{To avoid confusion we  use $\lor,\land$
    in the quantale and $\inf,\sup$ in the reals.} of the reals,
  i.e., $\le=\geR$ and with $\tensor = +$ is a unital
  commutative \kl{quantale}: the unit is $0$ and for every
  $y,z\in[0,\infty]$ we have $\righthom{y}{z}=z\dotminus y$ (truncated
  subtraction). Also $[0,1]$ is a unital quantale where
  $r\tensor s=\min(r+s,1)$ (truncated addition).
\end{ex}

\begin{defi}
  \label{def:vrel-vpred}
  \AP Given a set $X$ and a \kl{quantale} $\VV$, a \intro{$\VV$-valued
    predicate} on $X$ is a map $p:X\to\VV$.  A \intro{$\VV$-valued
    relation} on $X$ is a map $r:X\times X\to\VV$.
\end{defi}

Given two \kl{$\VV$-valued predicates} $p,q:X\to\VV$, we say that
$p\le q \iff \forall x\in X.\ p(x)\le q(x)$.

\begin{defi}
  \label{def:vpred-cat}
  A \intro{morphism between $\VV$-valued predicates} $p:X\to \VV$ and
  $q:Y\to\VV$ is a map $f:X\to Y$ such that $p\le q\circ f$. We
  consider the category $\intro\Vpred$ whose objects are
  \kl{$\VV$-valued predicates} and arrows are as above.
\end{defi}

\begin{defi}
  \label{def:vrel-cat}
  \AP A \intro{morphism between $\VV$-valued relations}
  $r:X\times X\to \VV$ and $q:Y\times Y\to\VV$ is a map $f:X\to Y$
  such that $p\le q\circ (f\times f)$. We consider the category
  $\intro\Vrel$ whose objects are $\VV$-valued relations and arrows
  are as above.
\end{defi}

\pseudopar{The bifibration of $\VV$-valued predicates.}  \AP The
forgetful functor $\Vpred\to\Set$ mapping a predicate $p:X\to\VV$ to
$X$ is a bifibration. \AP The \kl{fibre} \phantomintro{fibrePred}
$\fibrePred{X}$ is the lattice of \kl{$\VV$-valued predicates} on $X$.
For $f:X\to Y$
in $\Set$ the \kl{reindexing} and \kl{direct image} functors on a
predicate $p\in\fibrePred{Y}$ are  given by
\begin{equation*}\reindexf(p)=p\circ f \qquad \text{ and } 
\qquad \directImage{f}(p)(y)=\bigvee\{p(x)\mid x\in f^{-1}(y)\}\,.
  \end{equation*}

  \pseudopar{The bifibration of $\VV$-valued relations.} \AP Notice
  that we have the following pullback in $\mathsf{Cat}$, where
  $\intro\chbaseFunct X=X\times X$. This is a change-of-base situation
  and thus the functor $\Vrel\to\Set$ mapping each $\VV$-valued
  relation to its underlying set is also a bifibration.  
  \noindent\begin{minipage}{0.70\linewidth}
    We denote by \phantomintro{fibreRel} $\fibreRel{X}$ the fibre above a set $X$.
    For each set $X$ the functor $\intro\isoPR$ restricts to an
    isomorphism \phantomintro{isoFibre}
  $ \isoFibre{X}:\fibreRel{X}\to\fibrePred{X\times X}\,.$
\end{minipage}
\begin{minipage}{0.20\linewidth}
\begin{center}  $
  \xymatrix@R=1em{
  \Vrel \ar[d]\pullbackcorner\ar[r]^-{\isoPR} & \Vpred\ar[d] \\
  \Set\ar[r]_-{\chbaseFunct} & \Set \\
}
$
\end{center}
\end{minipage}

For $f:X\to Y$ in $\Set$ the \kl{reindexing} and \kl{direct image}
on $p\in\fibreRel{Y}$ are given by
\begin{equation*}
\reindexf(p)=p\circ (f\times f) \qquad \text{ and } \qquad
\directImage{f}(p)(y)=\bigvee\{p(x,x')\mid (x,x')\in (f\times f)^{-1}(y,y')\}\,.
\end{equation*}

\AP For two relations $p,q\in \Vrel_X$, we define their
\intro{composition} $p\comp q\colon X\times X \to \VV$ by
$p\comp q (x,y) =\bigvee \{p(x,z) \otimes q(z,y)\mid z\in X\}$. We
define the \intro{diagonal relation}
$\diagRel{X}\in \fibreRel{X} $ by $\diagRel{X}(x,y)=1$ if $x=y$ and
$\bot$ otherwise.

\begin{defi}
  \label{def:vcat}
  \AP We say that a \kl{$\VV$-valued relation} $r:X\times X\to\VV$ is
  \begin{itemize}
  \item \intro{reflexive} if for all $x\in X$ we have $r(x,x)\ge 1$,
    (i.e., $r\ge \diagRel{X}$);
  \item \intro{transitive} if $r\comp r\le r$;
  \item \intro{symmetric} if $r=r\circ \sym{X}$, where
    \phantomintro{sym} $\sym{X}\colon X\times X\to X\times X$ is the
    symmetry isomorphism.
  \end{itemize}
  \AP We denote by $\intro\Vcat$ the full subcategory of $\Vrel$
  consisting of \kl{reflexive}, \kl{transitive} relations and by
  $\intro\Vcatsym$ the full subcategory of $\Vrel$ that are
  additionally \kl{symmetric}.
\end{defi}

Note that $\Vcat$ is the category of small categories enriched over
the $\VV$ in the sense of~\cite{k:enriched-category-theory}.
\begin{ex} \label{ex:vcats} For $\VV=2$, \kl{$\VV$-valued
    relations} are just relations. \kl{Reflexivity},
  \kl{transitivity} and \kl{symmetry} coincide with the standard
  notions, so $\Vcat$ is the category of preorders, while $\Vcatsym$
  is the category of equivalence relations.

  For $\VV=[0,\infty]$, $\Vcat$ is the category of generalized metric
  spaces \`a la Lawvere \cite{l:metric-spaces-gen-logic}
  (i.e., directed pseudo-metrics and non-expansive maps),
  while $\Vcatsym$ is the one of pseudo-metrics.
\end{ex}

\section{Lifting functors to $\VV$-$\mathsf{Pred}$ and $\VV$-$\mathsf{Rel}$}
\label{sec:wass-lifting}

In the previous section, we have introduced the \kl{fibrations} of
interest for quantitative reasoning.
In order to deal with coinductive predicates in this
setting, it is convenient to have a structured way to lift
$\Set$-functors to $\VV$-valued predicates and relations, and
eventually to $\VV$-enriched categories. Our strategy is to first lift
functors to $\Vpred$ and then, by exploiting the change of base, move
these liftings to $\Vrel$. A comparison with the extensions of
$\Set$-monads to the bicategory of $\VV$-matrices~\cite{h:closed-objects} is provided in
Section~\ref{sec:conclusion}.

\subsection{$\VV$-predicate liftings}
\label{sec:Vpred-liftings}

Liftings of $\Set$-functors to the category $\Pred$ (for $\VV = 2$) of
predicates have been widely studied in the context of coalgebraic
modal logic, as they correspond to modal operators (see
e.g. \cite{s:coalg-logics-limits-beyond-journal}). For $\Vpred$, we
proceed in a similar way. Let us analyse what it means to have a
\kl{fibred lifting} $\liftF$ to $\Vpred$ of an endofunctor $F$ on
$\Set$. First, recall that the fibre $\fibrePred{X}$ is just the
preorder $\VV^X$. So the restriction $\liftFrest{X}$ to such a fibre
corresponds to a \emph{monotone} map $\VV^X\to\VV^{FX}$. The fact that
$\liftF$ is a \kl{fibred lifting} 
essentially means that the maps $(\VV^X\to\VV^{FX})_{X}$ form a
natural transformation between the contravariant functors $\VV^{-}$
and $\VV^{F-}$. Furthermore, by Yoneda lemma we know that natural
transformations $\VV^-\To\VV^{F-}$ are in one-to-one correspondence
with maps $\ev:F\VV\to\VV$, which we will call hereafter
\intro{evaluation maps}. One can characterise the \kl{evaluation maps}
which correspond to the \emph{monotone} natural transformations. \AP
These are the \intro{monotone evaluation maps}
$\ev:(F\VV,\orderFV)\to(\VV,\le)$ with respect to the usual order $\le$ on
$\VV$ and an order $\orderFV$ on $F\VV$ defined by applying the
standard canonical relation lifting of $F$ to~$\le$.

\begin{restatable}{prop}{propLiftCorrespondences}
  \label{prop:lift-correspondences}
  There is a one-to-one correspondence between
  \begin{itemize}
  \item \kl{fibred liftings} $\liftF$ of $F$ to $\Vpred$,
  \item monotone natural transformations $\VV^-\To\VV^{F-}$,
  \item \kl{monotone evaluation maps} $\ev:F\VV\to\VV$.  
  \end{itemize}
\end{restatable}

Notice that the correspondence between \kl{fibred liftings} and
monotone \kl{evaluation maps} is given in one direction by
$\ev=\liftF(\id_\VV)$, and conversely, by
$\liftF(p\colon X\to\VV)=\ev\circ F(p)$.

\smallskip

\pseudopar{Evaluation maps as Eilenberg-Moore algebras.} Evaluation
maps have also been extensively considered in the coalgebraic approach
to modal logics \cite{s:coalg-logics-limits-beyond-journal}.
\AP A special kind of evaluation map arises when the truth values
$\VV$ have an algebraic structure for a given monad $(T,\mu,\eta)$,
that is, we have $\VV=T\objTruth$ for some object $\intro\objTruth$
and the evaluation map $T\VV\to\VV$ is an Eilenberg-Moore algebra for
$T$.  This notion of monadic modality has been studied
in~\cite{DBLP:journals/tcs/Hasuo15} where the category of free
algebras for $T$ was assumed to be order enriched. \full{In
  Lemma~\ref{lem:evfct-monad} in Appendix~\ref{app:predlift} we show
  that under reasonable assumptions, the evaluation map obtained as
  the free Eilenberg-Moore algebra on $\objTruth$ (i.e.,
  $\ev\colon T\VV\to\VV$ is just
  $\mu_{\objTruth}\colon T^2\objTruth\to T\objTruth$) is a monotone
  evaluation map, and hence gives rise to a fibred lifting of
  $T$.}\short{Under reasonable assumptions the evaluation map obtained
  as the free Eilenberg-Moore algebra on $\objTruth$ (i.e.,
  $\ev\colon T\VV\to\VV$ is just
  $\mu_{\objTruth}\colon T^2\objTruth\to T\objTruth$) is a monotone
  evaluation map, and hence gives rise to a fibred lifting of $T$ (see
  \cite{bkp:up-to-behavioural-metrics-fibrations-arxiv} for more
  details.)}

We provide next several examples of monotone evaluation maps which
arise in this fashion.

\begin{ex}
  \label{ex:evaluation-as-mult-pow-1}
  When $T$ is the powerset monad $\Pow$ and $\objTruth=1$ we obtain
  $\VV=2$ and $\mu_1\colon\Pow 2\to 2$ corresponds to the $\Diamond$
  modality, i.e., to an existential predicate transformer,
  see~\cite{DBLP:journals/tcs/Hasuo15}.
\end{ex}

\begin{ex}
  \label{ex:evaluation-as-mult-distr-2}
  When $T$ is the probability distribution functor $\Distr$ on $\Set$
  and $\objTruth=2=\{0,1\}$ equipped with the order $1\sqsubseteq 0$
  we obtain $\VV=\mathcal{D}\{0,1\}\cong[0,1]$ with the reversed order
  of the reals, i.e., $\le\ = \ \geR$. In this case
  $\ev_\mathcal{D}(f) = \sum_{r\in [0,1]} r\cdot f(r)$ for
  $f\colon [0,1]\to [0,1]$ a probability distribution (expectation of
  the identity random variable).
\end{ex}
\pseudopar{The canonical evaluation map.} \AP In the case $\VV=2$,
there exists a simple way of lifting a functor $F \colon \Set\to\Set$:
given a predicate $p\colon U \rightarrowtail X$, one defines the
canonical predicate lifting $\intro\liftFcanTwo(U)$ of $F$ as the
epi-mono factorization of $Fp\colon FU \to FX$. This lifting
corresponds to a canonical \kl{evaluation map}
$\intro\true \colon 1 \to 2$ which maps the unique element of $1$ into
the element $1$ of the quantale $2$. For $\VV$-relations, a
generalized notion of canonical evaluation map was introduced
in~\cite{h:closed-objects}.  \AP For $r\in\VV$ consider the subset
$\intro\uparr=\{v\in\VV\mid v\ge r\}$ and write
$\intro\truer\colon\uparr\to\VV$ for the inclusion.  Given $u\in F\VV$
we write $u\in F(\uparr)$ when $u$ is in the image of the injective
function $F(\truer)$.  \AP Following~\cite{h:closed-objects}, we
define $\intro\evcan:F\VV\to\VV$ \phantomintro{canonical evaluation
  map} as follows:
\[
  \evcan(u)=\bigvee\{ r\mid u\in F(\uparr)\}.
\]
\begin{ex}
  \label{ex:pow-funct-can-ev}
  Assume $F$ is the powerset functor $\PP$ and let $u\in\PP(\VV)$.  We
  obtain that
  \[
    \evcan(u)=\bigvee\{ r\mid u\subseteq\,\uparr\} \text{, or
      equivalently, } \evcan(u)=\bigwedge u\,.
  \]
  When $\VV=2$ we obtain $\evcan\colon\PP2\to 2$ given by
  $\evcan(u)= 1$ iff $u=\emptyset$ or $u=\{1\}$. This corresponds to
  the $\Box$ operator from modal logic.
  If $\VV = [0,\infty]$ we have $\evcan(u)=\sup u$.
\end{ex}
\begin{ex}
  \label{ex:distr-funct-can-ev}
The canonical evaluation map for the distribution monad $\Distr$ and
$\VV=[0,1]$ is $\evcan(f) = \sup_{r\in [0,1]} f(r)$, which is not the monad
multiplication.  
\end{ex}

\AP The canonical evaluation map $\evcan$ is monotone whenever the
functor $F$ preserves weak pullbacks
\full{(see~Lemma~\ref{lem:can-ev-map-is-monotone} in
  Appendix~\ref{app:predrel})}\short{(see
  \cite{bkp:up-to-behavioural-metrics-fibrations-arxiv})}.  For such
functors, by Proposition~\ref{prop:lift-correspondences}, the map
$\evcan$ induces a \kl{fibred lifting} $\intro\liftFcan$ of $F$,
called the canonical $\Vpred$-lifting of $F$ and defined by
\[
  \liftFcan(p)(u) = \bigvee\{r\mid F(p)(u)\in F(\uparr)\}\ \text{ for
  } p\in\Vpred_X\text{ and } u\in FX\,.
\]

\subsection{From predicates to relations via Wasserstein}
\label{sec:wasserstein-lifting}

We describe next how functor liftings to $\Vrel$ can be systematically
obtained using the change-of-base situation described above. In
particular, we see how the Wasserstein metric between probability
distributions (defined in terms of couplings of distributions) can be
naturally modelled in the fibrational setting.

Consider a $\VV$-predicate lifting $\liftF$ of a $\Set$-functor $F$. A
natural way to lift $F$ to $\VV$-relations using $\liftF$ is to regard
a $\VV$-relation $r\colon X\times X\to\VV$ as a $\VV$-predicate on the
product $X\times X$. Formally, we will use the isomorphism
$\isoFibre{X}$ described in Section~\ref{sec:lifting-vrel-vcat}. We
can apply the functor $\liftF$ to the predicate $\isoFibre{X}(r)$ in
order to obtain the predicate $\liftF\circ \isoFibre{X}(r)$ on the set
$F(X\times X)$. Ideally, we would want to transform this predicate
into a relation on $FX$. So first, we have to transform it into a
predicate on $FX\times FX$. \AP To this end, we use the natural
transformation \phantomintro{lambdaChBaseF}
\begin{equation}
  \label{eq:nat-trans-chBase}
  \lambdaChBase^{F}\colon F\circ \chbaseFunct \To \chbaseFunct \circ F \text{ defined by } \lambdaChBaseF{X}=\langle F\pi_1, F\pi_2\rangle\colon F(X\times X)\to
FX\times FX\,.
\end{equation}
 We drop the superscript and simply write
$\intro\lambdaChBase$ when the functor $F$ is clear from the
context.
Additionally, the bifibrational structure of
$\Vrel$ plays a crucial role, as we can use the \kl{direct image}
functor $\directImage{\lambdaChBase_{X}}$ to transform
$\liftF\circ \isoFibre{X}(r)$ into a predicate on $FX\times FX$.
Putting all the pieces together, we define a lifting of $F$ on the
fibre $\fibreRel{X}$ as the composite $\barFrest{X}$ given by:
\begin{equation}
  \label{eq:def-V-rel-lift}
  \begin{tikzcd}
    \barFrest{X}\colon \fibreRel{X}\ar[r, "\isoFibre{X}"] &
    \fibrePred{\chbaseFunct X}\ar[r,"\liftFrest{\chbaseFunct X}"]
    & \fibrePred{F\chbaseFunct
      X}\ar[r,"\directImage{\lambdaChBase_{X}}"] &
    \fibrePred{\chbaseFunct FX}\ar[r, "\isoFibre{FX}^{-1}"] &
    \fibreRel{FX}
  \end{tikzcd}
\end{equation}

\AP The aim is to define a lifting $\barF$ of $F$ to $\Vrel$. The
above construction provides the definition of $\intro\barF$ on the
fibres and, in particular, on the objects of $\Vrel$.  For a
\kl{morphism between $\VV$-relations} $p\in\fibreRel{X}$ and
$q\in\fibreRel{Y}$, i.e., a map $f\colon X\to Y$ such that
$p\le \reindexf(q)$, we define $\barF(f)$ as the map
$Ff\colon FX\to FY$. To see that this is well defined it remains to
show that $\barF p\le \reindex{Ff}(\barF q)$. This is the first part
of the next proposition.
\begin{restatable}{prop}{propWassersteinLifting}
  \label{prop:Wasserstein-lifting}
  The functor $\barF$ defined above is a well defined \kl{lifting} of
  $F$ to $\Vrel$. Furthermore, when $F$ preserves weak pullbacks and
  $\liftF$ is a \kl{fibred lifting} of $F$ to $\Vpred$, then $\barF$
  is a \kl{fibred lifting} of $F$ to $\Vrel$.
\end{restatable}

Spelling out the concrete description of the \kl{direct image} functor
and of $\lambdaChBase_X$, we obtain for a relation $p\in\Vrel_{X}$ and
$t_1,t_2\in FX$, that
\begin{equation}
  \label{eq:Wasserstein-lift-formula}
  \barF(p)(t_1,t_2)=\bigvee\{\liftF(p)(t)\mid t\in F(X\times X), F\pi_i(t)=t_i\}
\end{equation}
Unraveling the definition of $\liftF(p)(t)=\ev\circ F(p)$, we obtain
for $\barF(p)$ the same formula as for the extension of $F$
on $\VV$-matrices, as given
in~\cite[Definition~3.4]{h:closed-objects}. This definition
in~\cite{h:closed-objects} is obtained by a direct generalisation of
the Barr extensions of $\Set$-functors to the bicategory of relations.
In contrast, we obtained~\eqref{eq:Wasserstein-lift-formula} by
exploiting the fibrational change-of-base situation and by first
considering a $\Vpred$-lifting.

\AP We call a lifting of the form $\barF$ the \intro{Wasserstein
  lifting} of $F$ corresponding to $\liftF$. This terminology is
motivated by the next example.

\begin{ex}
  \label{ex:wasserstein}
  When $F = \Distr$ (the distribution functor), $\VV = [0,1]$ and
  $ev_F$ is as in Example~\ref{ex:evaluation-as-mult-distr-2} then
  $\barF$ is the original Wasserstein metric from transportation
  theory \cite{v:optimal-transport}, which -- by the
  Kantorovich-Rubinstein duality -- is the same as the Kantorovich
  metric. Here we compare two probability distributions
  $t_1,t_2\in \Distr X$ and obtain as a result the coupling
  $t\in \Distr(X\times X)$ with marginal distributions $t_1,t_2$,
  giving us the optimal plan to transport the ``supply'' $t_1$ to the
  ``demand'' $t_2$. More concretely, given a metric $d\colon X\times
  X\to \VV$, the (discrete) Wasserstein metric is defined as
  \[ d^W(t_1,t_2) = \inf \{ \sum_{x,y\in X} d(x,y)\cdot t(x,y) \mid
    \sum_y t(x,y) = t_1(x), \sum_x t(x,y) = t_2(y) \}. \]
  On the other hand, when $\ev_F$ is the \kl{canonical evaluation map}
  of Example~\ref{ex:distr-funct-can-ev} the corresponding
  $\Vrel$-lifting $\barF$ minimizes the longest distance (and hence
  the required time) rather than the total cost of transport.
\end{ex}

\begin{ex}
  \label{ex:hausdorff}
  Let us spell out the definition when $F = \Pow$ (powerset functor),
  $\VV = [0,1]$ and $\ev_F\colon \Pow[0,1]\to [0,1]$ corresponds to
  $\sup$, which is clearly monotone and is the canonical evaluation
  map as in Example~\ref{ex:pow-funct-can-ev}.

  Then, given a metric $d\colon X\times X\to [0,1]$ and
  $X_1,X_2\subseteq X$, the lifted metric is defined as follows
  (remember that the order is reversed on $[0,1]$):
  \[ \barF(d)(X_1,X_2) = \inf \{ \sup d[Y] \mid Y \subseteq X\times X,
    \pi_i[Y] = X_i \} \] \AP As explained
  in~\cite{bbkk:coalgebraic-behavioral-metrics}, this is the same as
  the \intro{Hausdorff metric} $d^H$ defined by:
  \[ d^H(X_1,X_2) = \sup \{\sup_{x_1\in X_1} \inf_{x_2\in X_2}
    d(x_1,x_2), \sup_{x_2\in X_2} \inf_{x_1\in X_1} d(x_1,x_2)\} \]
\end{ex}

The next lemma establishes that this construction is functorial:
liftings of natural transformations to $\Vpred$ can be converted into
liftings of natural transformations between the corresponding
\kl{Wasserstein liftings} on $\Vrel$.

\AP \phantomintro{\liftDistr}
\phantomintro{\barDistr}

\begin{restatable}{lem}{lemWassIsFunctorial}
  \label{lem:Wass-is-Functorial}
  If there exists a lifting $\liftDistr\colon\liftF\To\liftG$ of
  a natural transformation $\distrLaw\colon F\To G$, then there exists
  a lifting $\barDistr\colon\barF\To\barG$ between the
  corresponding \kl{Wasserstein liftings}. Furthermore, when $\liftF$
  and $\liftG$ correspond to monotone evaluation maps $\ev_F$ and
  $\ev_G$, then the lifting $\liftDistr$ exists and is unique if and only if
  $\ev_F\le\ev_G\circ\distrLaw_\VV$.
\end{restatable}

For $\VV=[0,\infty]$, one is also interested in lifting functors
to the category of (generalized) pseudo-metric spaces, not just of
$[0,\infty]$-valued relations. This motivates the next question: when
does the lifting $\barF$ restrict to a functor on $\Vcat$ and
$\Vcatsym$? \AP We have the following characterization theorem, where
\phantomintro{constant predicate} $\constPred{X}\colon X\to \VV$ is
the constant function $x\mapsto 1$ and $u\tensor v\colon X\to \VV$
denotes the pointwise tensor of two predicates $u,v\colon X\to \VV$,
i.e., $(u\tensor v)(x)=u(x)\tensor v(x)$.

\begin{restatable}{theo}{thmRestrictionOfWasserstein}
  \label{thm:restricting-Wasserstein}
  Assume $\liftF$ is a lifting of $F$ to $\Vpred$ and  $\barF$ is the
  corresponding $\Vrel$ \kl{Wasserstein lifting}. Then
  \begin{itemize}
  \item If $\liftF(\constPred{X}) \ge \constPred{FX}$ then
    $\barF(\diagRel{X}) \ge \diagRel{FX}$, hence $\barF$ 
    preserves \kl{reflexive} relations;
  \item If $\liftF$ is a \kl{fibred lifting}, $F$ preserves weak
    pullbacks and $\liftF(p\otimes q) \ge \liftF(p)\otimes \liftF(q)$
    then $\barF(p\comp q) \ge \barF(p)\comp \barF(q)$, hence $\barF$
    preserves \kl{transitive} relations; 
  \item $\barF$ preserves \kl{symmetric} relations.
  \end{itemize}
  Consequently, when all the above hypotheses are satisfied, then the
  corresponding $\Vrel$ \kl{Wasserstein lifting} $\barF$ restricts to
  a lifting of $F$ to both $\Vcat$ and $\Vcatsym$.
\end{restatable}

For $\VV=[0,\infty]$, the first condition of
Theorem~\ref{thm:restricting-Wasserstein} is a relaxed version of a
condition in \cite[Definition
5.14]{bbkk:coalgebraic-behavioral-metrics} used to guarantee
reflexivity.  The second condition (for transitivity) is equivalent to
a non-symmetric variant of a condition in
\cite{bbkk:coalgebraic-behavioral-metrics} \full{(see
  Lemma~\ref{lem:wbhvd-additivity-compare} in
  Appendix~\ref{app:predrel})}\short{(see
  \cite{bkp:up-to-behavioural-metrics-fibrations-arxiv})}.

We can establish generic sufficient conditions on a monotone
evaluation map $\ev$ so that the corresponding $\Vpred$-lifting
$\liftF$ satisfies the conditions of
Theorem~\ref{thm:restricting-Wasserstein}. In
\full{Proposition~\ref{propCharEvMapWellBeh} in
  Appendix~\ref{app:predrel}}\short{\cite{bkp:up-to-behavioural-metrics-fibrations-arxiv}}
we show that $\liftF(p\tensor q)\ge\liftF(p)\tensor\liftF(q)$ holds
whenever the map $\tensor\colon\VV\times\VV\to\VV$ is the carrier of a
lax morphism in the category of $F$-algebras between
$(\VV,\ev)^2\to(\VV,\ev)$, i.e.,
$\tensor\circ(\ev\times\ev)\circ\lambdaChBase_\VV\le\ev\circ
F(\tensor)$. Furthermore, $\liftF(\constPred{X})\ge\constPred{X}$
holds whenever the map $\constPred{\onebb}\colon\onebb\to\VV$ is the
carrier of a lax morphism from the one-element $F$-algebra
$!\colon F\onebb\to\onebb$ to $(\VV,\ev)$, i.e.,
$\constPred{\onebb}\circ !\le\ev\circ F\constPred{\onebb}$. These two
requirements correspond to the conditions $(Q_\tensor)$, respectively
$(Q_{k})$ satisfied by a topological theory in the sense
of~\cite[Definition~3.1]{h:closed-objects}.
Since these two are satisfied by the canonical evaluation map $\evcan$,\footnote{The same observation is present in~\cite[Theorem~3.3(b)]{h:closed-objects} but in a slightly different setting.} we immediately obtain
\begin{restatable}{prop}{propCanLiftingIsWellBehaved}
  \label{prop:can-lifting-is-well-behaved}
  Whenever $F$ preserves weak pullbacks the canonical lifting
  $\liftFcan$ satisfies the conditions in
  Theorem~\ref{thm:restricting-Wasserstein}:
  \begin{enumerate}
  \item
    $\liftFcan(p\otimes q) \ge \liftFcan(p)\otimes \liftFcan(q)$,
    for all $p,q\in\fibrePred{X}$,
  \item $\liftFcan(\constPred{X})\ge \constPred{X}$.
  \end{enumerate}
\end{restatable}
\AP An immediate consequence of
Proposition~\ref{prop:can-lifting-is-well-behaved} and of
Theorem~\ref{thm:restricting-Wasserstein} is that the Wasserstein
lifting \phantomintro{canonical Wasserstein lifting} $\intro\barFcan$
that corresponds to $\liftFcan$ restricts to a lifting of $F$ to both
$\Vcat$ and $\Vcatsym$.

\section{Quantitative up-to techniques}
\label{sec:up-to-techniques}

The fibrational constructions of the previous section provides a
convenient setting to develop an abstract theory of quantitative
up-to techniques. The coinductive object of interest is the greatest
fixpoint of a monotone map $b$ on $\Vrel$, hereafter denoted by
$\nu b$.  Recall that an up-to technique, namely a monotone map $f$ on
$\Vrel$, is \kl{sound} whenever $d \le b(f(d))$ implies $d\le \nu b$,
for all $d\in\Vrel_X$; it is \kl{compatible} if
$f\circ b \le b\circ f$ in the pointwise order. It is well-known that
compatibility entails soundness. Another useful property is:
\begin{equation}\label{eq:inclusion}
  \text{if }f \text{ is \kl{compatible}, then }f(\nu b)\leq \nu b\,.
\end{equation}

Following~\cite{bppr:up-to-fibration}, we assume hereafter that $b$
can be seen as the composite
\begin{equation}
  \label{eq:metric-b-fib}
  \begin{tikzcd}
    b\colon\fibreRel{X}\ar[r,"\barF"]&\fibreRel{FX}\ar[r,"\reindexnobrk{\xi}"]&\fibreRel{X}\,.
  \end{tikzcd}
\end{equation}
where $\xi\colon X\to FX$ is some coalgebra for
$F\colon \Set \to \Set$. \AP When $F$ admits a final coalgebra
$\omega\colon \Omega\to F\Omega$, the unique morphism
$! \colon X \to \Omega$ induces the \intro{behavioural closure} up-to
technique
\begin{equation}
  \label{eq:up-to-beh}
  \begin{tikzcd}
    \cbhv\colon\fibreRel{X}\ar[r,"\directImage{!}"]&\fibreRel{\Omega}\ar[r,"\reindexnobrk{!}"]&\fibreRel{X}
  \end{tikzcd} 
\end{equation}
where
$\cbhv(p)(x,y)=\bigvee\{p(x',y') \mid \ !(x)=\ !(x') \text{ and }
!(y)=\ !(y')\}$. For $\VV=2$, behavioural closure corresponds to the
usual up-to behavioural equivalence (bisimilarity). \AP Other
immediate generalisations are the \intro{up-to reflexivity} ($\cref$),
\intro{up-to transitivity} ($\ctrn$) and \intro{up-to symmetry}
($\csym$) techniques. Whenever $\barF$ is obtained through the
Wasserstein construction of some $\liftF$ satisfying the conditions of
Theorem~\ref{thm:restricting-Wasserstein}, these techniques are
compatible \full{(see Appendix \ref{app:basic} for more
  details)}\short{(see
  \cite{bkp:up-to-behavioural-metrics-fibrations-arxiv} for more
  details)}.

\AP As usual, compatible techniques can be combined together either by
function composition ($\circ$) or by arbitrary joins ($\bigvee$). For
instance compatibility of \intro{up-to metric closure}, defined as the
composite $\cmtr=\ctrn \circ \csym \circ \cref$ follows immediately
from compatibility of $\ctrn$, $ \csym$ and $\cref$.  In $\Vrel$ there
is yet another useful way to combine up-to techniques -- called
\emph{chaining} in~\cite{cgv:up-to-bisim-metrics} -- and defined as
the \kl{composition} ($\comp$) of
relations.

\begin{restatable}{prop}{proptns}\label{prop:tns}
  Let $f_1,f_2\colon \Vrel_X \to \Vrel_X$ be compatible
  with respect to $b\colon \Vrel_X \to \Vrel_X$.  If
  $\barF(p\comp q) \geq \barF(p) \comp \barF(q)$ for all
  $p,q\in \Vrel_X$, then $f_1\comp f_2$ is $b$-compatible.
\end{restatable}

\AP In the reminder of this section, we focus on quantitative
generalizations of the \intro{up-to contextual closure} technique,
which given an \kl[algebra for a functor]{algebra}
$\alpha\colon TX\to X$, is seen as the composite:
\begin{equation}
  \label{eq:metric-f-fib}
  \begin{tikzcd}
    f\colon\fibreRel{X}\ar[r,"\barT"]&\fibreRel{TX}\ar[r,"\directImage{\alpha}"]&\fibreRel{X}\,.
  \end{tikzcd}
\end{equation}
\begin{ex}\label{ex:ctx}
  Consider a signature $\Sigma$ and the algebra of $\Sigma$-terms with
  variables in $X$
  $\mu_X \colon T_{\Sigma} T_{\Sigma} X \to T_{\Sigma}X$. The
  contextual closure
  $\cctx\colon \fibreRel{T_{\Sigma}X} \to \fibreRel{T_{\Sigma}X}$ is
  defined as in~\eqref{eq:metric-f-fib} by taking the canonical
  lifting of the functor $T_{\Sigma}$. For all $t_1,t_2\in T_\Sigma X$
  and $d\in \fibreRel{T_{\Sigma}X}$
  \begin{equation*} \label{eq:ctx} ctx(d)(t_1,t_2)=\bigvee_{C} \{
    \bigwedge_j d(s_j^1,s_j^2) \mid t_i=C(s_0^i, \dots, s_n^i) \}
  \end{equation*}
  where $C$ ranges over arbitrary contexts and $s_j^i$ over
  terms. Notice that for $\VV=2$, this boils down to the usual notion
  of contextual closure of a relation. \full{All details are in
    Appendix~\ref{app:ctx}.}\short{Details can be found in
    \cite{bkp:up-to-behavioural-metrics-fibrations-arxiv}.}
\end{ex}

\begin{ex}\label{ex:cvx}
  Let $\VV = [0,1]$. In~\cite{cgv:up-to-bisim-metrics}, the \intro[up-to convex
  closure]{convex closure} of $d \in \fibreRel{\Distr(X)}$ is defined
  for $\Delta,\Theta \in \Distr(X)$ as 
  \begin{equation*} \label{eq:cvx} \ccvx(d)(\Delta,\Theta) = \inf
    \{\sum_i p_i\cdot d(\Delta_i,\Theta_i) \mid \Delta = \sum_i
    p_i\cdot \Delta_i, \Theta = \sum_i p_i\cdot \Theta_i\}
  \end{equation*}
  where $\Delta_i,\Theta_i\in\Distr(X)$, $p_i\in[0,1]$.
  This can be obtained as in~\eqref{eq:metric-f-fib} by taking the
  lifting of $\Distr$ from Example~\ref{ex:wasserstein} and the
  algebra given by the multiplication
  $\mu_X\colon \Distr \Distr X \to \Distr X$. \full{All details are in
    Appendix \ref{app:cvx}.}\short{Details can be found in
    \cite{bkp:up-to-behavioural-metrics-fibrations-arxiv}.}
\end{ex}

We consider next systems modelled as \kl[bialgebra]{bialgebras}
$(X,\alpha\colon TX\to X,\xi\colon X\to FX)$ for a natural
transformation $\distrLaw\colon T\circ F\To F\circ T$.  When $b$ and
$f$ are as in~\eqref{eq:metric-b-fib},
respectively~\eqref{eq:metric-f-fib}, we
use~\cite[Theorem~2]{bppr:up-to-fibration} to obtain
\begin{prop}
  \label{prop:soundness-follows-from-lift-distr}
  If there exists a lifting
  $\barDistr\colon\barT\circ\barF\To\barF\circ\barT$ of $\distrLaw$,
  then $f$ is \kl{$b$-compatible}.
\end{prop}

The next theorem establishes sufficient conditions for the existence
of a lifting of $\distrLaw$. 
\begin{restatable}{theo}{thmLiftingDistrLawWass}
  \label{thm:Lifting-Distr-Law-Wass}
  Assume the natural transformation
  $\distrLaw\colon T\circ F\To F\circ T$ lifts to a natural
  transformation
  $\liftDistr\colon\liftT\circ\liftF\To\liftF\circ\liftT$ and that we
  have
  $\liftT\circ\directImage{\lambdaChBaseF{X}}
  \le\directImage{T\lambdaChBaseF{X}}\circ\liftT$.  Then $\distrLaw$
  lifts to a distributive law
  $\barDistr\colon\barT\circ\barF\To\barF\circ\barT$.
\end{restatable}

\begin{proof}[Proof Sketch]
  \AP Notice that $\liftTcF:=\liftT\circ\liftF$ and
  $\liftFcT:=\liftF\circ\liftT$ are liftings of the composite functors
  $T\circ F$, respectively $F\circ T$. We will denote by
  $\intro\barTF$ and $\intro\barFT$ the corresponding \kl{Wasserstein
    liftings} obtained from $\liftTcF$, respectively $\liftFcT$ as in
  Section~\ref{sec:wass-lifting}.  We split the proof obligation into
  three parts:
  \[ \barT\circ \barF \stackbin[\bf (1)]{}{\Rightarrow} \barTF
    \stackbin[\bf (2)]{\tilde{\zeta}}{\Rightarrow} \barFT
    \stackrel[\bf (3)]{}{\Rightarrow} \barF\circ\barT\,. \]
  \begin{description}
  \item[(1)] lifts the identity natural transformation on $T\circ F$.
    Its existence is proved using the hypothesis
    $\liftT\circ\directImage{\lambdaChBaseF{X}}
    \le\directImage{T\lambdaChBaseF{X}}\circ\liftT$\full{, see
    Lemma~\ref{lem:tilde-sigma} in Appendix~\ref{app:liftingdistr}}.
  \item[(2)] is obtained by applying
    Lemma~\ref{lem:Wass-is-Functorial} to $\liftDistr$. Such liftings
    have already been studied
    in~\cite{bbkk:trace-metrics-functor-lifting}.
  \item[(3)] lifts the identity natural transformation on $F\circ T$.
    \full{It exists by Lemma~\ref{lem:sigma-tilde} in
      Appendix~\ref{app:liftingdistr}.}\qedhere
  \end{description}
\end{proof}

The first requirement of the previous theorem holds for the canonical
$\Vpred$-liftings under mild assumptions on $F$ and $T$.

\begin{restatable}{prop}{thmLiftNatTransCan}
  \label{thm:lift-nat-transf-canonical}
  Assume that $\zeta\colon T\circ F\Rightarrow F\circ T$ is a natural
  transformation and that, furthermore, $T$ preserves weak pullbacks
  and $F$ preserves intersections. Then $\zeta$ lifts to a natural
  transformation
  $\widehat{\zeta}\colon\liftTcan\circ\liftFcan\To\liftFcan\circ\liftTcan$.
\end{restatable}

The next proposition establishes sufficient conditions for the second
hypothesis of Theorem~\ref{thm:Lifting-Distr-Law-Wass}. We need a
property on $\VV$ that holds for the quantales in
Example~\ref{ex:quantale} and was also assumed
in~\cite{h:closed-objects}. \AP Given $u,v\in\VV$ we write
$u\intro\VVlll v$ ($u$ is totally below $v$) if for every
$W\subseteq \mathcal{V}$, $v\le \bigvee W$ implies that there exists
$w\in W$ with $u\le w$.  The quantale $\VV$ is \intro{constructively
  completely distributive} iff for all $v\in\VV$ it holds that
$v = \bigvee\{u\in\VV\mid u\VVlll v\}$. In
\full{Appendix~\ref{app:ccdq}}\short{\cite{bkp:up-to-behavioural-metrics-fibrations-arxiv}}
we prove a more general statement in which the lifting of $T$ is not
assumed to be the canonical one, that is useful to guarantee the
result for interesting liftings, such as the one in
Example~\ref{ex:wasserstein}.

\begin{restatable}{prop}{propccdq}\label{prop:ccdq}
  Assume that $T$ preserves weak pullbacks and that $\VV$ is
  \kl{constructively completely distributive}. Then
  $\liftTcan\circ \directImage{f} \le
  \directImage{Tf}\circ \liftTcan$.\end{restatable}

Combining Theorem~\ref{thm:Lifting-Distr-Law-Wass} and
Propositions~\ref{prop:soundness-follows-from-lift-distr},~\ref{thm:lift-nat-transf-canonical}
and~\ref{prop:ccdq} we conclude:

\begin{restatable}{theo}{thmsummary}\label{thm:summary}
  Let $(X,\alpha\colon TX\to X,\xi\colon X\to FX)$ be a \kl{bialgebra}
  for a natural transformation
  $\zeta\colon T\circ F\Rightarrow F\circ T$.  If $\VV$ is
  \kl{constructively completely distributive}, $T$ preserves weak
  pullbacks and $F$ preserves intersections, then
  $f=\barTcan \circ \directImage{\alpha}$ is compatible
  with respect to $b=\barFcan \circ \reindexnobrk{\xi}$.
\end{restatable}

When $\alpha$ is the free algebra for a signature
$\mu_X\colon T_{\Sigma}T_\Sigma X \to T_{\Sigma}X$ (as in
Example~\ref{ex:ctx}), the above theorem guarantees that \kl{up-to
  contextual closure} is compatible with respect to $b$. By
\eqref{eq:inclusion}, the following holds.

\begin{restatable}{cor}{cornonex}\label{cor:nonex}
  For all terms $t_1,t_2$ and unary contexts $C$,
  $\nu b (t_1,t_2) \leq \nu b(C(t_1),C(t_2))$.
\end{restatable}
For $\VV=2$, since the canonical quantitative lifting coincides with
the canonical relational one, then $\nu b$ is exactly the standard
coalgebraic notion of behavioural
equivalence~\cite{hj:structural-induction-coinduction}. Therefore the
above corollary just means that behavioural equivalence is a
congruence.

For $\VV=[0,\infty]$ instead, this property boils down to
\emph{non-expansiveness} of contexts with respect to the behavioural
metric. It is worth to mention that this property often fails in
probabilistic process algebras when taking the standard Wasserstein
lifting which, as shown in Example~\ref{ex:wasserstein}, is \emph{not}
the canonical one. We leave as future work to explore the implications
of this insight.

\section{Example: distance between regular languages}
\label{sec:examples}

We will now work out the quantitative version of the up-to
congruence technique for non-deterministic automata. We consider the
\kl{shortest-distinguishing-word-distance} $\sdwDist$, proposed in
Section~\ref{sec:motivating-example}.  As explained, we will assume an
on-the-fly determinization of the non-deterministic automaton,
i.e., formally we will work with a coalgebra that corresponds to a
deterministic automaton on which we have a join-semilattice structure.

We explain next the various ingredients of the example:

\pseudopar{Coalgebra and algebra. } As outlined in
Section~\ref{sec:motivating-example} and Example~\ref{ex:NFA-as-bialg}
the determinization of an NFA with state space $Q$ is a \kl{bialgebra}
$(X,\alpha,\xi)$ for the \kl{distributive law}
$\distrLawNFA_X\colon \Pow(2\times X^A)\to 2\times (\Pow X)^A$, where
$X = \Pow Q$, $\alpha\colon\Pow X\to X$ is given by union and
$\xi\colon X\to 2\times X^A$ specifies the DFA structure of the
determinization.  Hence, we instantiate the generic results in the
previous section with $TX=\Pow X$, $FX=2\times X^A$ and $\distrLaw$
as defined in Example~\ref{ex:NFA-as-bialg}.

\pseudopar{Lifting the functors. } We take the quantale $\VV = [0,1]$
(Example~\ref{ex:quantale}) and consider the \kl{Wasserstein
  liftings} of the endofunctors $F$ and $T$ to $\Vrel$ corresponding
to the following evaluation maps:
\begin{itemize}
\item $\intro\evNFAF(b,f) := c\cdot \max_{a\in A} f(a)$, where
  $b\in \{0,1\}$, $f\colon A\to [0,1]$ and $c$ is the constant used
  in $\sdwDist$, and,
\item $\intro\evNFAT:=\evPcan=\sup$, the canonical evaluation map as in
  Example~\ref{ex:pow-funct-can-ev}.
\end{itemize}
These are \kl{monotone evaluation maps} that satisfy the hypothesis of
Theorem~\ref{thm:restricting-Wasserstein}\full{ (see
  Lemma~\ref{lem:well-behaved} in
  Appendix~\ref{sec:proofs-examples})}. Hence the corresponding
\kl{Wasserstein liftings} restrict to $\Vcat$. We computed the
\kl{Wasserstein lifting} of $T=\Pow$ in Example~\ref{ex:hausdorff}:
applying the lifted functor $\barT$ to a map
$d\colon X\times X\to [0,1]$, gives us the \kl{Hausdorff distance},
i.e., $\barT(d)(X_1,X_2) = d^H(X_1,X_2)$, where $X_1,X_2\subseteq X$
and $d^H$ denotes the \kl{Hausdorff metric} based on $d$.  On the
other hand, the Wasserstein lifting of $F$ corresponding to $\evNFAF$
associates to a metric $d\colon X\times X\to [0,1]$ the metric
$\barF(d)\colon FX\times FX \to [0,1]$ given by
\[
  ((b_1,f_1),(b_2,f_2))  \mapsto  \left\{
    \begin{array}{ll}
      1 & \mbox{if $b_1\neq b_2$} \\
      \max\limits_{a\in A} c\cdot \{ d(f_1(a),f_2(a)) \} &
      \mbox{otherwise} 
    \end{array}
    \right.
\]

\pseudopar{Fixpoint equation. } The map $b$ for the fixpoint equation
was defined in Section~\ref{sec:up-to-techniques} as the composite
$\reindexnobrk{\xi}\circ\barF$. Using the above lifting $\barF$, this
computation yields exactly the map $b$ defined in
~\eqref{eq:def-of-b-NFA}, whose largest fixpoint (smallest with
respect to the natural order on the reals) is the
\kl{shortest-distinguishing-word-distance} introduced in
Section~\ref{sec:motivating-example}.

\pseudopar{Up-to technique.} The next step is to determine the map
$f$ introduced in Section~\ref{sec:up-to-techniques} for the up-to
technique and defined as the composite
$\directImage{\alpha}\circ\barT$ on $\Vrel$. Combining the definition
of the \kl{direct image} functors on $\Vrel$ with the lifting $\barT$,
we obtain for a given a map $d\colon X\times X\to [0,1]$ that 
\[ f(d)(x_1,x_2) = \inf \{ d^H(X_1,X_2) \mid X_1,X_2 \subseteq X,
  \alpha(X_i) = x_i\} \]
To show that $f(d)(Q_1,Q_2) \leR r$ for two sets
$Q_1,Q_2\subseteq Q$ (i.e.,\ $Q_1,Q_2\in X$) and a constant $r$ we use
the following rules:
\begin{center}
    \begin{tabular}{c}
    $f(d)(\emptyset,\emptyset)\leR r$ \quad
    $\begin{array}{c}
      d(Q_1,Q_2) \leR r \\ \hline
      f(d)(Q_1,Q_2) \leR r
    \end{array}$ \quad
    $\begin{array}{c}
      f(d)(Q_1,Q_2) \leR r \quad f(d)(Q'_1,Q'_2) 
      \leR r \\ \hline
      f(d)(Q_1\cup Q'_1,Q_2\cup Q'_2) \leR r
    \end{array}$ 
  \end{tabular}
\end{center}

\pseudopar{Lifting of distributive law.} In order to prove that the
distributive law lifts to $\Vrel$ and hence that the up-to technique
is \kl{sound} by virtue of
Proposition~\ref{prop:soundness-follows-from-lift-distr}, we can prove
that the two conditions of Theorem~\ref{thm:Lifting-Distr-Law-Wass}
are met by the $\Vpred$ liftings of $F$ and $T$ corresponding to the
evaluation maps $\evNFAF$ and $\evNFAT$, see
\full{Lemma~\ref{lem:Lifting-of-Distributive-Law-NFA} in
  Appendix~\ref{sec:proofs-examples}}\short{\cite{bkp:up-to-behavioural-metrics-fibrations-arxiv}}.

Everything combined, we obtain a \kl{sound} up-to technique, which
implies that the reasoning in Section~\ref{sec:motivating-example} is
valid. Furthermore, as the example shows, the up-to technique can
significantly simplify behavioural distance arguments and speed up
computations.

\section{Related and future work}
\label{sec:conclusion}

Up-to techniques for behavioural metrics in a probabilistic setting
have been considered in~\cite{cgv:up-to-bisim-metrics} using a
generalization of the Kantorovich
lifting~\cite{DBLP:conf/concur/ChatzikokolakisGPX14}.
In Section~\ref{sec:up-to-techniques}, we have shown that the basic
techniques introduced in~\cite{cgv:up-to-bisim-metrics} (e.g., metric
closure, convex closure and contextual closure) as well as the ways to
combine them (composition, join and chaining) naturally fit within our
framework.
The main difference with our approach---beyond the fact that we
consider arbitrary coalgebras while the results
in~\cite{cgv:up-to-bisim-metrics} just cover coalgebras for a fixed
functor---is that the definition of up-to techniques and the criteria
to prove their soundness do not fit within the standard framework
of~\cite{ps:enhancements-coind-proofs}. Nevertheless, as illustrated
by a detailed comparison in \full{Appendix
  \ref{app:comparison}}\short{\cite{bkp:up-to-behavioural-metrics-fibrations-arxiv}},
the techniques of~\cite{cgv:up-to-bisim-metrics} can be reformulated
within the standard theory and thus proved sound by means of our
framework. An important observation brought to light by compositional
methodology inherent to the fibrational approach, is that for
probabilistic automata a bisimulation metric up-to convexity in the
sense of~\cite{cgv:up-to-bisim-metrics} is just a bisimulation metric,
see
\full{Lemma~\ref{lem:useless-up-to-conv}}\short{\cite{bkp:up-to-behavioural-metrics-fibrations-arxiv}}. Nevertheless,
the \kl{up-to convex closure} technique can find meaningful
applications in linear, trace-based behavioural metrics
(see~\cite{bbkk:trace-metrics-functor-lifting}).

The Wasserstein (respectively Kantorovich) lifting of the distribution
functor involving couplings was first used for defining behavioural
pseudometrics using final coalgebras
in~\cite{DBLP:conf/icalp/BreugelW01}.  Our work is based instead on
liftings for arbitrary functors, a problem that has been considered in
several works (see
e.g.~\cite{h:closed-objects,bkv:extensions-v-cat,bbkk:coalgebraic-behavioral-metrics,ks:codensity-liftings-monads}),
despite with different shades. The closest to our approach are
~\cite{h:closed-objects}
and~\cite{bbkk:coalgebraic-behavioral-metrics} that we discuss next.

In~\cite{h:closed-objects} Hofmann introduces a generalization of the
Barr extension (of $\Set$-functors to $\Rel$), namely he defines
extensions of $\Set$-monads to the bicategory of $\VV$-matrices, in
which $0$-cells are sets and the $\VV$-relations are $1$-cells. Some
of the definitions and techniques do overlap between the developments
in~\cite{h:closed-objects} and the results we presented in
Section~\ref{sec:wass-lifting}. However, there are also some (subtle)
differences which would not allow us to use off the shelf his
results.

First, in order to reuse the results in~\cite{bppr:up-to-fibration},
we need to recast the theory in a fibrational setting, rather than the
bicategorical setting of~\cite{h:closed-objects}. The definition of
\emph{topological theory}~\cite[Definition~3.1]{h:closed-objects}
comprises what we call an \kl{evaluation map}, but which additionally
has to satisfy various conditions.  An important difference with what
we do is that the condition $(Q_{\bigvee})$ in the aforementioned
definition entails that the predicate lifting one would obtain from
such an evaluation map would be an \emph{opfibred lifting}, rather
than a \kl{fibred lifting} as in our setting. Indeed, the condition
$(Q_{\bigvee})$ can be equivalently expressed in terms of a natural
transformation involving the \emph{covariant} functor $P_\VV$, as
opposed to the \emph{contravariant} one $\VV^{-}$ that we used in
Section~\ref{sec:Vpred-liftings}. Lastly, in our framework we need to
work with arbitrary functors, not necessarily carrying a monad
structure.

In~\cite{bbkk:coalgebraic-behavioral-metrics} we provided a generic
construction for the Wasserstein lifting of a functor to the category
of pseudo-metric spaces, rather than on arbitrary quantale-valued
relations. The realisation that this construction is an instance as a
change-of-base situation between $\Vrel$ and $\Vpred$ allows us to
exploit the theory in \cite{bppr:up-to-fibration} for up-to techniques
and, as a side result, provides simpler (and cleaner) conditions for
the restriction $\Vcat$ (Theorem \ref{thm:restricting-Wasserstein}).

We leave for future work several open problems. What is a universal
property for the canonical \kl{Wasserstein lifting}?  Secondly, can
the \kl{Wasserstein liftings} presented here be captured in the
framework of~\cite{bkv:extensions-v-cat}
or~\cite{ks:codensity-liftings-monads}?  We also leave for future work
the development of up-to techniques for other quantales than $2$ and
$[0,1]$. We are particularly interested in weighted
automata~\cite{dkv:weighted-automata} over quantales and in
conditional transition systems, a variant of featured transition
systems.

\short{\bibliography{references}}

\full{

}

\short{\end{document}}

\appendix

\section{A detailed comparison with
  \cite{cgv:up-to-bisim-metrics}}\label{app:comparison}
In this appendix we discuss in details the relationship between our
work and~\cite{cgv:up-to-bisim-metrics} where a general framework of
up-to techniques for behavioural metric is introduced.

The systems of interest in~\cite{cgv:up-to-bisim-metrics} are
probabilistic automata which are
known~\cite{DBLP:journals/tcs/BartelsSV04} to be coalgebras for the
functor $\Pow(A \times \Distr(-))$. The behavioural metrics under
consideration are defined as the greatest fixed points of
\begin{equation}
  \label{eq:metric-b-fib2}
  \begin{tikzcd}
    b\colon\fibreRel{X}\ar[r,"K"]&\fibreRel{\Distr(X)}
    \ar[r,"\overline{\Pow(A \times -)}"]
    &\fibreRel{\Pow(\Distr(X))^A}\ar[r,"\reindexnobrk{\xi}"]&\fibreRel{X}
  \end{tikzcd}
\end{equation}
where $\xi \colon X \to \Pow(A \times \Distr(X))$ is a probabilistic
automaton, $\overline{\Pow(A \times -)}$ is the canonical lifting of
$\Pow(A \times -)$ (based on the Hausdorff distance,
Example~\ref{ex:hausdorff}) and $K$ is some lifting of $\Distr$.
Please note that the quantale $\VV$ in~\cite{cgv:up-to-bisim-metrics}
is $[0,\infty]$ (Example \ref{ex:quantale}) so the ordering used in
this paper and the one in~\cite{cgv:up-to-bisim-metrics} are always
inverted.

Observe that the definition of $b$ as in \eqref{eq:metric-b-fib2} is
an instance of \eqref{eq:metric-b-fib} by taking
$\barF= K \circ \overline{\Pow(A \times -)}$. It is worth to mention
that $K$ in~\cite{cgv:up-to-bisim-metrics} is not arbitrary, but it is
supposed to be an instance of a parametric construction called
\emph{generalized Kantorovic metric}. For a certain value of the
parameter, this coincides (via the well known duality) with the
Wasserstein metric from transportation theory
(Example~\ref{ex:wasserstein}).

The authors of~\cite{cgv:up-to-bisim-metrics} introduced several basic
techniques ---which can be easily defined in our framework, e.g.,
metric closure (Section \ref{sec:up-to-techniques}), convex closure
(Example~\ref{ex:cvx}) or contextual closure (Example~\ref{ex:ctx})---
and combine them via composition ($\circ$), supremum ($\bigvee$) and
chaining ($\comp$). In Proposition~\ref{prop:tns}, we have provided
sufficient conditions ensuring that $\comp$ preserves
compatibility. The same result for $\circ$ and $\bigvee$ follows
immediately from the standard theory of compatible up-to
techniques~\cite{ps:enhancements-coind-proofs}. This is not the case
for~\cite{cgv:up-to-bisim-metrics}, where these results need novel
proofs since the basic notions of up-to techniques and compatibility
(or respectfullness) do not fit within the standard lattice-theoretic
framework.

Indeed in~\cite{cgv:up-to-bisim-metrics}, an up-to technique is
defined to be some map $f$ of type
$\fibreRel{\Distr(X)}\to \fibreRel{\Distr(X)}$ and a bisimulation
up-to $f$ to be some $d \in \fibreRel{X}$ such that
$d \le (\reindexnobrk{\xi} \circ \overline{\Pow}) \circ f \circ K
d$.\footnote{Note that
  $d \le \reindexnobrk{\xi} \circ \overline{\Pow(A \times -)} d'$
  means, in the language of~\cite{cgv:up-to-bisim-metrics}, that $d$
  progresses to $d'$.} Soundness is defined in the expected way. The
notion to prove soundness (Definition 5 in
\cite{cgv:up-to-bisim-metrics}) amounts to the following, modulo the
usual difference between compatibility and respectfullness (that is
well-known and deeply discussed in several papers
\cite{bppr:general-coinduction-up-to,DBLP:conf/lics/Pous16})
\begin{defi}
  A monotone map
  $f \colon \fibreRel{\Distr(X)}\to \fibreRel{\Distr(X)}$ is a
  \emph{well-behaved} up-to technique iff there exists an
  $f'\colon \fibreRel{X} \to \fibreRel{X}$ such that
  \begin{enumerate}
  \item $f \circ K \leq K \circ f'$ and
  \item $f' \circ b \leq b \circ f' $.
  \end{enumerate}
\end{defi}
Observe that whenever $f$ is well-behaved, a bisimulation up-to $f$ in
the sense of~\cite{cgv:up-to-bisim-metrics}, can be transformed into a
bisimulation up-to $f'$ in our sense by mean of the first item:
\begin{equation}
  \label{eq:fprim}
  d \le (\reindexnobrk{\xi} \circ \overline{\Pow(A \times -)}) \circ f \circ K d \le (\reindexnobrk{\xi} \circ \overline{\Pow(A \times -)}) \circ K \circ f' d = bf'd\,.
\end{equation} 
Moreover, thanks to the second item, $f'$ is compatible w.r.t. $b$.

This observation shows that the techniques
in~\cite{cgv:up-to-bisim-metrics} can be reformulated within the
standard theory of~\cite{ps:enhancements-coind-proofs} and thus proved
compatible by means of our framework.

\begin{lem}
  \label{lem:useless-up-to-conv}
  Consider a probabilistic automaton and let $K$ denote a convex (in
  the sense of~\cite{cgv:up-to-bisim-metrics}) lifting of the
  probability distribution functor.  Then a bisimulation metric up-to
  convex closure in the sense of~\cite{cgv:up-to-bisim-metrics} is
  just a bisimulation metric, i.e., a post-fixpoint of $b$
  in~\eqref{eq:metric-b-fib2}.
\end{lem}

\begin{proof}
  As above, let $\xi\colon X\to \Pow(A \times \Distr(X))$ denote the
  coalgebra structure corresponding to the probabilistic
  automaton. The \kl{up-to convex closure} is defined as in
  Example~\ref{ex:cvx}. Recall that a bisimulation metric \kl{up-to
    convex closure} in the sense of~\cite{cgv:up-to-bisim-metrics} is
  a bisimulation metric $d$ such that $d$ \emph{progresses} to
  $\ccvx\circ K(d)$, written using the notation
  in~\cite[Definition~2]{cgv:up-to-bisim-metrics} as
  $d\rightarrowtail \ccvx\circ K(d)$. Spelling out that definition, we
  obtain that, in the quantale order (i.e., the reversed of the order
  on the reals used~\cite{cgv:up-to-bisim-metrics}) we have
  \begin{equation}
    \label{eq:comp-cat-1}
    d\le \reindexnobrk{\xi}\circ \overline{\Pow(A \times -)}\circ
    \ccvx\circ K(d)\,.
  \end{equation}
  On the other hand, the respectfulness of $\ccvx$---established
  via~\cite[Theorem~11]{cgv:up-to-bisim-metrics}---uses the fact that
  for all $d\in\Vrel_X$ we have that $K(d)$ is convex, hence the $f'$
  used above is simply the identity function on $\Vrel_{\Distr X}$.
  In other words we have
  \begin{equation}
    \label{eq:comp-cat-2}
    \ccvx\circ K(d)\le K(d)
  \end{equation}
  Combining~\eqref{eq:comp-cat-1} and~\eqref{eq:comp-cat-2} we obtain
  that
  \[
    d\le \reindexnobrk{\xi}\circ \overline{\Pow(A \times -)}\circ
    K(d)\,,
  \]
  or equivalently, that $d$ is simply a bisimulation metric.
\end{proof}

\section{Proofs and additional material}
\label{sec:proofs}

\subsection{Proofs and additional material for
  Section~\ref{sec:lifting-vrel-vcat}}
\label{sec:proofs-preliminaries}

We will use the the Beck-Chevalley condition for \kl{fibrations}
$p:\P\to\B$, which will be needed in some of the proofs. Assume we
have a commuting square:
\begin{equation}
  \label{eq:os-BCC}
  \xymatrix{
    A\ar[r]^-f\ar[d]_-u& B\ar[d]^-v\\
    C\ar[r]_-{g}& D\\
  }
\end{equation}
Since the \kl{fibration} is split we have a commuting diagram
\[
  \xymatrix{
    \P_A& \P_B\ar[l]_-{f^*}\\
    \P_C\ar[u]^-{u^*}& \P_D\ar[u]_-{v^*}\ar[l]^-{g^*}\\
  }
\]
Using the adjunctions $\directImage{f}\dashv \reindexf$ and
$\directImage{g}\dashv \reindexg$ we obtain the so-called mate of the
above square
\begin{equation}
  \label{eq:mate-BCC}
  \xymatrix{
    \P_A\ar[r]^-{\directImage{f}}\ar@{=>}[rd]& \P_B\\
    \P_C\ar[u]^-{\reindexnobrk{u}}\ar[r]_-{\directImage{g}}& \P_D\ar[u]_-{\reindexnobrk{v}}\\
  }
\end{equation}
obtained using the unit and the counit of the above adjunctions, as
the composite
\[
  \begin{tikzcd}[column sep=1em]
    \directImage{f}\reindexnobrk{u}\ar[rr,Rightarrow,"{\directImage{f}\reindexnobrk{u}\eta}"]&\
    &\directImage{f}\reindexnobrk{u}\reindexg\directImage{g}\ar[r,Rightarrow]&\directImage{f}\reindexf\reindexnobrk{v}\directImage{g}\ar[rr,Rightarrow,"{\varepsilon
      \reindexnobrk{v}\directImage{g}}"]&\
    &\reindexnobrk{v}\directImage{g}
  \end{tikzcd}
\]

\begin{defi}
  \label{def:beck-chevalley}
  We say that the square~\eqref{eq:os-BCC} has the
  \intro{Beck-Chevalley} condition if the mate~\eqref{eq:mate-BCC} is
  an isomorphism.
\end{defi}

\begin{ex}
  The bifibration $\Vpred\to\Set$ has the \kl{Beck-Chevalley}
  condition for weak pullback squares in $\Set$.  Essentially we have
  to show that if~\eqref{eq:os-BCC} is a weak pullback, then for every
  $p\in\Vpred_C$ and $b\in B$ we have
  \begin{equation}
    \label{eq:1}
    \bigvee\limits_{a\in f^{-1}(b)} p(u(a))= \bigvee\limits_{c\in g^{-1}(v(b))}p(c) 
  \end{equation}
  Proving $\le$ is easy (we just use that the square commutes), but
  for $\ge$ we need that~\eqref{eq:os-BCC} is a weak pullback.
\end{ex}

\subsection{Proofs and additional material for
  Section~\ref{sec:wass-lifting}}
\label{sec:proofs-wass-lifting}

\subsubsection{$\VV$-predicate liftings}
\label{app:predlift}

In order to state the following proposition we first have to spell out
what it means for an evaluation map to be monotone. For this, we first
define an (order) relation $\orderFV$ on $F\VV$.

\begin{defi}[Relation $\orderFV$ on $F\VV$]
  \label{def:orderFV}
  We define a relation $\intro{\orderFV}$ on $F\VV$: let
  $v_1,v_2\in F\VV$. We define $v_1\orderFV v_2$ whenever
  \[ \exists r\in F[\le] \mbox{ s.t. } F(\pi_1\circ o)(r) = v_1\mbox{
      and } F(\pi_2\circ o)(r) = v_2 \] The relation $\orderFV$ will
  also be denoted by $\le^F$ (order $\le$ lifted under $F$ via the
  standard relation lifting).
\end{defi}

According to \cite{bk:finitary-functors-set-preord-poset} relation
lifting transforms preorders into preorders whenever $F$ preserves
weak pullbacks (but not necessarily orders into orders).

\propLiftCorrespondences*

\begin{proof}
  The equivalence of the first two bullets is well-known in
  coalgebraic modal logic for $\VV=2$. For the sake of completeness we
  include here full details.
  
  $\liftF$ is a lifting of $F$ to $\Vpred$ if and only if the
  following two conditions are met for all sets $X$ and functions
  $f\colon X\to Y$:
  \begin{enumerate}
  \item\label{it-lift-1}
    $\liftFrest{X}\colon\fibrePred{X}\to\fibrePred{FX}$ is monotone,
    and,
  \item\label{it-lift-2} the inequality
    $\liftFrest{X}\circ \reindexf(R)\le\reindex{Ff}\circ\liftFrest{Y}$
    holds.
  \end{enumerate}

  These two conditions alone are equivalent to the laxness of the
  following square
  \[
    \begin{tikzcd}
      \VV^Y\ar[d,swap,"\liftFrest{Y}"]\ar[r,"f"]&\VV^X\ar[d,"\liftFrest{X}"]\ar[dl,color=white,"\textcolor{black}{\rotle}" description]\\
      \VV^{FY}\ar[r,swap,"Ff"]& \VV^{FX}
    \end{tikzcd}
  \]
  However, $\liftF$ is a \kl{fibred lifting} of $F$ if and only if
  item~\ref{it-lift-1} holds and the inequality in
  item~\ref{it-lift-2} above is in fact an equality. Hence $\liftF$ is
  a \kl{fibred lifting} if and only if the above square is actually
  commutative, which ammounts to the existence of a natural
  transformation $\gamma\colon\VV^{-}\to\VV^{F-}$ with each component
  $\gamma_X$ being \emph{monotone}.

  We have thus proved the equivalence of the two first bullets. Now,
  let us turn to the third one. By Yoneda lemma we know that natural
  transformations $\VV^{-}\to\VV^{F-}$ are in one-to-one
  correspondence with evaluation maps $\ev\colon F(\VV)\to\VV$. It
  remains to characterize the monotonicity condition.  We show that
  this is equivalent to requiring that $\ev_F\colon F\VV\to \VV$ is
  monotone for the order $\orderFV$ on $F\VV$ and $\le$ on $\VV$.

  \noindent ``$\Leftarrow$'' Assume that $\ev_F$ is monotone and take
  $f_1,f_2\colon X\to \VV$ such that $f_1\le f_2$.  This means that
  $\pair{f_1}{f_2}$ factors through $o$ as depicted below, where
  $u\colon X\to [\le]$ is defined as $u(x) = (f_1(x),f_2(x))$.
  \[
    \xymatrix{ X \ar[r]_-u \ar@/^1pc/[rr]^{\pair{f_1}{f_2}} & {[\le]}
      \ar@{^{(}->}[r]_-o & \VV\times \VV }
  \]
  If we apply $F$ to the diagram above and post-compose with
  $F\pi_1,F\pi_2,\ev_F$, we obtain the following diagram.
  \[
    \xymatrix{ FX \ar[r]_-{Fu} \ar@/^1pc/[rr]^{F\pair{f_1}{f_2}} &
      {F[\le]} \ar@{^{(}->}[r]_-{Fo} & F(\VV\times \VV)
      \ar@<1ex>[r]^-{F\pi_1} \ar@<-1ex>[r]_-{F\pi_2} & F\VV
      \ar[r]^-{\ev_F} & \VV }
  \]
  Let $t\in FX$. Our aim is to show $\liftF f_1(t) \le \liftF f_2(t)$,
  which implies $\liftF f_1 \le \liftF f_2$.

  First, define $r = Fu(t) \in F[\le]$. Now observe that
  $F(\pi_1\circ o)(r) = F(\pi_1\circ o\circ u)(t) = F(\pi_1\circ
  \pair{f_1}{f_2})(t) = Ff_1(t)$. Analogously,
  $F(\pi_2\circ o)(r) = Ff_2(t)$. Hence $Ff_1(t) \orderFV Ff_2(t)$.

  Using the monotonicity of $\ev_F$ we can conclude that
  \[\liftF f_1(t) = \ev_F(Ff_1(t)) \le \ev_F(Ff_2(t)) = \liftF
    f_2(t)\,.\]

  \medskip

  \noindent ``$\Rightarrow$'' Assume that $\liftF$ is monotone. In
  order to show monotonicity of $\ev_F$ take $v_1,v_2\in F\VV$ such
  that $v_1\orderFV v_2$.  This means that there exists $r\in F[\le]$
  such that $F(\pi_1\circ o)(r) = v_1$, $F(\pi_2\circ o)(r) = v_2$.

  Now consider $\pi_1\circ o$, $\pi_2\circ o\colon [\le]\to \VV$. It
  holds that $\pi_1\circ o\le \pi_2\circ o$ and with monotonicity of
  $\liftF$ we can conclude
  $\liftF(\pi_1\circ o) \le \liftF(\pi_2\circ o)$.

  Hence
  $\ev_F(v_1) = \ev_F(F(\pi_1\circ o)(r)) = \liftF(\pi_1\circ o)(r)
  \le \liftF(\pi_2\circ o)(r) = \ev_F(F(\pi_2\circ o)(r)) =
  \ev_F(v_1)$, i.e., we have shown that $\ev_F$ is monotone.
\end{proof}

\begin{lem}
  \label{lem:evfct-monad}
  Assume that $T$ is a monad and $\VV = T\objTruth$ a quantale as
  detailed above. Assume that there is a partial order $\sqsubseteq$
  on $\objTruth$ such that the lattice order $\le$ of the quantale is
  obtained by lifting $\sqsubseteq$ under $T$, i.e.,
  $\le\ =\ \sqsubseteq^T$ (as in Definition~\ref{def:orderFV}).  Then
  $\ev=\mu_{\objTruth}\colon (T\VV,\le^T) \to (\VV,\le)$ is monotone,
  and consequently corresponds to a fibred lifting $\liftT$ of $T$.
\end{lem}
  
\begin{proof}
  Let $t'_1,t'_2\in T\VV$ such that $t'_1\orderFV t'_2$, i.e.,\
  $t'_1\le^T t'_2$. We have to show that
  $\mu_{\objTruth}(t'_1) \le \mu_{\objTruth}(t'_2)$.

  Since $\le$ is obtained by lifting $\sqsubseteq$ under $T$ we can
  infer that there exists a witness function
  $w\colon \le\ \to T(\sqsubseteq)$ that assigns to every pair of
  elements $t_1,t_2\in \VV$ with $t_1\le t_2$ a witness
  $w(t_1,t_2)\in T(\sqsubseteq)$ with $T\pi_i(w(t_1,t_2)) =
  t_i$. Hence $T\pi_i\circ w = \pi'_i$, where
  $\pi_i\colon \sqsubseteq\ \to {\objTruth}$ and
  $\pi'_i\colon \le\ \to \VV$ are the usual projections.

  Since $t'_1\le^T t'_2$, there exists a witness $t'\in T(\le)$ with
  $T\pi'_i(t') = t'_i$.

  We show that $t = \mu_\sqsubseteq(Tw(t'))$ is a witness for
  $\mu_{\objTruth}(t'_1) \le \mu_{\objTruth}(t'_2)$. It holds that
  $T\pi_i\circ \mu_\sqsubseteq \circ Tw = \mu_{\objTruth}\circ TT\pi_i
  \circ Tw = \mu_{\objTruth}\circ T(T\pi_i\circ w) =
  \mu_{\objTruth}\circ T\pi'_i$, where the first equality holds since
  $\mu$ is a natural transformation.  This implies
  $T\pi_i(t) = (T\pi_i\circ \mu_\sqsubseteq \circ Tw)(t') =
  (\mu_{\objTruth}\circ T\pi'_i)(t') = \mu_{\objTruth}(t'_i)$.
\end{proof}

\subsubsection{From predicates to relations via
  Wasserstein}\label{app:predrel}

\propWassersteinLifting*

\begin{proof}
  To prove that $\barF$ is a well defined functor on $\Vrel$ it
  remains to show that $\barF p\le \reindex{Ff}(\barF q)$ whenever
  $p\le \reindexf q$ (for $f\colon X\to Y$).  From the definition of
  $\barF$ as given in~\eqref{eq:def-V-rel-lift}, we know that on each
  fibre $\barF$ is monotone, hence $\barF p\le\barF
  (\reindexf(q))$. Hence it suffices to show that
  $\barF (\reindexf(q))\le \reindex{F f}\circ\barF(q)$.

  This follows from the sequence of
  (in)equalities~\eqref{eq:prop-wass:1}-\eqref{eq:prop-wass:6}, where
  on each line we underlined the sub-expression that was rewritten and
  which we will explain in turn.
  \begin{align}
    \barF\circ \reindexf(q) & = \isoFibre{FX}^{-1} \circ
                              \directImage{\lambdaChBase_{X}}\circ\liftFrest{\chbaseFunct X}\circ \underline{ \isoFibre{X}\circ\reindexf} (q)\label{eq:prop-wass:1}\\
                            & =\isoFibre{FX}^{-1} \circ
                              \directImage{\lambdaChBase_{X}}\circ\underline{\liftFrest{\chbaseFunct X}\circ\reindex{\chbaseFunct f}}\circ  \isoFibre{Y}(q)\label{eq:prop-wass:2}\\
                            &\le \isoFibre{FX}^{-1} \circ
                              \underline{\directImage{\lambdaChBase_{X}}\circ\reindex{F\chbaseFunct f}}\circ \liftFrest{\chbaseFunct Y}\circ  \isoFibre{Y}(q)\label{eq:prop-wass:3}\\
                            &\le \underline{\isoFibre{FX}^{-1} \circ\reindex{F\chbaseFunct f}}\circ
                              \directImage{\lambdaChBase_{Y}}\circ \liftFrest{\chbaseFunct Y}\circ  \isoFibre{Y}(q)\label{eq:prop-wass:4}\\
                            &=  \reindex{F f}\circ\underline{\isoFibre{FY}^{-1}\circ
                              \directImage{\lambdaChBase_{Y}}\circ \liftFrest{\chbaseFunct Y}\circ  \isoFibre{Y}}(q)\label{eq:prop-wass:5}\\
                            &=  \reindex{F f}\circ\barF(q)\label{eq:prop-wass:6}
  \end{align}

  We obtained~\eqref{eq:prop-wass:1} and~\eqref{eq:prop-wass:6} using
  the definition of $\barF$. To derive the equalities
  in~\eqref{eq:prop-wass:2} and ~\eqref{eq:prop-wass:5} we used the
  fact that $\isoPR$ is a fibred lifting of $\chbaseFunct$. The
  inequality~\eqref{eq:prop-wass:3} follows from the fact that
  $\liftF$ is a \kl{lifting} of $F$ and hence we have the inequality
  \begin{equation}
    \label{eq:prop-wass:7}
    \liftFrest{\chbaseFunct X}\circ\reindex{\chbaseFunct f}\le
    \reindex{F\chbaseFunct f}\circ \liftFrest{\chbaseFunct Y}\,.    
  \end{equation}
  Finally the inequality~\eqref{eq:prop-wass:4} follows from the
  commutativity of the naturality squares of $\lambdaChBase$ and the
  \kl{Beck-Chevalley} condition (see
  Appendix~\ref{sec:proofs-preliminaries}):
  \begin{equation}
    \label{eq:prop-wass:8}
    \directImage{\lambdaChBase_{X}}\circ\reindex{F\chbaseFunct f}\le \reindex{F\chbaseFunct f}\circ
    \directImage{\lambdaChBase_{Y}}\,.
  \end{equation}

  Now let us focus on the second part of the proof. Since $\liftF$ is
  a \kl{fibred lifting} by assumption, then the
  inequality~\eqref{eq:prop-wass:7} becomes an equality. When the
  functor $F$ preserves weak pullbacks, then by
  Lemma~\ref{lem:wpb-naturality} we know that the naturality squares
  of $\lambdaChBase$ are weak pullbacks. Hence, since the fibration
  $\Vrel$ has the \kl{Beck-Chevalley} property for weak pullback
  squares, it follows that~\eqref{eq:prop-wass:8} is also an
  equality. We obtain that all the
  inequalities~\eqref{eq:prop-wass:1}-\eqref{eq:prop-wass:6} are in
  fact equalities. This amounts to the fact that $\barF$ is a
  \kl{fibred lifting}.
\end{proof}

We now prove Lemma~\ref{lem:Wass-is-Functorial}:

\lemWassIsFunctorial*

\begin{proof} The existence (and in this case uniqueness) of the
  lifting $\liftDistr$ is equivalent to the fact that
  $\liftFrest{X}\le\reindex{\distrLaw_X}\circ\liftGrest{X}$ for all
  $X$. This is fairly standard, but we include here an explanation for
  the sake of completeness. If $\liftDistr$ exists, then for all
  $p\in\fibrePred{X}$ we have the next diagram, where the dashed arrow
  exists and is unique by the universal property in
  Definition~\ref{def:fibration}.
  \[
    \begin{tikzcd}[column sep=4em,row sep=2em]
      \liftF(p)\ar[r,"\liftDistr_{p}"]\ar[d,dashed,"\exists !"] &\liftG(p)\\
      \reindex{\distrLaw_X}(\liftG(p))\ar[ur,swap,"\cartLift{\distrLaw_X}{\liftG(p)}"] & \\
      & \\
      X\ar[r,"\distrLaw_{X}"]& Y
    \end{tikzcd}
  \]
  Since the fibres in $\Vpred$ are posets, this means that
  $\liftF(p)\le\reindex{\distrLaw_X}\circ\liftG(p)$, since there is a
  unique morphism in the fibre from $\liftF(p)$ to $\liftG(p)$.  For
  the same reason any two liftings of $\distrLaw$ must
  coincide. Conversely, if the inequality
  $\liftF(p)\le\reindex{\distrLaw_X}\circ\liftG(p)$ holds, we compose
  with $\cartLift{\distrLaw_X}{\liftG(p)}$ in order to obtain
  $\liftDistr_{p}$.

  We have to show that
  $\barFrest{X}\le\reindex{\distrLaw_X}\circ\barGrest{X}$.  We obtain:
  \begin{eqnarray*}
    \reindex{\distrLaw_X}\circ \barG & = & \reindex{\distrLaw_X}\circ \isoFibre{GX}^{-1} \circ
                                           \directImage{\lambdaChBaseG{X}} \circ {\liftG} \circ \isoFibre{X} \\
                                     & = &  \isoFibre{FX}^{-1} \circ \reindex{\distrLaw_X\times \distrLaw_X}\circ 
                                           \directImage{\lambdaChBaseG{X}} \circ {\liftG} \circ \isoFibre{X} \\
                                     & \ge &  \isoFibre{FX}^{-1} \circ 
                                             \directImage{\lambdaChBaseF{X}} \circ \reindex{\distrLaw_{X\times X}}\circ {\liftG} \circ
                                             \isoFibre{X} \\
                                     & \ge & \isoFibre{FX}^{-1}\circ \directImage{\lambdaChBaseF{X}}\circ {\liftF}
                                             \circ \isoFibre{X} \\
                                     & = & \barF
  \end{eqnarray*}
  The inequality on the third line follows from the fact that the
  square below commutes (this follows from naturality of $\distrLaw$
  and uniqueness of mediating morphisms into the product)
  \[
    \xymatrix{ F(X\times X) \ar[r]^{\lambdaChBaseF{X}}
      \ar[d]_{\distrLaw_{X\times X}} &
      FX \times FX \ar[d]^{\distrLaw_X\times \distrLaw_X} \\
      G(X\times X) \ar[r]_{\lambdaChBaseG{X}} & GX\times GX }
  \]
  and that this implies
  $\directImage{\lambdaChBase_X^F}\circ\reindex{\distrLaw_{X\times X}}
  \le \reindex{\distrLaw_X\times \distrLaw_X}\circ
  \directImage{\lambdaChBase_X^G}$ (see
  Appendix~\ref{sec:proofs-preliminaries}).

  To summarize, the proof of the first part of the lemma follows from
  the next lax diagram, by composing with the isomorphisms
  $\isoFibre{X}$ and $\isoFibre{FX}^{-1}$.
  \[
    \begin{tikzcd}[column sep=1em]
      \fibrePred{\Delta X}\ar[rr,"\liftF"]\ar[rrdd,bend right,swap,"\liftG"]&&    \fibrePred{F\Delta X}\ar[rr,"\directImage{\lambdaChBaseF{X}}"]&&    \fibrePred{\Delta FX}\\
      &\rotlen& &\rotlen& \\
      && \fibrePred{G\Delta
        X}\ar[rr,swap,"\directImage{\lambdaChBaseG{X}}"]\ar[uu,"\reindex{\distrLaw_{\Delta
          X}}"]&& \fibrePred{\Delta GX}\ar[uu,swap,"\reindex{\Delta
        \distrLaw_{X}}"]
    \end{tikzcd}
  \]
  It remains to prove that ${\liftF} \le \distrLaw^*_X\circ {\liftG}$
  is equivalent to $\ev_F\le \ev_G\circ \distrLaw_\VV$. The
  implication from left to right is obtained by setting $X = \VV$ and
  applying the functors on both sides to $\id_\VV$. We get the other
  direction by taking $p\colon X\to \VV$ and computing
  $(\reindex{\distrLaw_X}\circ {\liftG})(p) = \ev_G\circ Gp\circ
  \distrLaw_X = \ev_G \circ \distrLaw_V\circ Fp \ge \ev_F\circ Fp =
  {\liftF}(p)$. Note that this uses the naturality of $\distrLaw$.
\end{proof}

We will now focus on proving
Theorem~\ref{thm:restricting-Wasserstein}.

\thmRestrictionOfWasserstein*

The proof is immediate from
Lemmas~\ref{lem:pres-refl},~\ref{lem:pres-trans}
and~\ref{lem:lifting-preserves-symmetry} which we prove next.

\AP Let us denote by $\constPred{X}\colon X\to\VV$ the predicate on
$X$ constant to $1$.  Let $\diago{X}:X\to X\times X$ be the
\intro{diagonal function} on a set $X$. A relation $r:X\times X\to\VV$
is reflexive if and only if
\begin{equation}
  \label{eq:refl-rel-char}
  \reindexnobrk{\diago{X}}\circ \isoFibre{X}(r)\ge \constPred{X}\,.
\end{equation}
\begin{lem}
  \label{lem:pres-refl}
  Assume ${\liftF}$ is a lifting of $F$ such that
  \[
    {\liftF}(\constPred{X})\ge\constPred{FX}\,.
  \]
  Then $\barF(\diagRel{X}) \ge \diagRel{FX}$, hence $\barF$ preserves
  \kl{reflexive} relations.
\end{lem}

\begin{proof}
  Notice that
  \[\diagRel{X}=\isoFibre{X}^{-1}\directImage{\diago{X}}(\constPred{X})\,.\]
  Using this observation, we obtain that
  \begin{align}
    \barF(\diagRel{X}) & = \isoFibre{FX}^{-1}\circ\directImage{\lambdaChBase_{X}}\circ\liftF\circ \directImage{\diago{X}}(\constPred{X})\label{lemma42:1}\\
                       & \ge \isoFibre{FX}^{-1}\circ\directImage{\lambdaChBase_{X}}\circ\directImage{F\diago{X}}\circ \liftF (\constPred{X})\label{lemma42:2}\\
                       & \ge \isoFibre{FX}^{-1}\circ\directImage{\lambdaChBase_{X}}\circ\directImage{F\diago{X}} (\constPred{FX})\label{lemma42:3}\\
                       & = \isoFibre{FX}^{-1}\circ\directImage{\diago{FX}} (\constPred{FX})\label{lemma42:4}\\
                       & = \diagRel{FX}
  \end{align}
  In~\eqref{lemma42:1} we used the definition of $\barF$. For the
  inequality~\eqref{lemma42:2} we used that $\liftF$ is a lifting of
  $F$ and the mate of~\eqref{eq:canonical-ineq}, i.e.,
  $\liftF\circ \directImage{\diago{X}}\ge
  \directImage{F\diago{X}}\circ \liftF$. The
  inequality~\eqref{lemma42:3} is the hypothesis, while
  in~\eqref{lemma42:4} we used that
  $\lambdaChBase_{X}\circ F\diago{X}=\diago{FX}$.

  Preservation of \kl{reflexive} relations is now immediate. For
  $r\in\Vrel_{X}$ is \kl{reflexive} when $r\ge\diagRel{X}$.  Hence
  $\barF(r)\ge\barF(\diagRel{X})\ge\diagRel{FX}$, which entails that
  $\barF$ is \kl{reflexive}.
\end{proof}

We now turn our attention to the preservation of \kl{composition} of
reltions and of the \kl{transitivity} property.

We will use the notations $\pi_i:X\times X\times X\to X$ to denote the
$i^\mathit{th}$ projection on $X^3$ and
$\tau_i:FX\times FX\times FX\to FX$ to denote the $i^\mathit{th}$
projection on $(FX)^3$.

We will use the fact that the \kl{composition} $p\comp q$ of two
relations $p,q\colon X\times X\to \VV$ can be written as the composite
\begin{equation}
  \label{eq:rel-comp-fib}
  p\comp q=\isoFibre{X}^{-1}\directImage{\pair{\pi_1}{\pi_3}}(\reindexnobrk{\pair{\pi_2}{\pi_3}}(\isoFibre{X}q)
  \tensor
  \reindexnobrk{\pair{\pi_1}{\pi_2}}(\isoFibre{X}p))
\end{equation}

\begin{lem}
  \label{lem:pres-trans}
  Assume $F$ preserves weak pullbacks and ${\liftF}$ is a fibred
  lifting of $F$ such that
  \begin{equation}
    \label{eq:cond-F-trans}
    {\liftF}(u\tensor v)\ge{\liftF}(u)\tensor{\liftF}(v)
  \end{equation}
  Then $\barF(p\comp q) \ge \barF(p)\comp \barF(q)$, hence $\barF$
  preserves \kl{transitive} relations.
\end{lem}
\begin{proof}
  We denote by $\nu_X\colon F(X^3)\to (FX)^3$ the map defined as
  $\nu_X = \langle F\pi_1,F\pi_2,F\pi_3\rangle$.
  \begin{align}
    \barF(p\comp q)& = \isoFibre{X}^{-1}\directImage{\lambdaChBase_X}\liftF
                     \directImage{\pair{\pi_1}{\pi_3}}(\reindexnobrk{\pair{\pi_2}{\pi_3}}(\isoFibre{X}q)\tensor\reindexnobrk{\pair{\pi_1}{\pi_2}}(\isoFibre{X}p))\label{lemma43:1}\\
                   & \ge \isoFibre{X}^{-1}\directImage{\lambdaChBase_X}\directImage{F\pair{\pi_1}{\pi_3}}\liftF
                     (\reindexnobrk{\pair{\pi_2}{\pi_3}}(\isoFibre{X}q)\tensor\reindexnobrk{\pair{\pi_1}{\pi_2}}(\isoFibre{X}p))\label{lemma43:2}\\
                   & \ge \isoFibre{X}^{-1}\directImage{\lambdaChBase_X}\directImage{F\pair{\pi_1}{\pi_3}}\liftF(\reindexnobrk{\pair{\pi_2}{\pi_3}}(\isoFibre{X}q))\tensor\liftF(\reindexnobrk{\pair{\pi_1}{\pi_2}}(\isoFibre{X}p))\label{lemma43:3}\\
                   & = \isoFibre{X}^{-1}\directImage{\lambdaChBase_X}\directImage{F\pair{\pi_1}{\pi_3}}(\reindexnobrk{F\pair{\pi_2}{\pi_3}}\liftF(\isoFibre{X}q)\tensor\reindexnobrk{F\pair{\pi_1}{\pi_2}}\liftF(\isoFibre{X}p))\label{lemma43:4}\\
                   & = \isoFibre{X}^{-1}\directImage{\pair{\tau_1}{\tau_3}}\directImage{\nu_X}(\reindexnobrk{F\pair{\pi_2}{\pi_3}}\liftF(\isoFibre{X}q)\tensor\reindexnobrk{F\pair{\pi_1}{\pi_2}}\liftF(\isoFibre{X}p))\label{lemma43:5}\\
                   & =\isoFibre{X}^{-1}\directImage{\pair{\tau_1}{\tau_3}}(\reindexnobrk{\pair{\tau_2}{\tau_3}}\directImage{\lambdaChBase_X}(\liftF(\isoFibre{X}q))\tensor\reindexnobrk{\pair{\tau_1}{\tau_2}}\directImage{\lambdaChBase_X}(\liftF(\isoFibre{X}p)))\label{lemma43:6}\\
                   & =\isoFibre{X}^{-1}\directImage{\pair{\tau_1}{\tau_3}}(\reindexnobrk{\pair{\tau_2}{\tau_3}}(\isoFibre{X}\barF q)\tensor\reindexnobrk{\pair{\tau_1}{\tau_2}}(\isoFibre{X}\barF p))\label{lemma43:7}\\
                   & =\barF p\comp\barF q\label{lemma43:8}
  \end{align}
  The equalities~\eqref{lemma43:1},~\eqref{lemma43:7}
  and~\eqref{lemma43:8} follow by unraveling the definition of $\barF$
  and from~\eqref{eq:rel-comp-fib}.  The inequality~\eqref{lemma43:2}
  follows using by the mate of~\eqref{eq:canonical-ineq}.  The
  inequality~\eqref{lemma43:3} follows from the hypothesis on
  $\liftF$.  The equality~\eqref{lemma43:4} is obtained using the
  $\barF$ is a \kl{fibred lifting}.  To prove the equality
  in~\eqref{lemma43:5} we use that
  $\lambdaChBase_X\circ
  F\pair{\pi_1}{\pi_3}=\pair{\tau_1}{\tau_3}\circ\nu_X$.
  Finally~\eqref{lemma43:6} follows from Lemma~\ref{lem:BCC}.

  Assume $r\in\Vrel_X$ is transitive, that is, $r\comp r\le r$. Then
  we have $\barF r\comp \barF r\le\barF(r\comp r)\le \barF r$, hence
  $\barF r$ is transitive.
\end{proof}

\begin{lem}
  \label{lem:BCC}
  Assume $F$ preserves weak pullbacks and
  $u,v\colon F(X\times X)\to \VV$ is in $\fibrePred{F(X\times
    X)}$. Then we have:
  $\directImage{\nu_X}(\reindex{F\pair{\pi_2}{\pi_3}}(u)\tensor\reindex{F\pair{\pi_1}{\pi_2}}(v))=\reindexnobrk{\pair{\tau_2}{\tau_3}}\directImage{\lambdaChBase_X}(u)\tensor\reindexnobrk{\pair{\tau_1}{\tau_2}}\directImage{\lambdaChBase_X}(v)$.
\end{lem}

\begin{proof}
  It is easy to verify that the square below is a pullback.
  
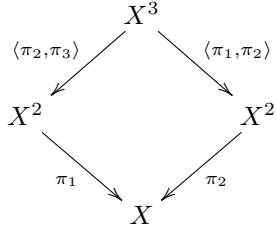
\[
    \xymatrix{
      & X^3 \ar[ld]_{\pair{\pi_2}{\pi_3}} \ar[rd]^{\pair{\pi_1}{\pi_2}} & \\
      X^2 \ar[rd]_{\pi_1} & & X^2 \ar[dl]^{\pi_2} \\
      & X & }
  \]
  By applying $F$ to the diagram we obtain the diagram below where the
  square is a weak pullback (since $F$ preserves weak pullbacks).
  
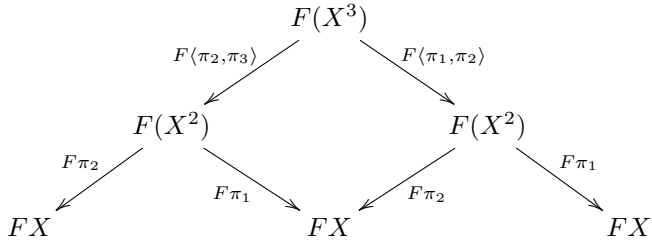
\[
    \xymatrix{ & & F(X^3) \ar[ld]_{F\pair{\pi_2}{\pi_3}}
      \ar[rd]^{F\pair{\pi_1}{\pi_2}} & \\
      & F(X^2) \ar[ld]_{F\pi_2} \ar[rd]_{F\pi_1} & & F(X^2)
      \ar[dl]^{F\pi_2}  \ar[rd]^{F\pi_1} \\
      FX & & FX & & FX }
  \]
  Using this diagram we can show that the square below is a pullback
  as well. Assume that $t_1,t_2\in F(X^2)$, $(s_1,s_2,s_3)\in (FX)^3$
  are given such that $\lambdaChBase_X(t_1) = (s_2,s_3)$ (which means
  $F\pi_1(t_1) = s_2$, $F\pi_2(t_1) = s_3$) and
  $\lambdaChBase_X(t_2) = (s_1,s_2)$ (which means $F\pi_1(t_2) = s_1$,
  $F\pi_2(t_2) = s_2$). That is, $t_1,t_2$ live on the middle level
  and $s_3,s_2,s_1$ on the lower level (in that order) in the diagram
  above. Since the square is a weak pullback, there exists
  $t\in F(X^3)$ such that $F\pair{\pi_2}{\pi_3}(t) = t_1$ and
  $F\pair{\pi_1}{\pi_2}(t) = t_2$. It remains to verify that
  $\nu_X(t) = (s_1,s_2,s_3)$: for instance
  $F\pi_1(t) = (F\pi_1\circ F\pair{\pi_1}{\pi_2})(t) = F\pi_1(t_2) =
  s_1$. (Analogously for $s_2,s_3$.)

  \[
    \xymatrix{
      F(X^3)\ar[r]^-{\nu_X}\ar[d]_-{\pair{F\pair{\pi_2}{\pi_3}}{F\pair{\pi_1}{\pi_2}}} & (FX)^3\ar[d]^-{\pair{\pair{\tau_2}{\tau_3}}{\pair{\tau_1}{\tau_2}}} \\
      F(X^2)\times F(X^2)\ar[r]_{\lambdaChBase_X\times \lambdaChBase_X} & (FX)^2\times (FX)^2 \\
    }
  \]
  Since the \kl{Beck-Chevalley} condition holds we obtain
  \[
    \directImage{\nu_X}\reindexnobrk{\pair{F\pair{\pi_2}{\pi_3}}
      {F\pair{\pi_1}{\pi_2}}} =
    \reindexnobrk{\pair{\pair{\tau_2}{\tau_3}}{\pair{\tau_1}{\tau_2}}}\directImage{\lambdaChBase_X\times
      \lambdaChBase_X}.
  \]
  Then we will apply this to a predicate of the form
  $\otimes \circ (u\times v)$ and using the facts
  \begin{itemize}
  \item
    $\reindexnobrk{\pair{h_1}{h_2}}(\otimes\circ(u\times
    v))=\reindexnobrk{h_1}(u) \tensor \reindexnobrk{h_2}(v)$.
  \item
    $\directImage{f\times f}(\otimes\circ (u\times v)) = \otimes
    \circ(\directImage{f}(u)\times \directImage{f}(v))$.
  \end{itemize}
  we derive the desired equality.

  While the first item above is straightforward, the second has to be
  further explained. Whenever $f\colon X\to Y$, $p,p'\colon X\to \VV$,
  $y,y'\in Y$, we have (using distributivity):
  \begin{eqnarray*}
    && \directImage{f\times f}(\otimes\circ (p\times p'))(y,y') \\
    & = & 
          \bigvee\{p(x)\otimes p'(x')\mid f(x)=y, f(x')=y'\} \\
    & = & \left(\bigvee_{f(x)=y} p(x)\right) \otimes
          \left(\bigvee_{f(x')=y'} p'(x')\right) \\
    & = & \directImage{f}(p)(y)\otimes \directImage{f}(p')(y') \\
    & = & \otimes\circ (\directImage{f}(p)\times \directImage{f}(p'))(y,y')\qquad\qquad\qquad\qquad\qedhere
  \end{eqnarray*}
\end{proof}

\begin{lem}
  \label{lem:lifting-preserves-symmetry}
  The lifting preserves symmetric $\VV$-valued relations.
\end{lem}

\begin{proof}
  We first observe that the square below commutes.
  \[
    \begin{tikzcd}
      F(X\times X) \ar[r,"\lambdaChBase_X"] \ar[d,swap,"F\sym{X}"]
      & FX\times FX \ar[d,"\sym{FX}"] \\
      F(X\times X) \ar[r,"\lambdaChBase_X"] & FX\times FX
    \end{tikzcd}
  \]
  Knowing that $\lambdaChBase_X = \pair{F\pi_1^X}{F\pi_2^X}$ and that
  $\sym{X} = \pair{\pi_2}{\pi_1^X}$, where
  $\pi_i^X\colon X\times X\to X$, we can easily show that the square
  commutes:
  \begin{eqnarray*}
    && \sym{FX} \circ \lambdaChBase_X \\
    & = & \pair{\pi_2^{FX}}{\pi_1^{FX}} \circ \pair{F\pi_1^X}{F\pi_2^X} \\
    & = & \pair{\pi_2^{FX}\circ \pair{F\pi_1^X}{F\pi_2^X}}{\pi_1^{FX}\circ \pair{F\pi_1^X}{F\pi_2^X}}  \\
    & = & \pair{F\pi_2^X}{F\pi_1^X} \\
    & = & \pair{F(\pi_1^X\circ \pair{
          \pi_2^X}{\pi_1^X})}{F(\pi_2^X\circ \pair{
          \pi_2^X}{\pi_1^X})} \\
    & = & \pair{F(\pi_1^X\circ \sym{X})}{F(\pi_2^X\circ 
          \sym{X})} \\
    & = & \pair{F\pi_1^X}{F\pi_2^X}\circ F\sym{X} \\
    & = & \lambdaChBase_X\circ F\sym{X}
  \end{eqnarray*}
  Recall that $p\in\fibreRel{Y}$ is symmetric when $p=p\circ
  \sym{Y}$. We cannot perform a reindexing along $\sym{Y}$ in the
  fibration $\Vrel$, since $\sym{Y}$ is not a morphims on $Y$, but on
  $Y\times Y$. Instead, we have that $p$ is \kl{symmetric} if and only
  if \[\isoFibre{Y}p=\reindex{\sym{Y}}(\isoFibre{Y}p)\] in $\Vpred$.
  Hence, we want to show that for any $r\in\fibreRel{X}$ the
  implication holds
  \[\isoFibre{X}r=\reindex{\sym{X}}(\isoFibre{X}r)\Rightarrow
    \isoFibre{FX}\barF r=\reindex{\sym{FX}}(\isoFibre{FX}\barF r)\] We
  have the following inequalities:
  \begin{align*}
    \isoFibre{FX}\barF r & = \directImage{\lambdaChBase_{X}}\circ\liftFrest{\chbaseFunct X}\circ \isoFibre{X} (r)\\
                         & = \directImage{\lambdaChBase_{X}}\circ\liftFrest{\chbaseFunct X} \circ \reindex{\sym{X}} (\isoFibre{X} r)\\
                         & \le \directImage{\lambdaChBase_{X}} \circ \reindex{F \sym{X}} \circ\liftFrest{\chbaseFunct X} (\isoFibre{X} r)\\
                         & \le  \reindex{\sym{FX}} \circ \directImage{\lambdaChBase_{X}} \circ \liftFrest{\chbaseFunct X} (\isoFibre{X} r)\\
                         & =   \reindex{\sym{FX}} (\isoFibre{FX}\barF r)
  \end{align*}
  However, using the idempotency of $\sym{FX}$ and the monotonicity of
  $\reindex{\sym{FX}}$ from the inequality
  \[\isoFibre{FX}\barF r\le \reindex{\sym{FX}} (\isoFibre{FX}\barF
    r)\] that we have just proved above we can infer that the equality
  also holds.
\end{proof}

\begin{lem}[Corollary~2.7 in~\cite{h:closed-objects}]
  \label{lem:wpb-naturality}
  If $F\colon\Set\to\Set$ is weak pullback-preserving, then the
  naturality squares of the binatural transformation
  $\pair{F\pi_1d}{,F\pi_2}\colon F(X\times Y)\to FX\times FY$ are weak
  pullbacks, where $\pi_1\colon X\times Y\to X$ and
  $\pi_2\colon X\times Y\to X$ denote the projections. In particular,
  the naturality squares of the natural transformation $\lambdaChBase$
  are weak pullbacks.
\end{lem}

\begin{proof}
  Consider morphisms $f\colon X\to X'$ and $g\colon Y\to Y'$. We want
  to prove that the square
  \begin{equation}
    \label{eq:wpb1}
    \begin{tikzcd}
      F(X\times Y) \ar[dd,"F(f\times g)" description]
      \ar[rr,"\pair{F\pi_1}{F\pi_2}"] & & F(X)\times F(Y)
      \ar[dd,"Ff\times Fg" description]
      \\
      & & \\
      F(X'\times Y')\ar[rr,"\pair{F{\pi_1}}{F{\pi_2}}"] & & F(X')\times F(Y')\\
    \end{tikzcd}
  \end{equation}
  is a weak pullback. To this end we will consider the following
  diagram:
  \begin{equation}
    \label{eq:wpb2}
    \begin{tikzcd}[column sep=0.9em]
      & & F(X\times Y)\ar[ld,swap,"F(X\times g)" ]\ar[rd,"F(f\times Y)"] & & \\
      & F(X\times Y')\ar[ld,swap,"F\pi_1"]\ar[rd,"F(f\times Y')" description]& & F(X'\times Y)\ar[ld,"F(X'\times g)" description]\ar[rd,"F\pi_2"] &\\
      F(X)\ar[rd,swap,"Ff"] & & F(X'\times Y')\ar[ld,"F\pi_1"]\ar[rd,swap,"F\pi_2"] & & F(Y)\ar[ld,"Fg"]\\
      & F(X') & & F(Y') & \\
    \end{tikzcd}
  \end{equation}
      
  The three squares above are obtained by applying the functor $F$ to
  weak pullbacks, hence, by the assumption on $F$, they are also weak
  pullbacks.

  Assume $s'\in F(X')$ and $t'\in F(Y')$ are such that there exist
  $s\in F(X)$, $t\in F(Y)$ and $u\in F(X'\times Y')$ satisfying
  $Ff(s)=s'$, $Fg(t)=t'$ and
  $\langle F\pi_1,F\pi_2\rangle(u)=(s',t')$. Proving
  that~\eqref{eq:wpb1} is a weak pullback amounts to showing the
  existence of $v\in F(X\times Y)$ so that $F\pi_1(v)=s$,
  $F\pi_2(v)=t$ and $F(f\times g)(v)=u$.

  From the fact that the lower left square in~\eqref{eq:wpb2} is a
  weak pullback we infer the existence of $u_1\in F(X\times Y')$ such
  that $F\pi_1(u_1)=s$ and $F(f\times Y')(u_1)=u$.

  Analoguosly, using that the lower right square in~\eqref{eq:wpb2} is
  a weak pullback we obtain the existence of $u_2\in F(X'\times Y)$
  such that $F\pi_2(u_2)=t$ and $F(X'\times g)(u_2)=u$.

  Since the upper square is also a weak pullback, we deduce the
  existence of $v\in F(X\times Y)$ satisfying $F(X\times g)(v)=u_1$
  and $F(f\times Y)(v)=u_2$. Upon noticing that
  $F\pi_1\circ F(X\times g)=F(\pi_1)$,
  $F\pi_2\circ F(f\times Y)=F(\pi_2)$, and
  $F(f\times Y')\circ F(X\times g)=F(f\times g)$, we conclude that $v$
  is the element we were looking for in $F(X\times Y)$.
\end{proof}

\begin{lem}
  \label{lem:can-ev-map-is-monotone}
  Assume the functor $F$ preserves weak pullbacks. The map
  $\evcan\colon F\VV\to\VV$ is a monotone evaluation map, that is, it
  is monotone with respect to the order $\orderFV$ on $F\VV$ defined
  in~\ref{def:orderFV} and the order $\le$ on $\VV$.
\end{lem}

\begin{proof}
  It is suffient to show that
  $\evcan\colon (F\VV,\orderFV) \to (\VV,\le)$ is monotone.

  Hence, let $u'_1,u'_2\in F\VV$ such that $u'_1\orderFV u'_2$, which
  implies that there exists $u'\in F(\VV\times \VV)$ with
  $F\pi_i(u') = u'_i$, where $\pi_i\colon \le\ \to \VV$ are the
  projections.

  We have to show that $\evcan(u'_1) \le \evcan(u'_2)$. It is
  sufficient to show that $u'_1\in F(\uparr)$ implies
  $u'_2\in F(\uparr)$.  Assume $r\in\VV$ is such that
  $u'_1\in F(\uparr)$. Then there exists $u_1\in\, \uparr$ such that
  $F\truer(u_1) = u'_1$.

  Now consider the diagram on the left
  in~\eqref{eq:monotonicity-ev-map}, where
  \[ e_r\colon \restrictionle\ \to\ \le\] embeds $\le$ restricted to
  $\uparr$ into the full relation. Furthermore the functions
  $\overline{\pi}_i$ are the projections for $\restrictionle$. This
  diagram commutes for $i=1,2$ and is a weak pullback for $i=1$. Hence
  the diagram on the right in~\eqref{eq:monotonicity-ev-map} is also a
  weak pullback.
  \begin{equation}
    \label{eq:monotonicity-ev-map}
    \xymatrix{ \restrictionle \ar[r]^{e_r} \ar[d]_{\overline{\pi}_i} &
      \le \ar[d]^{\pi_i} \\
      \uparr \ar[r]^{\truer} & \VV } \qquad \xymatrix{
      F(\restrictionle) \ar[r]^{Fe_r} \ar[d]_{F\overline{\pi}_1} &
      F(\le) \ar[d]^{F\pi_1} \\
      F(\uparr) \ar[r]^{F\truer} & F\VV }
  \end{equation}
  We have $u_1\in F(\uparr)$ and $u'\in F(\le)$ with
  $F\truer(u_1) = u'_1 = F\pi_1(u')$. Hence there must be an element
  $u\in F(\restrictionle)$ with $F\overline{\pi}_1(u) = u_1$ and
  $Fe_r(u) = u'$.

  We set $u_2 = F\overline{\pi}_2(u)\in F(\uparr)$ and observe that
  $F\truer(u_2) = F(\truer\circ \overline{\pi}_2)(u) = F(\pi_2\circ
  e_r)(u) = F\pi_2(u') = u'_2$. This means that $u'_2\in F(\uparr)$ as
  required.
\end{proof}

\begin{prop}
  \label{propCharEvMapWellBeh}
  Assume $\ev\colon F\VV\to\VV$ is monotone evaluation map and let
  $\liftF$ be the corresponding \kl{fibred lifting} of $F$. Then we
  have:
  \begin{enumerate}
  \item $\liftF(p\tensor q)\ge\liftF(p)\tensor\liftF(q)$ holds
    whenever the map $\tensor\colon\VV\times\VV\to\VV$ is a lax
    $F$-algebra morphism, in the sense that we have a lax diagram:
    \[
      \begin{tikzcd}[row sep=0.5em]
        F(\VV\times\VV)\ar[dd,swap,"F(\tensor)"]\ar[r,"\lambdaChBase_\VV"] &F\VV\times F\VV\ar[r,"\ev\times\ev"] & \VV\times\VV\ar[dd,"\tensor"] \\
        & \ge & \\
        F\VV\ar[rr,swap,"\ev"] & & \VV
      \end{tikzcd}
    \]
  \item $\liftF(\constPred{X})\ge\constPred{X}$ holds whenever the map
    $\constPred{\onebb}\colon\onebb\to\VV$ is a lax algebra morphism,
    i.e., we have the lax diagram
    \[
      \begin{tikzcd}[row sep=0.5em, column sep=0.5em]
        F\onebb\ar[rr, "!"]\ar[dd,swap, "F\constPred{\onebb}"] & & \onebb\ar[dd,"\constPred{\onebb}"] \\
        & \ge & \\
        F\VV\ar[rr,swap,"\ev"] & & \VV
      \end{tikzcd}
    \]
  \end{enumerate}
\end{prop}

\begin{proof}
  1. We start with the observation that the predicate $p\tensor q$ is
  computed as the composite
  \[
    \begin{tikzcd}
      X\ar[r,"\diago{X}"] & X\times X\ar[r,"p\times q"] & \VV\times
      \VV\ar[r,"\tensor"] & \VV
    \end{tikzcd}
  \]
  Upon recalling that $\liftF(p)=\ev\circ F(p)$, we notice that the
  leftmost path from $FX$ to $\VV$ in the next diagram evaluates to
  $\liftF(p\tensor q)$. Similarly, the rightmost path from $FX$ to
  $\VV$ evaluates to $\liftF(p)\tensor\liftF(q)$. Now, the desired
  inequality $\liftF(p\tensor q)\ge\liftF(p)\tensor\liftF(q)$ follows
  using the fact that the upper triangle commutes, the naturality of
  $\lambdaChBase$ and the lax diagram from the hypothesis.
  
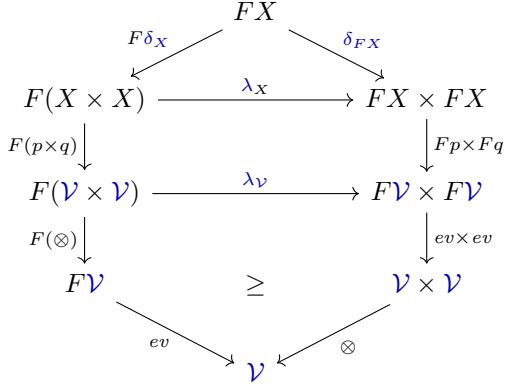
\[
    \begin{tikzcd}
      & FX\ar[dl,swap,"F\diago{X}"]\ar[dr,"\diago{FX}"] & \\
      F(X\times X)\ar[d,swap,"F(p\times q)"]\ar[rr,"\lambdaChBase_{X}"] & & FX\times FX\ar[d,"Fp\times Fq"]  \\
      F(\VV\times \VV)\ar[rr,"\lambdaChBase_{\VV}"]\ar[d,swap,"F(\tensor)"] & & F\VV\times F\VV\ar[d,"\ev\times\ev"]\\
      F\VV\ar[rd,swap,"\ev"] & \ge & \VV\times\VV\ar[dl,"\tensor"]\\
      & \VV &
    \end{tikzcd}
  \]

  \medskip

  2. We consider the following diagram
  
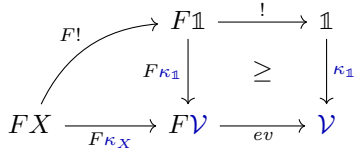
\[
    \begin{tikzcd}[row sep=0.5em, column sep=0.8em]
      &&  F\onebb\ar[rr, "!"]\ar[dd,swap, "F\constPred{\onebb}"] & & \onebb\ar[dd,"\constPred{\onebb}"] \\
      && & \ge & \\
      FX\ar[rr,swap,"F\constPred{X}"]\ar[rruu,bend left,"F!"] &\quad &
      F\VV\ar[rr,swap,"\ev"] & & \VV
    \end{tikzcd}
  \]
  Since $\constPred{X}=\constPred{\onebb}\circ !$ the left traingle
  commutes. Hence we obtain
  \[
    \ev\circ F\constPred{X}\ge \constPred{\onebb}\circ !\circ F!\,,
  \]
  or, equivalently,
  \[
    \liftF(\constPred{X})\ge\constPred{X}\,.
  \]
\end{proof}

We show how one of the conditions for well-behavedness that we
required for the Wasserstein lifting in \cite[Definition
5.14]{bbkk:coalgebraic-behavioral-metrics} for the quantale
$\VV = [0,\infty]$ is related to the conditions in
Proposition~\ref{thm:restricting-Wasserstein}. Our original condition
was
$d_e\circ (\ev_F\times \ev_F)\circ \pair{F\pi_1}{F\pi_2} \le
\liftF(d_e)$ where
$d_e({-}_1,{-}_2) = \righthom{{-}_1}{{-}_2} \land
\righthom{{-}_2}{{-}_1}$ (which evaluates to $d_e(r,s) = |r-s|$ on the
reals). This clearly implies the non-symmetric variant stated in the
lemma below.

\todo{BK: the following lemma was only stated vor $\VV =
  [0,\infty]$. I checked and the proof generalizes to any quantale. I
  updated the lemma accordingly.}

\begin{lem}
  \label{lem:wbhvd-additivity-compare}
  $\liftF(p\tensor q)\ge\liftF(p)\tensor\liftF(q)$ holds for all
  $p,q\colon X\to \VV$ if and only if
  $\righthom{\pi_1}{\pi_2} (\ev_F \times \ev_F) \pair{F\pi_1}{F\pi_2}
  \geq \liftF\righthom{\pi_1}{\pi_2}$.
\end{lem}

\begin{proof}
  \begin{align}
    \quad& \righthom{\pi_1}{\pi_2} (\ev_F \times \ev_F) \pair{F\pi_1}{F\pi_2} \geq \liftF\righthom{\pi_1}{\pi_2}\label{eq:remark:conf-F-trans:2}\\
    \iff\quad &\righthom{\pi_1}{\pi_2}\pair{\liftF\pi_1}{\liftF\pi_2} \geq  \liftF\righthom{\pi_1}{\pi_2}\label{eq:remark:conf-F-trans:3}\\
    \iff\quad &\righthom{\liftF\pi_1}{\liftF\pi_2} \geq  \liftF\righthom{\pi_1}{\pi_2}\label{eq:remark:conf-F-trans:4}\\
    \iff\quad &\liftF\pi_2 \geq \liftF\righthom{\pi_1}{\pi_2} \tensor
                \liftF\pi_1\label{eq:remark:conf-F-trans:5}
  \end{align}
  using
  \begin{itemize}
  \item For the equivalences \eqref{eq:remark:conf-F-trans:2} $\iff$
    \eqref{eq:remark:conf-F-trans:3} $\iff$
    \eqref{eq:remark:conf-F-trans:4} just simple rewriting along with
    $\liftF = \ev_F \circ F$.
  \item For the equivalence \eqref{eq:remark:conf-F-trans:4} $\iff$
    \eqref{eq:remark:conf-F-trans:5} the tensor property
    $x\tensor y\le z \iff x\le\righthom{y}{z}$.
  \end{itemize}
  Using this, we now aim to show that \eqref{eq:remark:conf-F-trans:5}
  $\iff$ \eqref{eq:cond-F-trans}.
  \begin{itemize}
  \item Showing \eqref{eq:cond-F-trans} $\implies$
    \eqref{eq:remark:conf-F-trans:5} is straightforward: from
    $\righthom{\pi_1}{\pi_2} \le \righthom{\pi_1}{\pi_2}$ we can infer
    $\righthom{\pi_1}{\pi_2} \tensor \pi_1 \leq \pi_2$. Using this,
    the monotonicity of $\liftF$ and \eqref{eq:cond-F-trans} (by
    taking $u = \righthom{\pi_1}{\pi_2} \colon X \to \VV$ and
    $v = \pi_1 \colon X \to \VV$) we obtain inequality
    \eqref{eq:remark:conf-F-trans:5} as follows:
    \begin{align*}
      \liftF\pi_2 \geq \liftF\left(\righthom{\pi_1}{\pi_2} \tensor \pi_1\right) \geq \liftF\righthom{\pi_1}{\pi_2} \tensor \liftF\pi_1
    \end{align*}
	
  \item The implication \eqref{eq:remark:conf-F-trans:5} $\implies$
    \eqref{eq:cond-F-trans} can be shown by rewriting
    $u \tensor v = \pi_2 \circ \langle v, u \tensor v\rangle$ and then
    using \eqref{eq:remark:conf-F-trans:5} as follows
    \begin{align*}
      &\liftF(u \tensor v) \\
      =\quad&\liftF\left(\pi_2 \circ \pair{v}{u \tensor v}\right) = \liftF\pi_2 \circ F\pair{v}{u \tensor v} \\
      \geq\quad&\left(\liftF\righthom{\pi_1}{\pi_2} \tensor \liftF\pi_1\right) \circ F\pair{v}{u \tensor v}\\
      =\quad&\liftF\left(\righthom{\pi_1}{\pi_2}  \circ\pair{v}{u \tensor v} \right) \tensor \liftF\left(\pi_1 \circ \pair{v}{u \tensor v}\right)\\
      =\quad&\liftF\righthom{v}{u\tensor v} \tensor \liftF v \geq \liftF u \tensor \liftF v
    \end{align*}
    where the last inequality follows again from monotonicity of
    $\liftF$ and the definitions of $\tensor$ and $[-,-]$. In
    particular $u\tensor v\le u\tensor v$ and hence
    $u \le \righthom{v}{u\tensor v}$.\qedhere
  \end{itemize}
\end{proof}

Next we will prove Proposition~\ref{prop:can-lifting-is-well-behaved},
which captures the properties of the canonical lifting $\liftFcan$.

\propCanLiftingIsWellBehaved*

\begin{proof}
  1.  Given $t\in FX$, we have on one hand that
  \[
    \begin{aligned}
      \liftFcan(p\tensor q)(t) &=       \evcan(F(p\tensor q)(t)) \\
      &= \bigvee\{ r\mid F(p\tensor q)(t) \in F(\uparr)\}\,,
    \end{aligned}
  \]
  and on the other, that
  \begin{eqnarray*}
    &   & (\liftFcan (p)\tensor \liftFcan (q))(t) \\
    & = & \evcan(Fp(t)) \tensor\evcan(Fq(t)) \\
    & = & \bigvee\{ r\mid Fp(t) \in
          F(\uparr)\} \tensor
          \bigvee\{ s\mid Fq(t) \in F(\upars)\} \\
    & = & \bigvee\{ r\tensor s\mid Fp(t) \in F(\uparr), Fq(t)
          \in F(\upars) \}\,.
  \end{eqnarray*}
  Hence, in order to show the desired inequality it is sufficient to
  show that
  \[
    Fp(t)\in F(\uparr),\ Fq(t)\in F(\upars)\textrm{ imply } F(p\otimes
    q)(t) \in F(\uparrs)\,.
  \]
  Let $r,s\in\VV$ so that $Fp(t)\in F(\uparr)$ and
  $Fq(t)\in F(\upars)$.  Note that $p\tensor q\colon X\to \VV$ is the
  composite:
  \[
    \xymatrix{ X \ar[r]^(.4){\diago{X}} \ar@/_1pc/[rrr]_{p\tensor q} &
      X\times X \ar[r]^{p\times q} & \VV\times \VV
      \ar[r]^(.6){\tensor} & \VV }
  \]
  Hence $F(p\otimes q)$ is the composite of the arrows on the top line
  of the diagram below:
  \[
    \xymatrix{ FX \ar[r]^(.4){F\diago{X}} \ar@/_1pc/[dr]_{\diago{FX}}
      & F(X\times X) \ar[d]_{\lambdaChBase_X} \ar[r]^{F(p\times q)} &
      F(\VV\times \VV) \ar[r]^(.6){F(\tensor)}
      \ar[d]^{\lambdaChBase_{\VV}} & F\VV \\
      & FX\times FX \ar[r]^{Fp\times Fq} & F\VV\times F\VV \\
    }
  \]
  Note that the triangle and the square above are commutative. Using
  the abbreviation $\theta = F((p\times q)\circ \diago{X})(t)$ we have
  that:
  \begin{eqnarray}
    F(p\tensor q)(t) & = & F(\tensor)(\theta)
    \\
    ((Fp)(t),(Fq)(t)) & = & \lambdaChBase_\VV(\theta)\label{eq:trans-ev-can-aux}
  \end{eqnarray}
  From Lemma~\ref{lem:wpb-naturality} we know that the square in the
  diagram below is a weak pullback.
  \begin{equation}
    \label{eq:wpb-r-s-VV}
    \begin{tikzcd}[column sep=6em]
      F((\uparr)\times (\upars)) \ar[r,"{F(\truer\times \trues)}"]
      \ar[d,"{\langle F\pi_1,F\pi_2\rangle}"] &
      F(\VV\times \VV) \ar[d,"{\lambdaChBase_\VV}"] \\
      F(\uparr) \times F(\upars) \ar[r,"{F\truer\times F\trues}"] &
      F\VV\times F\VV
    \end{tikzcd}
  \end{equation}
  By hypothesis we know that there exists $u\in F(\uparr)$ and
  $v\in F(\upars)$ such that $Fp(t) = F\truer(u)$ and
  $Fq(t) = F\trues(v)$. Hence
  \[
    (F\truer\times F\trues)(u,v)=\lambdaChBase_\VV(\theta)\,.
  \]
  Using the fact that the square~\eqref{eq:wpb-r-s-VV} is a weak
  pullback, there exists $w\in F((\uparr)\times (\upars))$ such that
  $F(\truer\times \trues)(w) = \theta$, $F\pi_1(w) = u$ and
  $F\pi_2(w) = v$.

  Thus far we have shown that
  \[
    F(p\otimes q)(t)=F(\otimes)(\theta) =F(\otimes)\circ
    F(\truer\times \trues)(w)
  \]
  for some $w\in F((\uparr)\times (\upars))$. To finish the proof of
  the first item, we will prove that
  $F(\otimes)\circ F(\truer\times \trues)$ factors through
  $F\truers\colon F(\uparrs)\to F\VV$.
    
  To this end, notice that due to monotonicity of the tensor product,
  we know that $(\uparr)\otimes (\upars) \subseteq \,\uparrs$. Hence,
  $\tensor\colon\VV\times\VV\to\VV$ restricts to a function
  ${\otimes_{|\uparr,\upars}}$ on $\uparr\times\upars$ so that the
  square below commutes.
  \[
    \begin{tikzcd}[column sep=6em]
      \uparr\times \upars \ar[r,"{\truer\times \trues}"]
      \ar[d,"{\otimes_{|\uparr,\upars}}"] &
      \VV\times\VV \ar[d,"\otimes"] \\
      \uparrow (r\otimes s) \ar[r,"{\truers}"] & \VV
    \end{tikzcd}
  \]
  Now we put $z:=F(\otimes_{|\uparr,\upars})(w)$ and observe that
  \[
    \begin{aligned}
      F(p\otimes q)(t) & =       F(\otimes)\circ F(\truer\times \trues)(w) \\
      & = F(\truers)\circ F(\otimes_{|\uparr,\upars})(w)\\
      & = F(\truers)(z)\,.
    \end{aligned}
  \]
  We conclude that $F(p\otimes q)(t)\in F(\uparrs)$.

  \medskip

  2. Now let us prove the second item.  Given $t\in FX$, we know that
  \[
    \begin{aligned}
      \liftFcan\constPred{X}(t) & = \evcan(F\constPred{X}(t))\\
      & = \bigvee\{ r\mid F\constPred{X}(t)\in F(\uparr)\}\,.
    \end{aligned}
  \]
  In order to show that $\liftFcan\constPred{X}(t) \ge 1$ it is
  sufficient to prove that $F\constPred{X}(t) \in F(\upar{1})$.

  Let $e\colon X\to \upar{1}$ a constant mapping with $e(x) = 1$. Then
  the diagram to the left below commutes and by applying the functor
  $F$ we obtain the diagram below.
  \[
    \xymatrix{
      & X \ar[dl]_e \ar[d]^{\constPred{X}} \\
      \upar{1} \ar[r]^{\trueone} & \VV } \qquad \xymatrix{
      & FX \ar[dl]_{Fe} \ar[d]^{F\constPred{X}} \\
      F(\upar{1}) \ar[r]^{F\trueone} & F\VV }
  \]
  Now $F\constPred{X}(t)=F\trueone(Fe(t))$, hence
  $F\constPred{X}(t)\in F(\upar{1})$.
\end{proof}

\subsection{Proofs and additional material for
  Section~\ref{sec:up-to-techniques}}
\label{sec:proofs-up-to-techniques}

\subsubsection{Basic up-to techniques and chaining}
\label{app:basic}

In Section \ref{sec:up-to-techniques} we have introduced the technique
$\cbhv$ and mentioned in passing other basic techniques \kl{up-to
  reflexivity} $\cref$, \kl{up-to transitivity} $\ctrn$ and \kl{up-to
  symmetry} $\csym$. In this appendix we give precise definitions for
these techniques and show sufficient criteria ensuring their
soundness.

Inductively, take $(-)^0 = id\colon \fibreRel{X} \to \fibreRel{X}$ and
$(-)^{n+1}= id \comp (-)^n$.  Call $diag\colon \Vrel_X \to \Vrel_X $
the constant function to $\diagRel{X}$ and
$inv \colon \Vrel_X \to \Vrel_X $ be the inversion function mapping
$d$ into $d\circ \sym{X}$. Then $\cref$,$\ctrn$ and $\csym$ are
defined as follows.
\begin{equation*}
  \cref=id\vee diag \qquad\qquad\qquad
  \ctrn= \bigvee_i (-)^i \qquad\qquad\qquad \csym = id \vee inv
\end{equation*}

With these definitions and two results
in~\cite{bppr:general-coinduction-up-to}, it is immediate to prove the
following 
\begin{prop}
  Let $\tilde{F} \colon \Vrel \to \Vrel$ be a an arbitrary lifting of
  $F\colon \Set \to \Set$. Assume that
  $b=\reindexnobrk{\xi} \circ \tilde{F}$.
  \begin{itemize}
  \item If $\tilde{F} (\diagRel{X}) \geq \diagRel{FX}$, then $\cref$
    is \kl{$b$-compatible}.
  \item If $\tilde{F}(p\comp q) \geq \tilde{F}(p) \comp \tilde{F}(q)$
    for all $p,q\in \Vrel_X$, then $\ctrn$ is \kl{$b$-compatible}.
  \item If
    $\tilde{F}(d)\circ \sym{X} \leq \tilde{F}(d \circ \sym{X})$, then
    $\csym$ is \kl{$b$-compatible}.
  \item If $\tilde{F}$ is a fibred lifting, then $\cbhv$ is
    \kl{$b$-compatible}.
  \end{itemize}
\end{prop}
\begin{proof}
  Observe that
  \begin{itemize}
  \item $diag$ is compatible by the hypothesis and
    ~\cite[Proposition~6.3]{bppr:general-coinduction-up-to}. Then
    $\cref$ is compatible since $id$ is compatible and the join of
    compatible functions is compatible.
  \item For all $i$, $(-)^i$ is compatible (the proof goes by
    induction: for the base case, $id$ is compatible; for the
    inductive case, we use Proposition~\ref{prop:tns}). Then $\ctrn$
    is compatible (following the same argument as above).
  \item $inv$ is compatible by the hypothesis
    and~\cite[Proposition~6.3]{bppr:general-coinduction-up-to}. Then
    $\csym$ is compatible.
  \item Theorem 6.1 in~\cite{bppr:general-coinduction-up-to} entails
    that $\cbhv$ is compatible.\qedhere
  \end{itemize}
\end{proof}

Notice that whenever $\tilde{F}$ is the \kl{Wasserstein lifting}
corresponding to some $\Vpred$-lifting $\liftF$ which satisfies the
conditions of Theorem \ref{thm:restricting-Wasserstein}, then the
hypotheses of the above proposition immediately hold, so all the basic
up-to techniques are compatible.

\proptns*
\begin{proof}
  Follows immediately from Proposition 6.3
  in~\cite{bppr:general-coinduction-up-to}.
\end{proof}

The proof of Theorem~\ref{thm:Lifting-Distr-Law-Wass}, which we recall
next, follows immediately from
Lemmas~\ref{lem:tilde-sigma},~\ref{lem:Wass-is-Functorial}
and~\ref{lem:sigma-tilde}.

\subsubsection{Details on the contextual closure
  (Example~\ref{ex:ctx})}
\label{app:ctx}

Recall that the functor $T_\Sigma \colon \Set \to \Set$ maps a set $X$
into the set of $\Sigma$-terms with variables in $X$.  It carries the
structure of a monad: the unit $\eta_X \colon X \to T_\Sigma X$ maps
an elements into variables and the multiplication
$\mu_X \colon T_\Sigma T_\Sigma X \to T_\Sigma X$ is just term
composition.

Let $T_1,T_2\in T_\Sigma T_\Sigma X$ with
$T_1=C_1(s_1^1, \dots, s_n^1 )$ and $T_2= C_2(s_1^2, \dots, s_m^2)$.
The \kl[canonical Wasserstein lifting]{canonical lifting} of
$T_{\Sigma}$ is defined for all $d\in \fibreRel{T_{\Sigma}X}$ as
\begin{equation} \label{eq:ctx}
  \overline{T_{\Sigma}}_\mathsf{can}(d)(T_1,T_2)=
  \begin{cases}
    \bot               & C_1\neq C_2\\
    \bigwedge_j d(s_j^1,s_j^2)             & \text{otherwise}\\
  \end{cases}
\end{equation}
Now, for all $t_1, t_2\in T_{\Sigma}X$, we have that
\begin{align*}
  ctx(d)(t_1,t_2) & = \Sigma_{\mu_X} (\overline{T_{\Sigma}}_\mathsf{can}(d)) (t_1,t_2) \\
                  & =\bigvee \{ \overline{T_{\Sigma}}_\mathsf{can}(d) (T_1,T_2) \mid \mu_X(T_i)=t_i  \} \\
                  &= \bigvee_{C} \{ \bigwedge_j d(s_j^1,s_j^2)   \mid t_i=C(s_0^i, \dots, s_n^i)  \}\,.\\
\end{align*}

\subsubsection{Details on the convex closure (Example~\ref{ex:cvx})}
\label{app:cvx}

We show that the \kl{up-to convex closure} as defined
in~\cite{cgv:up-to-bisim-metrics} coincides with the one obtained as
the composite in~\eqref{eq:metric-f-fib}:
\[
  \begin{tikzcd}
    f\colon\fibreRel{\Distr
      X}\ar[r,"\overline{\Distr}"]&\fibreRel{\Distr\Distr
      X}\ar[r,"\directImage{\mu_X}"]&\fibreRel{\Distr X}\,,
  \end{tikzcd}
\]
by taking the \kl{Wasserstein lifting} $\overline{\Distr}$ of $\Distr$
from Example~\ref{ex:wasserstein} corresponding to the evaluation map
of Example~\ref{ex:evaluation-as-mult-distr-2}, and the free algebra
structure on $\Distr X$ given by the monad multiplication
$\mu_X\colon \Distr \Distr X \to \Distr X$.

Let $\VV=[0,1]$, $\Delta,\Theta\in\Distr X$ and
$d\colon \Distr X\times \Distr X\to [0,1]$. Then, by expanding the
definitions of the \kl{direct image} and of the \kl{Wasserstein
  lifting} we obtain:
\begin{eqnarray*}
  f(d)(\Delta,\Theta) 
  & = &
        \directImage{\mu_X}(\overline{\Distr}(d))(\Delta,\theta)
  \\
  & = & \inf\{\overline{\Distr}(d)(\tilde{\Delta},\tilde{\Theta}) \mid
        \tilde{\Delta},\tilde{\Theta} \in \Distr\Distr X,
        \mu_X(\tilde{\Delta}) = \Delta, \mu_X(\tilde{\Theta}) = \Theta \} \\
  & = & \inf\{\inf\{\widehat{\Distr}(d)(\Gamma) \mid \Gamma\in
        \Distr(\Distr X\times \Distr X),
        \mathcal{D}\pi_1(\Gamma)=\tilde{\Delta},
        \mathcal{D}\pi_2(\Gamma)=\tilde{\Theta}\}
        \mid \\
  && \qquad\qquad \tilde{\Delta},\tilde{\Theta} \in \Distr\Distr X,
     \mu_X(\tilde{\Delta}) = \Delta, \mu_X(\tilde{\Theta}) = \Theta \} \\
  & = & \inf\{\widehat{\Distr}(d)(\Gamma) \mid \Gamma\in \Distr(\Distr
        X\times \Distr X), \mu_X(\mathcal{D}\pi_1(\Gamma))=\Delta,
        \mu_X(\mathcal{D}\pi_2(\Gamma))=\Theta\}
\end{eqnarray*}
Now observe that $\Gamma\in \Distr(\Distr X\times \Distr X)$ can be
written as a formal sum $\Gamma = \sum_i p_i\cdot (\Delta_i,\Theta_i)$
where $\Delta_i,\Theta_i\in\Distr X$ and
$p_i = \Gamma(\Delta_i,\Theta_i)$. Then
\begin{eqnarray*}
  \widehat{\mathcal{D}}(d)(\Gamma) &=& \ev(\mathcal{D}(d)(\Gamma)) \\
                                   &=&\sum_{r\in[0,1]} r\cdot \sum_{\Gamma(\Delta',\Theta')=r}
                                       d(\Delta',\Theta') \\
                                   & = & \sum_{\Delta',\Theta'}
                                         d(\Delta',\Theta')\cdot \Gamma(\Delta',\Theta')\\ & = &  
                                                                                                 \sum_i p_i\cdot d(\Delta_i,\Theta_i) 
\end{eqnarray*}

In addition $\mu_X(\Distr\pi_1(\Gamma)) = \Delta$ means
$\sum_i p_i\cdot\Delta_i = \Delta$ and similarly
$\mu_X(\Distr\pi_2(\Gamma)) = \Theta$ means
$\sum_i p_i\cdot\Theta_i = \Theta$.  Therefore
$f(d)(\Delta,\Theta) = \ccvx(d)(\Delta,\Theta)$.

\subsubsection{Lifting of the distributive
  law}\label{app:liftingdistr}

\begin{lem}
  \label{lem:tilde-sigma}
  Whenever
  $\liftT\circ\directImage{\lambdaChBaseF{X}}
  \le\directImage{T\lambdaChBaseF{X}}\circ\liftT$, the identity
  natural transformation $TF\Rightarrow TF$ lifts to
  $\barT\circ\barF \Rightarrow \barTF$.
\end{lem}

\begin{proof}
  We have to show that the (identity) maps underlying the natural
  transformation are non-expansive, in particular
  $\barT\circ\barF(r) \le \barTF(r)$ for every $r$ in $\fibreRel{X}$.

  Let $\lambdaChBaseF{X}\colon F(X\times X)\to FX\times FX$,
  $\lambdaChBaseT{FX}\colon T(FX\times FX)\to TFX\times TFX$,
  $\lambdaChBaseTF{X}\colon TF(X\times X)\to TFX\times TFX$ be the
  natural transformations on which the three liftings are based. From
  the uniqueness of the mediating morphism of the product we obtain
  $\lambdaChBaseTF{X} = \lambdaChBaseT{FX}\circ
  T\lambdaChBaseF{X}$. This allows us to prove:
  \begin{eqnarray*}
    \barT\circ\barF & = & \isoFibre{TFX}^{-1}\circ
                          \directImage{\lambda^T_{FX}}\circ \liftT\circ \isoFibre{FX} \circ \isoFibre{FX}^{-1}\circ
                          \directImage{\lambda^F_X}\circ{\liftF} \circ \isoFibre{X} \\
                    & = & \isoFibre{TFX}^{-1}\circ
                          \directImage{\lambda^T_{FX}}\circ \liftT\circ
                          \directImage{\lambda^F_X}\circ{\liftF} \circ \isoFibre{X} \\
                    & \ge & \isoFibre{TFX}^{-1}\circ\directImage{\lambda^T_{FX}}\circ\directImage{T\lambda^F_X}\circ
                            \liftT\circ{\liftF} \circ \isoFibre{X} \\
                    & \ge & \isoFibre{TFX}^{-1}\circ \directImage{\lambda^{TF}_{FX}}\circ
                            \liftTF \circ \isoFibre{X} = \barTF\,.\qquad\qquad\qedhere
  \end{eqnarray*}
\end{proof}

\begin{lem}
  \label{lem:sigma-tilde}
  Using the notations in Theorem~\ref{thm:Lifting-Distr-Law-Wass}, the
  identity natural transformation $F\circ T\Rightarrow F\circ T$ lifts
  to $\barFT \Rightarrow \barF\circ \barT$.
\end{lem}

\begin{proof}
  It always holds that
  $\directImage{Ff}\circ {\liftF} \le {\liftF}\circ \directImage{f}$
  for all $f\colon X\to Y$. Indeed, by~\eqref{eq:canonical-ineq},
  ${\liftF}\circ f^* \le (Ff)^*\circ {\liftF}$ holds in $\Vrel$.
  Then, using the fact that $\directImage{f}$, $f^*$ (respectively
  $\directImage{Ff}, (Ff)^*$) are adjoint, we obtain the desired
  inequality.

  The rest of the proof is analogous to Lemma~\ref{lem:tilde-sigma}.
\end{proof}

We also include an alternative proof of
Theorem~\ref{thm:Lifting-Distr-Law-Wass} which brings the pieces of
proof in one large diagram.\todo{DP: we can remove this, though it
  helped me understand what is going on. BK: It's fine, we can also
  keep it.}

\begin{proof}[Alternative proof of Theorem~\ref{thm:Lifting-Distr-Law-Wass}]
  The existence of the lifting $\liftDistr$ is equivalent to
  \[\liftT\circ\liftF\le\reindex{\distrLaw_{X}}\circ\liftF\circ\liftT\,,\]
  while the existence of the lifting $\barDistr$ is equivalent
  to
  \[\barT\circ\barF\le\reindex{\distrLaw_{X}}\circ\barF\circ\barT\,.\]
  The latter inequality, which we have to prove, is in turn equivalent
  to the inequality obtained by using the isomorphism $\isoPR$.
  \[
    \begin{aligned}
      \isoFibre{TFX}\circ\barT\circ\barF\circ\isoFibre{X}^{-1}& \le \isoFibre{TFX}\circ\reindex{\distrLaw_{X}}\circ\barF\circ\barT\circ\isoFibre{X}^{-1}\\
      &=
      \reindex{\chbaseFunct\distrLaw_{X}}\circ\isoFibre{FTX}\circ\barF\circ\barT\circ\isoFibre{X}^{-1}
    \end{aligned}
  \]
  The left hand side of the above inequality rewrites using the
  definitions of the \kl{Wasserstein liftings} as the composite of the
  outermost right-then-down path $\fibrePred{\chbaseFunct X}$ to
  $\fibrePred{\chbaseFunct TFX}$ in the next diagram.  The right hand
  side similarly evaluates to the outermost down-then-right path in
  the diagram. So it suffices to establish the inequality between
  these two paths. We do this by decomposing the diagram into smaller
  pieces (see Figure~\ref{fig:existence-liftDistr}) and explaining
  each inequality in turn.

\begin{figure*}
  \centering $
  \begin{tikzcd}[row sep =3em]
    \fibrePred{\chbaseFunct X}\ar[rr,"\liftF"]\ar[d,swap,"\liftT"] & &
    \fibrePred{F\chbaseFunct
      X}\ar[r,"\directImage{\lambdaChBaseF{X}}"]\ar[d,"\liftT"]\ar[dll,color=white,"\textcolor{black}{\rotle}"
    description] & \fibrePred{\chbaseFunct
      FX}\ar[d,"\liftT"]\ar[dl,color=white,"\textcolor{black}{\rotle}"
    description]
    \\
    \fibrePred{T\chbaseFunct
      X}\ar[d,swap,"\directImage{\lambdaChBaseT{X}}"]\ar[r,"\liftF"] &
    \fibrePred{FT\chbaseFunct
      X}\ar[dr,"\directImage{\lambdaChBaseFT{X}}"
    description]\ar[d,"\directImage{F\lambdaChBaseT{X}}"]\ar[r,"\reindex{\distrLaw_{\chbaseFunct
        X}}"]\ar[dl,color=white,"\textcolor{black}{\rotle}"
    description] & \fibrePred{TF\chbaseFunct
      X}\ar[dr,"\directImage{\lambdaChBaseTF{X}}"
    description]\ar[r,"\directImage{T\lambdaChBaseF{X}}"]\ar[d,color=white,"\textcolor{black}{\rotle}"
    description] & \fibrePred{T\chbaseFunct
      FX}\ar[d,"\directImage{\lambdaChBaseT{FX}}"]
    \\
    \fibrePred{\chbaseFunct TX}\ar[r,swap,"\liftF"] &
    \fibrePred{F\chbaseFunct
      TX}\ar[r,swap,"\directImage{\lambdaChBaseF{TX}}"] &
    \fibrePred{\chbaseFunct
      FTX}\ar[r,swap,"\reindex{\chbaseFunct\distrLaw_{X}}"] &
    \fibrePred{\chbaseFunct TFX}
  \end{tikzcd}
  $
  \caption{Existence of the lifting
    $\liftDistr$ \label{fig:existence-liftDistr}}
\end{figure*}

The two inequalities in the top pentagon and top-right square follow
from the hypothesis. The two triangles at the bottom are equalities
that follow from the fact that
\[
  \lambdaChBaseFT{X}=\lambdaChBaseF{TX}\circ
  F\lambdaChBaseT{X}\textrm{\quad and\quad }
  \lambdaChBaseTF{X}=\lambdaChBaseT{FX}\circ T\lambdaChBaseF{X}
\]
The inequality in the left-down square holds since $\liftF$ is a
lifting and is obtained via adjoint transposes
from~\eqref{eq:canonical-ineq}.\todo{add this explicitly.}

Using the naturality of $\distrLaw$ one can show that the next square
commutes
\[
  \begin{tikzcd}
    TF\chbaseFunct X\ar[d,swap,"\distrLaw_{\chbaseFunct X}"]\ar[r,"\lambdaChBaseTF{X}"] & \chbaseFunct TF X\ar[d,"\chbaseFunct\distrLaw_{X}"] \\
    FT\chbaseFunct X\ar[r,"\lambdaChBaseFT{X}"] & \chbaseFunct FT X\,. \\
  \end{tikzcd}
\]
and hence, the inequality in the bottom rhombus can be derived as an
instance of a generic result for bifibrations,
see~\eqref{eq:mate-BCC}.
\end{proof}

We now turn to proving
Proposition~\ref{thm:lift-nat-transf-canonical}.

\thmLiftNatTransCan*

\begin{proof}
  We use the following notations:
  \[\ev_{TF}:=\evTcan\circ T(\evFcan)\] and
  \[\ev_{FT}:=\evFcan\circ F(\evTcan)\,.\]
  Notice that $\ev_{TF}$ and $\ev_{FT}$ are exactly the evluation maps
  corresponding to the liftings $\liftTcan\circ\liftFcan$,
  respectively $\liftFcan\circ\liftTcan$.  Using the second part of
  Lemma~\ref{lem:Wass-is-Functorial}, it suffices to show
  $\ev_{TF} \le \ev_{FT}\circ \zeta_\VV$.

  We will consider the inclusion maps $\truer\colon \uparr\to \VV$ and
  write $t\in F(\uparr)$ for $t\in (F\uparr)$.

  We first consider the diagram below. We will show that the dotted
  arrow exists and that the resulting square is a weak pullback.
  \[
    \xymatrix{ F(\uparr)\ \ \ar@{>->}[r]^{F\truer}
      \ar@{.>}[d]_{\evFcan|_{\uparr}} & F\VV
      \ar[d]^{\evFcan} \\
      \uparr \ \ \ar@{>->}[r]^{\truer} & \VV }
  \]
  Let $t\in F(\uparr)$. This means that
  $\evFcan(t) = \bigvee \{s\mid t\in F(\upars)\} \ge r$, since the set
  contains $r$ itself. Hence $\evFcan$ restricts to
  $\evFcan|_{\uparr}$.

  In order to show that the square is a weak pullback take
  $t'\in F\VV$ such that
  $\evFcan(t') = \bigvee\{s\mid t\in F(\upars)\} = \bar{s} \ge r$. We
  have to show that $t\in F(\uparr)$, i.e., that $r$ is contained in
  the set, which we will do by showing that $\{s\mid t\in F(\upars)\}$
  is downward-closed and contains its supremum. The set
  $\{s\mid t\in F(\upars)\}$ is downward-closed since $F$ as a
  $\Set$-functor preserves injections with non-empty domains and hence
  $s'\le s$ implies $\upars\subseteq \ \upars'$ and thus
  $t\in F(\upars) \subseteq F(\upars')$. If $F$ preserves
  intersections, it also contains its supremum:
  $\uparrow\!\bar{s} = \bigcap \{\upars\mid t\in F(\upars)\}$ and
  hence
  $F(\uparrow\!\bar{s}) = \bigcap \{F(\upars)\mid t\in F(\upars )\}
  \ni t$.
  
  Similarly one obtains such a commuting square (not necessarily a
  weak pullback) for $T$ and $\evTcan$. This results in the following
  diagram where the right-hand square and the upper ``square'' commute
  and the left-hand square is a weak pullback (using pullback
  preservation of $T$).
  \[
    \xymatrix{ TF(\uparr)\ \ \ar@{>->}[r]^{TF\truer}
      \ar@{.>}[d]_{T(\evFcan|_{\uparr})}
      \ar@/^2pc/[rrr]^{\zeta_{\uparr}} & TF\VV \ar[r]^{\zeta_\VV}
      \ar[d]^{T\evFcan} & FT\VV \ar[d]_{F\evTcan} & \ \ FT(\uparr)
      \ar@{>->}[l]_{FT\truer} \ar@{.>}[d]^{F(\evTcan|_{\uparr})} \\
      T(\uparr) \ \ \ar@{>->}[r]^{T\truer} & T\VV & F\VV & \
      F(\uparrow r) \ar@{>->}[l]_{F\truer} }
  \]
  In order to prove that
  \[\ev_{TF} = \evTcan\circ T\evFcan \le \evFcan\circ F\evTcan\circ
    \zeta_\VV = \ev_{FT}\circ\zeta_\VV\,,\] let $t\in TF\VV$. Since
  \[\evTcan(T\evFcan(t)) = \bigvee \{r\mid T\evFcan(t)\in
    T(\uparr)\}\] and
  \[\evFcan(F\evTcan(\zeta_\VV(t))) = \bigvee \{r\mid
    F\evTcan(\zeta_\VV(t))\in F(\uparr)\}\] it suffices to show that
  \[ T\evFcan(t)\in T(\uparr)\textrm{ implies }
    F\evTcan(\zeta_\VV(t))\in F(\uparr)\,.\]

  So let $T\evFcan(t)\in T(\uparr)$ and the fact that the left-hand
  square is a weak pullback implies that there exists
  $t'\in TF(\uparr)$ with $TF\truer(t') = t$.

  Then, using naturality of $\zeta$, we obtain
  \[
    \begin{aligned}
      F\evTcan(\zeta_\VV(t))&  = F\evTcan(\zeta_\VV(TF\truer(t')))\\
      & =       F\evTcan(FT\truer(\zeta_{\uparr}(t')))\\
      & = F\truer(F(\evTcan|_{\uparr})(\zeta_{\uparr}(t')))\\
      & \in F(\uparr)\qquad\qquad\qquad\qquad\qquad\qquad \qedhere
    \end{aligned}
  \]
\end{proof}

\propccdq*
\begin{proof}
  In Appendix~\ref{app:ccdq}, we prove a more general result
  (Proposition~\ref{prop:ccdqgeneral}). Proposition~\ref{prop:ccdq}
  follows thus from Proposition \ref{prop:ccdqgeneral}, whose
  conditions are shown to be satisfied by the canonical lifting in
  Lemmas~\ref{lemma:ccdq1} and~\ref{lemma:ccdq2}.
\end{proof}

\subsubsection{Details on constructively completely distributive
  quantales}\label{app:ccdq}
In this appendix, we provide a result
(Proposition~\ref{prop:ccdqgeneral} below) for proving
$\liftT\circ \directImage{f} \le \directImage{Tf}\circ \liftT$ that is
more general than Proposition~\ref{prop:ccdq}. This is useful, for
instance to prove such property for liftings different than the
canonical one.

We assume that the quantale $\VV$ is \kl{constructively completely
  distributive} and we start with two examples of such quantales, in
order to give some more intuition.

\begin{ex}
  In the reals the order $\VVlll$ coincides with $>_\mathbb{R}$,
  whereas in a powerset lattice $\Pow M$ we have that $M_1\VVlll M_2$
  for $M_1,M_2\subseteq M$ whenever $M_1\subseteq M_2$ and $M_1$
  contains at most one element. Both lattices are \kl{constructively
    completely distributive}.
\end{ex}

For this more general result, we need some additional properties, in
particular the lifting $\liftT$ must preserve a special type of
supremum of predicates (even stronger than uniform convergence).

\begin{defi}
  Let $(p_i\colon X\to\VV)_{i\in I}$ be a family of predicates. We
  call its sup \intro{constructively-convergent} \todo{Find a
    different name for \emph{constructively-convergent}.}  if for
  every predicate $q\colon X\to \VV$ with
  $q\VVlll \bigvee_{i\in I} p_i$ (pointwise), there exists $i\in I$
  with $q\le p_i$.
\end{defi}

\begin{prop}\label{prop:ccdqgeneral}
  Assume $\VV$ is a \kl{constructively completely distributive}
  quantale and assume $\liftT$ is a lifting a $\Set$-functor $T$. Then
  we have that
  $\liftT\circ \directImage{f} \le \directImage{Tf}\circ \liftT$
  whenever either of the conditions below is met
  \begin{itemize}
  \item $f$ is surjective and $\liftT$ preserves
    \kl{constructively-convergent} sups.
  \item $f$ is injective, $T$ preserves weak pullbacks, $\liftT$ is a
    \kl{fibred lifting} corresonding to an \kl{evaluation map} $\ev$
    such that for every $t\in T\VV$, $\ev(t)\neq \bot$ implies
    $t\in T(\VV\backslash\{\bot\})$. (In other words:
    $\ev^{-1}(\VV\backslash \{\bot\})\subseteq
    T(\VV\backslash\{\bot\})$.)
  \item $f$ is an arbitrary function and all the above properties are
    satisfied.
  \end{itemize}
\end{prop}

\begin{proof}
  Let $f = m\circ e$ be the epi-mono factorization of $f$, i.e.,
  $e\colon X\to Z$ is surjective and $m\colon Z\to Y$ is injective. We
  will show the inequality separately for $m$, $e$, from which we can
  straightforwardly derive the inequality for $f$.

  \begin{itemize}
  \item
    $\liftT\circ \directImage{e} \le \directImage{Te}\circ \liftT$:

    Let $p\colon X\to\VV$, $z\in Z$. Observe that
    \[ \directImage{e}(p)(z) = \bigvee \{p(x)\mid e(x) = z\} = \bigvee
      \{(p\circ g)(z) \mid g\in G\}
    \]
    where $G = \{g\colon Z\to X\mid e\circ g = \id_Z\}$ is the set of
    all choice functions. Note that the last equality in the displayed
    equation above requires surjectivity of $e$, because otherwise no
    choice functions exist.

    So
    $\directImage{e}(p) = \bigvee_{g\in G} p\circ g = \bigvee_{g\in G}
    \reindexg(p)$ and we show that this sup is
    \kl{constructively-convergent}. Let $q\colon Z\to \VV$ with
    $q\VVlll \directImage{e}(p)$. Now for a $z\in Z$ we observe that
    $q(z) \VVlll \bigvee_{e(x)=z} p(x)$ and hence, since we are
    working in a \kl{cccd-lattice}, there exists an $x_z\in X$ with
    $e(x_z)=z$ and $p(x_z)\ge q(z)$. On $z$ we define the choice
    function $g$ as $g(z) = x_z$. \todo{This argument requires the
      axiom of choice.} We have $e\circ g = \id_Z$ and furthermore for
    all $z\in Z$ we have $q(z) \le p(x_z) = (p\circ g)(z)$, hence
    $q\le p\circ g$ as desired.

    According to our assumption $\liftT$ preserves such suprema and we
    get:
    \[
      \liftT(\directImage{e}(p)) = \liftT(\bigvee_{g\in G}
      \reindexg(p)) = \bigvee_{g\in G}\liftT(\reindexg(p))\le
      \bigvee_{g\in G} \reindex{Tg}(\liftT p)
    \]
    We will now show $\reindex{Tg} \le \directImage{Te}$ as an
    intermediate result: Let $p\colon TX\to\VV$ and $t\in TZ$. Then
    \begin{eqnarray*}
      && \reindex{Tg}(p)(t) = (p\circ Tg)(t) = p(Tg(t)) \le
         \bigvee_{Te(s)=t} p(s) = \directImage{Te}(p)(t)
    \end{eqnarray*}
    since $s=Tg(t)$ satisfies
    $Te(s) = Te(Tg(t)) = T\mathit{id}_Z(t) = t$.  This implies
    \[ \bigvee_{g\in G} \reindex{Tg}(\liftT p) \le \bigvee_{g\in G}
      \directImage{Te}(\liftT p) = (\directImage{Te}\circ
      \liftT)(p) \] By combining everything we obtain the desired
    result.
  \item
    $\liftT\circ \directImage{m} \le \directImage{Tm}\circ \liftT:$

    Let $p\colon Z\to\VV$, $t\in FY$, we have to show that
    $\liftT(\directImage{m}(p))(t) \le \directImage{Tm}(\liftT p)(t)$.

    We consider the following two cases:
    \begin{itemize}
    \item $t$ is in the image of $Tm$: in this case there exists
      $s\in TX$ with $Tm(s) = t$.

      Since $m$ is injective we have that for any $y\in Y$ in the
      image of $m$ $\directImage{m}(p)(y) = p(z)$, where $z\in Z$ is
      the unique preimage of $y$. Hence
      $\directImage{m}(p)\circ m = p$.
      
      Using the fact that we have a fibred lifting
      (Proposition~\ref{prop:lift-correspondences}), this means that
      \begin{eqnarray*}
        && \liftT(\directImage{m}(p))(t) =
           \liftT(\directImage{m}(p))(Tm(s)) =
           \liftT(\directImage{m}(p)\circ m)(s) = \liftT p(s) \\
        & \le &
                \bigvee_{Tm(s)=t} \liftT p(s) = \directImage{Tm}(\liftT p))(t)
      \end{eqnarray*}
    \item $t$ is not in the image of $Tm$: we show that in this case
      $\liftT(\directImage{m}(p))(t) = \bot$. (The right-hand side of
      the inequality is also $\bot$, due to the empty supremum.) Note
      that $\directImage{m}(p)(y) = \bot$ for all $y\in Y$ which are
      not in the image of $m$.

      Now assume that $\liftT(\directImage{m}(p))(t) \neq \bot$. Take
      the pullback on the left below and observe that
      $Y' = \{y\in Y\mid \directImage{m}(p)(y) \neq \bot\} \subseteq
      m[X]$.
      \[
        \xymatrix{
          \VV\backslash\{\bot\} \ar@{>->}[r] & \VV \\
          Y' \ar[u] \ar@{>->}[r]_m & Y \ar[u]^{\directImage{m}(p)} }
        \qquad \xymatrix{ T(\VV\backslash\{\bot\}) \ar@{>->}[r] & T\VV
          \ar[r]^{\ev}
          & \VV \\
          T(Y') \ar[u] \ar@{>->}[r]_{Tm} & TY
          \ar[u]^{T(\directImage{m}(p))}
          \ar[ur]_{\liftT(\directImage{m}(p))} & }
      \]
      Since $T$ is weak pullback preserving, the square on the right
      above is a weak pullback.

      By assumption we have
      $\ev(T(\directImage{m}(p))(t)) = \liftT(\directImage{m}(p))(t)
      \neq \bot$. This implies that
      $T(\directImage{m}(p))(t) \in T(\VV\backslash\{\bot\})$.

      Since the square to the right above is a weak pullback, this
      means that $t\in T(Y')$, hence $t\in T(m[X])$, which is a
      contradiction, since $t$ is not in the image of $Tm$. \qedhere
    \end{itemize}
  \end{itemize}
\end{proof}

\begin{lem}\label{lemma:ccdq1}
  Assume that $T$ preserves weak pullbacks and $\VV$ is a
  \kl{constructively completely distributive} quantale. Then the
  canonical predicate lifting $\liftTcan$ preserves
  \kl{constructively-convergent} sups, i.e.,
  $\liftTcan(\bigvee_{i\in I} p_i) = \bigvee_{i\in I} \liftTcan(p_i)$.
\end{lem}

\begin{proof}
  First, we obviously have
  $\liftTcan(\bigvee_{i\in I} p_i) \ge \bigvee_{i\in I}
  \liftTcan(p_i)$ due to monotonicity.

  We now show
  $\liftTcan(\bigvee_{i\in I} p_i) \le \bigvee_{i\in I}
  \liftTcan(p_i)$: first denote $\bigvee_{i\in I} p_i$ by $p$. Let
  $t\in TX$, hence the inequality spells out to
  \[ \bigvee \{r\mid Tp (t) \in T(\uparr)\} \le \bigvee_{i\in
      I}\bigvee\{s\mid Tp_i(t)\in T(\upars)\} \] Now let $r$ be such
  that $Tp(t)\in T(\uparr)$ and take the pullback on the left
  below. Note that $X_r = \{x\in X\mid p(x)\ge r\}$.
  \[
    \begin{tikzcd}
      \uparr \ar[r,rightarrowtail] & \VV  & &       T(\uparr) \ar[r,rightarrowtail] & T\VV  \\
      {X_r} \ar[u]\ar[r,rightarrowtail] &X \ar[u,swap,"p"] & & T(X_r)
      \ar[u] \ar[r,rightarrowtail] & TX \ar[u,swap,"Tp"]
    \end{tikzcd}
  \]
  Due to weak pullback preservation the square above on the right is a
  weak pullback. This means that $t\in TX$, which satisfies
  $t\in T(\uparr)$, is also contained in $T(X_r)$.

  Now let $u\VVlll r$. We define a predicate $p'\colon X\to\VV$ with
  \[ p'(x) = \left\{
      \begin{array}{ll}
        u & \mbox{if $x\in X_r$} \\
        \bot & \mbox{otherwise} 
      \end{array}
    \right.
  \]
  Note that $p'$ satisfies $p'\VVlll p$. This is also true whenever
  $p'(x) = \bot$, since $\bot$ is totally below everything. Then, due
  to constructively-convergence of the sup there exists an index
  $i\in I$ with $p'\le p_i$.

  Now obtain the set $X_u^i$ with the following pullback on the left,
  where $X_u^i = \{x\in X\mid p_i(x) \ge u\}$.
  \[
    \begin{tikzcd}
      \uparu \ar[r,rightarrowtail] & \VV & & T(\uparu) \ar[r,rightarrowtail] & T\VV  \\
      X_u^i \ar[u] \ar[r,rightarrowtail] & X \ar[u,swap,"{p_i}"] & &
      T(X_u^i) \ar[u] \ar[r,rightarrowtail] & TX \ar[u,swap,"{Tp_i}"]
    \end{tikzcd}
  \]
  We can observe that $X_r\subseteq X_u^i$: let $x\in X_r$, then
  $u = p'(x) \le p_i(x)$, hence $x\in X_u^i$.

  This means that $t\in T(X_u^i)$ and since the square on the right
  above commutes, this gives us $Tp_i(t)\in T(\uparu)$. From this we
  infer
  \[ \bigvee_{i\in I}\bigvee\{s\mid Tp_i(t)\in T(\upars)\} \ge u \]
  Since this holds for all $u\VVlll r$, we have
  \[ r = \bigvee_{u\VVlll r} u \le \bigvee_{i\in I}\bigvee\{s\mid
    Tp_i(t)\in T(\upars)\} \] which entails the required inequality.
\end{proof}

\begin{lem}\label{lemma:ccdq2}
  For a \kl{constructively completely distributive} quantale, it holds
  for the canonical predicate lifting $\liftTcan$ that
  $\evcan(t)\neq \bot$ implies $t\in T(\VV\backslash\{\bot\})$ for
  $t\in T\VV$.
\end{lem}

\begin{proof}
  Assume that
  \[ \evcan(t) = \bigvee \{ r\mid t\in T(\uparr)\} \neq \bot \] Hence,
  there is at least one $r\neq \bot$ with
  $t\in T(\uparr) \subseteq T(\VV\backslash\{\bot\})$, otherwise the
  supremum would equal $\bot$.
\end{proof}

Whenever $T = \mathcal{D}$ (the probability distribution functor), we
have that $\mathcal{D}$ preserves \kl{constructively-convergent} sups
(since it preserves uniform convergence). We now consider the
\kl{evaluation map} as in Example~\ref{ex:evaluation-as-mult-distr-2},
namely the expectation $\ev(t) = \sum_{r\in [0,1]} r\cdot t(r)$ with
$t\in \mathcal{D}[0,1]$. Note that in the quantale $[0,1]$ we have
$\bot = 1$.

However the property
$\ev(t)\neq 1\ \Rightarrow\ t\in \mathcal{D}[0,1)$ for all
$t\in \mathcal{D}[0,1]$ does not hold in this case. If $t(0) = 1$,
$t(r) = 0$ for all $r\neq 0$, we have $\ev(t) = 0 \neq 1$, but
$t\not\in \mathcal{D}[0,1)$.

Nevertheless, the corresponding $\Vpred$-lifting of $\Distr$ still
satisfies the property
$\widehat{\Distr}\circ \directImage{\lambdaChBase} \le
\directImage{\Distr \lambdaChBase}\circ \widehat{\Distr}$, required in
Theorem~\ref{thm:Lifting-Distr-Law-Wass}. Here we can rely on the fact
that for $\mathcal{D}$ the components
$\lambdaChBase_X\colon \mathcal{D}(X\times X)\to \mathcal{D}X\times
\mathcal{D}X$ are \emph{surjective}, so we can apply the second item
of Proposition~\ref{prop:ccdqgeneral}.

\subsubsection{Non-expansiveness of contexts}

\cornonex*
\begin{proof} First observe that by definition of $ctx$ (Example
  \ref{ex:ctx}), we have that
  \begin{equation*}\nu b (t_1,t_2) \leq ctx(\nu b)
    (C(t_1),C(t_2))\,. \end{equation*}
  Moreover, since $ctx$ is $b$-compatible, by \eqref{eq:inclusion}, it holds that 
  \begin{equation*}ctx(\nu b) (C(t_1),C(t_2))\leq (\nu b)
    (C(t_1),C(t_2))\,. \end{equation*}
\end{proof}

\subsection{Proofs and additional material for
  Section~\ref{sec:examples}}
\label{sec:proofs-examples}

\begin{lem}
  \label{lem:well-behaved}
  The evaluation maps $\evNFAF$ and $\evNFAT$ defined in
  Section~\ref{sec:examples} induce liftings which satisfy the
  requirements of Theorem~\ref{thm:restricting-Wasserstein}.
\end{lem}

\begin{proof}
  These evaluation maps satisfy the required properties: $\emph{ev}_T$
  is the canonical evaluation map (see
  Section~\ref{sec:wasserstein-lifting}), thus the statement follows
  from Proposition~\ref{prop:can-lifting-is-well-behaved}.  For
  $\emph{ev}_F$, notice that this is of the form $g\circ \ev$, where
  $\ev$ satisfies the requirements of
  Proposition~\ref{propCharEvMapWellBeh} (since it is canonical) and
  $g\colon \VV\to \VV$ with $g(r) = c\cdot r$ is monotone,
  $g(a\otimes b) \ge g(a)\otimes g(b)$ and $g(1)\ge 1$. It is thus
  straightforward to see that $\emph{ev}_F$ fulfils the conditions of
  Proposition~\ref{propCharEvMapWellBeh} and thus, the corresponding
  lifting, those of Theorem~\ref{thm:restricting-Wasserstein}.
\end{proof}

The next lemma establishes the fact that the distributive law
considered in the example in Section~\ref{sec:examples} satisfies the
hypothesis of Theorem~\ref{thm:Lifting-Distr-Law-Wass}.

\begin{lem}
  \label{lem:Lifting-of-Distributive-Law-NFA}
  Assume $\liftF$ and $\liftT$ are the $[0,1]$-$\Pred$ liftings of
  $FX=2\times X^A$ and $TX=\Pow X$ which correspond to the evaluation
  maps $\evNFAF$ and $\evNFAT$ defined in the example in
  Section~\ref{sec:examples}.  Then we have:
  \begin{enumerate}
  \item
    $\liftT\circ\directImage{\lambdaChBaseF{X}}
    \le\directImage{T\lambdaChBaseF{X}}\circ\liftT$, and,
  \item $\distrLawNFA\colon T\circ F\To F\circ T$ lifts to a natural
    transformation
    $\liftDistr\colon\liftT\circ\liftF\To\liftF\circ\liftT$
  \end{enumerate}
\end{lem}

\begin{proof}
  Recall that on the quantale $[0,1]$ the quantale order is the
  reversed order on the reals, so in order to avoid confusion we use
  $\le, \lor,\land$ in the quantale and $\geR,\inf,\sup$ in the reals.

  \noindent\emph{1.} To prove the first item, we can rely on
  Proposition~\ref{prop:ccdq}, since $\liftTcan$ is the canonical
  lifting and we are working in the quantale $\VV=[0,1]$, which is
  \kl{constructively completely distributive}.

  \noindent\emph{2.} Recall that $\liftT\circ\liftF$ is a lifting of $T\circ F$
  which corresponds to the evaluation map
  $\ev_{TF}=\evNFAT\circ T(\evNFAF)$. Similarly, $\liftF\circ\liftT$
  corresponds to the evaluation map
  $\ev_{FT}=\evNFAF\circ F(\evNFAT)$.  The existence of $\liftDistr$
  is then equivalent to the inequality
  \begin{equation}
    \label{eq:lemmaD2}
    \ev_{TF}\geR \ev_{FT}\circ \zeta_\VV\,.    
  \end{equation}
  Here we are almost in the setting of canonical liftings treated in
  Proposition~\ref{thm:lift-nat-transf-canonical}, apart from the fact
  that $\evNFAF=g\circ\evFcan$, where the function $g$ is given by
  $g(r) = c\cdot r$. Recall $\evNFAT=\evTcan$. Furthermore $T$
  preserves weak pullbacks and $F$ preserves intersections, hence by
  (the proof of) Proposition~\ref{thm:lift-nat-transf-canonical}, we
  know that
  \[
    \evTcan\circ T(\evFcan)\geR \evFcan\circ F(\evTcan)\circ
    \zeta_\VV\,.
  \]
  In order to obtain the desired lifting of natural transformations,
  we first notice that $\evTcan\circ Tg = g\circ \evTcan$. Indeed, for
  all $R\subseteq [0,1]$ we have
  \[ \evTcan(Tg(R)) = \sup c\cdot R = c \cdot \sup R =
    g(\evTcan(R))\,.\] To conclude, we use the above equalities and
  the monotonicity of $g$:
  \[
    \begin{aligned}
      \ev_{TF} & = \evNFAT\circ T(\evNFAF)\\
      & = \evTcan\circ T(g)\circ T(\evNFAF)\\
      & \geR g\circ \evTcan \circ T(\evNFAF)\\
      & \geR g\circ \evFcan\circ F(\evTcan)\circ \zeta_\VV\\
      & = \evNFAF\circ F(\evTcan)\circ \zeta_\VV\\
      & =\ev_{FT}\circ \zeta_\VV\,.
    \end{aligned}
  \]
\end{proof}
\end{document}